\documentclass[letterpaper,twocolumn,10pt]{article}
\usepackage{usenix}
\usepackage{tikz}

\usepackage{xcolor}
\usepackage{amssymb,amsmath}
\usepackage{bm}
\usepackage{amsthm}
\usepackage{graphicx}    
\usepackage{subcaption}
\usepackage{multirow}
\usepackage{algorithm}
\usepackage{algorithmic}
\usepackage{subcaption}
\usepackage{booktabs}
\usepackage{siunitx}

\newtheorem{theorem}{Theorem}
\newtheorem{definition}{Definition}

\newtheorem{proposition}{Proposition}
\newtheorem{lemma}{Lemma}
\newtheorem{assumption}{Assumption}
\newtheorem{corollary}{Corollary}
\newtheorem{remark}{Remark}
\definecolor{byzantine}{rgb}{0.74, 0.2, 0.64}
\definecolor{americanrose}{rgb}{1.0, 0.01, 0.24}
\definecolor{ao(english)}{rgb}{0.0, 0.5, 0.0}

\begin{document}

\date{}


\title{\Large \bf Residual-PAC Privacy: Automatic Privacy Control Beyond the Gaussian Barrier}

\author{
{\rm Tao Zhang}\\
Washington University in St. Louis
\and
{\rm Yevgeniy Vorobeychik}\\
Washington University in St. Louis
} 

\maketitle

\begin{abstract}

The Probably Approximately Correct (PAC) Privacy framework \cite{xiao2023pac} provides a powerful instance-based methodology to preserve privacy in complex data-driven systems. 
Existing PAC Privacy algorithms (we call them Auto-PAC) rely on a Gaussian mutual information upper bound.
However, we show that the upper bound obtained by Auto-PAC is tight if and only if under the data distribution, the unperturbed output is Gaussian and the noise is independent Gaussian.
We propose two approaches for addressing this issue.
First, we introduce two tractable post‐processing methods for Auto-PAC, based on Donsker–Varadhan representation and sliced Wasserstein distances.
However, the result still leaves "wasted" privacy budget.
To address this issue more fundamentally, 
we introduce Residual-PAC (R-PAC) Privacy, an $f$-divergence-based measure to quantify privacy that remains after adversarial inference. 
To implement R-PAC Privacy in practice, we propose a Stackelberg Residual-PAC (SR-PAC) automatic privatization algorithm, a game-theoretic framework that selects optimal noise distributions through convex bilevel optimization. 
Our approach achieves efficient privacy budget utilization for arbitrary data distributions and naturally composes when multiple mechanisms access the dataset.
Our experiments demonstrate that SR-PAC obtains consistently a better privacy-utility tradeoff than both PAC and differential privacy baselines.

\end{abstract}

\section{Introduction}

Data-driven decision systems power critical applications ranging from medical diagnosis to autonomous vehicles, yet their outputs can inadvertently expose sensitive information contained in the data. As data pipelines grow in scale and complexity, practitioners need rigorous and scalable privacy guarantees that go beyond ad-hoc testing.
Over the past two decades, formal privacy frameworks have proliferated. 
Differential Privacy (DP) \cite{dwork2006differential} (and its variants such as R\'enyi DP \cite{mironov2017renyi}) delivers input-independent worst-case indistinguishability by bounding output distribution shifts from single-record changes. 
Alternative information-theoretic definitions, such as mutual-information DP \cite{cuff2016differential}, Fisher-information bounds \cite{farokhi2017fisher,hannun2021measuring,guo2022bounding}, and Maximal Leakage \cite{issa2019operational,saeidian2023pointwise}, provide complementary guarantees and offer alternative trade-offs between privacy and utility.

Nevertheless, provable privacy guarantees for modern data‐processing algorithms remains a challenge. First, worst‐case frameworks like DP require computing global sensitivity, which is generally NP-hard \cite{xiao2008output}.
Moreover, computing the optimal privacy bound of DP under composition is, in general, \#P-complete \cite{murtagh2015complexity}.
In practice, finding the minimal noise needed to meet a target guarantee is intractable for most real-world algorithms, especially when the effect of each operation on privacy is unclear. 
On the other hand, empirical or simulation-based methods (e.g., testing resistance to membership inference \cite{shokri2017membership}) address specific threats but lack rigorous, adversary-agnostic assurance. Bridging this gap requires a new, broadly applicable framework that can quantify and enforce privacy risk without relying on sensitivity.

A promising alternative has recently emerged: the Probably Approximately Correct (PAC) Privacy framework \cite{xiao2023pac}. 
PAC Privacy shifts from indistinguishability-based guarantees to an operational notion that measures the \textit{information-theoretic hardness} of reconstructing sensitive information. 
It is defined by an impossibility-of-inference guarantee for a chosen adversarial task and data prior, and the framework provides algorithms that enforce tractable mutual-information upper bounds to certify this guarantee.
This approach enables automatic privatization via black-box simulation, 
and enjoys additive composition bounds and automatic privacy budget implementations for adaptive sequential compositions of mechanisms with arbitrary interdependencies.
Notably, PAC Privacy often requires only $O(1)$ noise magnitude to achieve its privacy guarantees (independent of the output dimension), whereas differential privacy's worst‐case, input‐independent noise magnitude scales as $\Theta(\sqrt{d})$ for a $d$-dimensional release.

However, existing PAC privacy algorithms (which we refer to as Auto-PAC) are fundamentally conservative.
In particular, we show (Proposition~\ref{prop:MI_gap}) that Auto-PAC achieves the designated privacy budget exactly if and only if under the data distribution, the unperturbed output is Gaussian and the noise is independent Gaussian, so that the unperturbed and perturbed outputs are jointly Gaussian.
Consequently, Auto-PAC will in general make inefficient use of the privacy budget.
Conservative privacy accounting is a central practical concern in DP and PAC Privacy because conservative bounds impose unnecessary noise and waste privacy budget, particularly under composition.
Narrowing this conservativeness remains an open challenge in PAC privacy~\cite{xiao2025pac}.

We address this limitation of Auto-PAC in two ways.
First, working within the general PAC Privacy framework, we develop two tractable post‐processing methods for Auto-PAC conservativeness reduction, based on Donsker–Varadhan representation \cite{donsker1975asymptotic}  and sliced Wasserstein distances \cite{rabin2011wasserstein,bonneel2015sliced}.
However, even these methods fail to fully close the privacy budget gap.
To address this issue more fundamentally, we introduce the notion of \textit{Residual-PAC Privacy} (R-PAC privacy).
Unlike PAC privacy, which aims to quantify and bound the privacy leaked, R-PAC privacy focuses instead on quantifying \textit{privacy remaining after information has been leaked by a data processing mechanism}, using $f$-divergence to this end.
When $f$-divergence is instantiated as Kullback–Leibler (KL) divergence, we show that Residual-PAC Privacy is fully characterized by the conditional entropy up to a known constant that does not depend on the mechanism or the applied noise.

To implement R-PAC privacy with KL divergence, we propose a novel \textit{Stackelberg Residual-PAC (SR-PAC)} automatic  privatization algorithm.
SR-PAC formulates the problem of privatization via noise perturbation, given a privacy budget, as a Stackelberg game in which the leader selects a noise distribution with the goal of minimizing the magnitude of the perturbation, while the follower chooses a stochastic inference strategy to recover the sensitive data.
We show that when the entire probability space is considered, the resulting bilevel optimization problem becomes a convex program. 
Moreover, we prove that the mixed-strategy Stackelberg equilibrium of this game yields the optimal noise distribution, ensuring that the conditional entropy of the perturbed mechanism precisely attains the specified privacy budget.
Finally, our experimental evaluation demonstrates that the proposed SR-PAC privacy framework consistently outperforms both PAC-privacy and differential privacy baselines.

A complete appendix, including all proofs, is provided in the online extended version of this paper \cite{zhang2025breaking}.
We summarize our main contributions as follows:
\begin{itemize}
    \item We characterize the conservativeness of Auto-PAC \cite{xiao2023pac,sridhar2024pac}, showing that it arises from the gap between the surrogate Gaussian mutual information bound and the true non-Gaussian mutual information of the privatized mechanism.

    \item We propose two computationally tractable approaches to reduce this gap: one based on the Donsker-Varadham representation (Theorem \ref{thm:DV}) and the other based on the sliced Wasserstein distances (Theorem \ref{thm:SWD}). 

    \item We propose a novel privacy framework, Residual-PAC (R-PAC), to quantify the portion of privacy that remains rather than the amount leaked. This offers a complementary perspective to PAC privacy, and enables efficient implementation of tight privacy budget.

    \item We present an automatic privatization algorithm, Stackelberg R-PAC (SR-PAC), to efficiently compute noise distributions for a given privacy budget. SR-PAC algorithm achieves tight budget utilization, can operate with only black-box access via Monte Carlo simulation, and adaptively concentrates noise in privacy-sensitive directions while preserving task-relevant information.
\end{itemize}

\subsection{Related Work}

\textbf{Privacy Quantification Notions. }
Differential privacy (DP) and its variants have become the gold standard for formal privacy quantification and guarantees, with the original definitions by Dwork et al.~\cite{dwork2006calibrating,dwork2006differential} formalizing privacy loss through bounds on the distinguishability of outputs under neighboring datasets. Variants such as concentrated differential privacy (CDP)~\cite{bun2016concentrated,dwork2016concentrated}, zero-concentrated DP (zCDP)~\cite{bun2018composable}, and R\'enyi differential privacy (RDP)~\cite{mironov2017renyi} have further extended this framework by parameterizing privacy loss with different statistical divergences (e.g., R\'enyi divergence), thereby enhancing flexibility in privacy accounting, especially for compositions and adaptive mechanisms.
Pufferfish privacy \cite{kifer2014pufferfish,Song2017,pierquin2024renyi,zhang2025differential,zhang2025sliced} generalizes DP by considering secrets that go beyond DP's presence and absence of individual records.
Information-theoretic measures provide alternative and complementary approaches for quantifying privacy loss. For instance, mutual information has been used to analyze privacy leakage in a variety of settings~\cite{chatzikokolakis2010statistical,cuff2016differential}, with $f$-divergence and Fisher information offering finer-grained or context-specific metrics~\cite{xiao2023pac,farokhi2017fisher,hannun2021measuring,guo2022bounding}. These frameworks help to bridge the gap between statistical risk and adversarial inference, and are closely connected to privacy-utility trade-offs in mechanism design. 
Maximal Leakage \cite{issa2019operational}, hypothesis testing interpretations \cite{balle2020hypothesis}, and other relaxations further broaden the analytic toolkit for measuring privacy risk.

\textbf{Privacy-Utility Trade-off. }
Balancing the trade-off between privacy and utility is a central challenge in the design of privacy-preserving mechanisms. This challenge is frequently formulated as an optimization problem~\cite{lebanon2009beyond,sankar2013utility,lopuhaa2024mechanisms,ghosh2009universally,gupte2010universally,geng2020tight,du2012privacy,alghamdi2022cactus,goseling2022robust}. For example, Ghosh et al.~\cite{ghosh2009universally} demonstrated that the geometric mechanism is universally optimal for DP under certain loss-minimizing criteria in Bayesian settings, while Lebanon et al.~\cite{lebanon2009beyond} and Alghamdi et al.~\cite{alghamdi2022cactus} studied utility-constrained optimization. Gupte et al.~\cite{gupte2010universally} modeled the privacy-utility trade-off as a zero-sum game between privacy mechanism designers and adversaries, illustrating the interplay between optimal privacy protection and worst-case adversarial loss minimization.

\textbf{Optimization Approaches for Privacy. }
A growing body of work frames the design of privacy-preserving mechanisms as explicit optimization problems, aiming to maximize data utility subject to formal privacy constraints. Many adversarial or game-theoretic approaches—such as generative adversarial privacy (GAP)\cite{huang2018generative} and related GAN-based frameworks\cite{chen2018understanding,nasr2018machine,jordon2018pate}—cast the privacy mechanism designer and the adversary as players in a min-max game, optimizing utility loss and privacy leakage, respectively. More recently, Selvi et al.~\cite{selvi2025differential} introduced a rigorous optimization framework for DP based on distributionally robust optimization (DRO), formulating the mechanism design problem as an infinite-dimensional DRO to derive noise-adding mechanisms that are nonasymptotically and unconditionally optimal for a given privacy level. Their approach yields implementable mechanisms via tractable finite-dimensional relaxations, often outperforming classical Laplace or Gaussian mechanisms on benchmark tasks. Collectively, these lines of research illustrate the power of optimization and game-theoretic perspectives in achieving privacy-utility trade-offs beyond conventional privatization mechanisms.

\section{Preliminaries}

\subsection{PAC Privacy}

\noindent\textbf{Privacy Threat Model.}
We consider the following general privacy problem. 
A sensitive input $X$ (e.g., a dataset, membership status) is drawn from a distribution $\mathcal{D}$, which may be unknown or inaccessible.
There is a data processing (possibly randomized) mechanism $\mathcal{M}:\mathcal{X} \mapsto \mathcal{Y} \subset \mathbb{R}^{d}$, where $\mathcal{Y}$ is measurable.
An adversary observes the output $Y = \mathcal{M}(X)$ and attempts to estimate the original input $X$ with an estimation $\widetilde{X}$.
The adversary has complete knowledge of both the data distribution $\mathcal{D}$ and the mechanism $\mathcal{M}$, representing the worst-case scenario.
The central privacy concern is determining whether the adversary can accurately estimate the true input, meeting some predefined success criterion captured by an indicator function $\rho$.
The PAC privacy framework \cite{xiao2023pac} addresses this threat model and is formally defined as follows.

\begin{definition}[$(\delta, \rho, \mathcal{D})$ PAC Privacy~\cite{xiao2023pac}]\label{def:pac_original1}
For a data processing mechanism $\mathcal{M}$, given some data distribution $\mathcal{D}$, and a measure function $\rho(\cdot, \cdot)$, we say $\mathcal{M}$ satisfies \textup{$( \delta, \rho, \mathcal{D})$ PAC Privacy} if the following experiment is impossible:

A user generates data $X$ from distribution $\mathcal{D}$ and sends $\mathcal{M}(X)$ to an adversary. The adversary who knows $\mathcal{D}$ and $\mathcal{M}$ is asked to return an estimation $\widetilde{X} \in \mathcal{X}$ on $X$ such that with probability at least $1-\delta$, $\rho(\widetilde{X}, X) =1$. 
\end{definition}

Definition \ref{def:pac_original1} formalizes privacy in terms of the adversary's difficulty in achieving accurate reconstruction, capturing the semantics of the \textit{impossibility of customized adversarial inference}~\cite{xiao2025pac}.
The function $\rho$ specifies the success criterion for reconstruction, adapting to the requirements of the specific application. 
For example, when $\mathcal{X} \subset \mathbb{R}^{d'}$, one may define success as $\rho(\widetilde{X}, X) \equiv \mathbf{1}\{|\widetilde{X} - X|_2 \leq \epsilon\}=1$ or some small $\epsilon > 0$ with $\mathbf{1}\{\cdot\}$ as the indicator.
If $X$ is a finite set of size $n$, success may be defined as correctly recovering more than $n - \epsilon$ elements. Notably, $\rho$ need not admit a closed-form expression; it simply indicates whether the reconstruction satisfies the designated criterion for success.

In PAC Privacy, $X$ may represent general \textit{secrets} as considered in Pufferfish privacy frameworks \cite{kifer2014pufferfish,Song2017,pierquin2024renyi,zhang2025differential, zhang2025sliced}, which go beyond data points.
For example, a secret may be a dataset attribute or a global feature of the dataset. 
For ease of exposition, this paper focuses on the setting where $X$ denotes the data.
PAC Privacy treats the secrets $X$ as a random variable drawn from a distribution $\mathcal{D}$.
When $\mathcal{D}$ is not available in closed form, we may explicitly create $\mathcal{D}$ via a sampling rule and access it through i.i.d. samples from a data pool \cite{xiao2025pac}.

PAC Privacy considers the following adversarial worst-case scenario: a computationally unbounded adversary with full
knowledge of both $\mathcal{D}$ and the underlying query function.
The randomness inherent in data (or secret) generation and the randomness in the query function are the only elements unknown to the adversary \cite{xiao2023pac,xiao2025pac}.
PAC Privacy is highly flexible by enabling $\rho$ to encode a wide range of adversary models and user-specified risk criteria. For example, in membership inference attacks~\cite{carlini2022membership}, $\rho(\tilde{X}, X) = 1$ may indicate that $\widetilde{X}$ successfully determines the presence of a target data point in $X$.
In reconstruction attacks~\cite{balle2022reconstructing}, success may be defined by $\rho(\widetilde{X}, X) = 1$ if $|\widetilde{X} - X|_2 \leq 1$, representing a close approximation of the original data.

Given the data distribution $\mathcal{D}$ and the adversary's criterion $\rho$, the \textit{optimal prior success rate} $(1 - \delta^{\rho}_{o})$ is defined as the highest achievable success probability for the adversary without observing the output $\mathcal{M}(X)$: 
$\delta^{\rho}_{o} = \inf_{\widetilde{X}_0} \Pr_{X \sim \mathcal{D}}\left(\rho(\widetilde{X}_0, X) \neq 1\right).$
Similarly, the \textit{posterior success rate} $(1 - \delta)$ is defined as the adversary's probability of success after observing $\mathcal{M}(X)$.

The notion of \textit{PAC advantage privacy} quantifies how much the mechanism output $\mathcal{M}(X)$ can improve the adversary's success rate, based on $f$-divergence.

\begin{definition}[$f$-Divergence]\label{def:f_divergence}
Given a convex function $f : (0, +\infty) \rightarrow \mathbb{R}$ with $f(1) = 0$, extend $f$ to $t=0$ by setting $f(0) = \lim_{t \to 0^+} f(t)$ (in $\mathbb{R} \cup \{+\infty, -\infty\}$). The $f$-divergence between two probability distributions $P$ and $Q$ over a common measurable space is: 
\[
\mathtt{D}_f(P \| Q) \equiv 
\begin{cases} 
\mathbb{E}_Q \left[ f\left( \frac{dP}{dQ} \right) \right] & \text{if } P \ll Q, \\
+\infty & \text{otherwise},
\end{cases}
\]
where $\frac{dP}{dQ}$ is the Radon-Nikodym derivative.
\end{definition}

\begin{definition}[$(\Delta_f^\delta, \rho, \mathcal{D})$ PAC Advantage Privacy \cite{xiao2023pac}]\label{def:PAC_advantage}
A mechanism $\mathcal{M}$ is termed \textup{$(\Delta_f^\delta, \rho, \mathcal{D})$ PAC advantage private} if it is $(\delta, \rho, \mathcal{D})$ PAC private and 
$$ \Delta_f^{\delta} \equiv\mathcal{D}_{f}(\bm{1}_{\delta}\|\bm{1}_{\delta^{\rho}_o}) = \delta^{\rho}_of(\frac{\delta}{\delta^{\rho}_o}) + (1-\delta^{\rho}_o)f (\frac{1-\delta}{1-\delta^{\rho}_o}) .$$
Here, $\bm{1}_{\delta}$ and $\bm{1}_{\delta^{\rho}_o}$ represent two Bernoulli distributions of parameters $\delta$ and $\delta^{\rho}_o$, respectively. 
\end{definition}

Here, PAC Advantage Privacy is defined on top of PAC Privacy and quantifies the amount of \textit{privacy loss} incurred from releasing $\mathcal{M}(X)$, captured by the additional \textit{posterior advantage} $\Delta_f^{\delta}$.

\subsection{Automatic PAC Privatization Algorithms}\label{sec:PAC_Orignal_Algs}

PAC Privacy enables automatic privatization, which supports simulation-based implementation for arbitrary black-box mechanisms, without requiring the worst-case adversarial analysis, such as sensitivity computation.
In this section, we present the main theorems and algorithms underlying automatic PAC privatization as introduced in \cite{xiao2023pac} (hereafter "Auto-PAC") and the efficiency-improved version proposed in \cite{sridhar2024pac} (hereafter "Efficient-PAC"; algorithm details in Appendix~\ref*{app:efficient_pac} in \cite{zhang2025breaking}).
We start by defining the \textit{mutual information}.

\begin{definition}[Mutual Information]
For random variables $A$ and $B$, the mutual information is defined as
$$\mathtt{MI}(A;B) \equiv \mathcal{D}_{KL}(\mathsf{P}_{A,B}\|\mathsf{P}_{A} \otimes \mathsf{P}_{B}),$$
the KL-divergence between their joint distribution (i.e., $\mathsf{P}_{A,B}$) and the product of their marginals (i.e., $\mathsf{P}_{A}$ and $\mathsf{P}_{B}$).
\end{definition}

When the $f$-divergence in $\Delta_f^{\delta}$ is instantiated as the KL divergence (denoted as $\Delta_{\mathrm{KL}}^{\delta}$), Theorem 1 of \cite{xiao2023pac} shows 
\begin{equation}\label{eq:KL_upper_MI}
    \Delta_{\mathrm{KL}}^{\delta} \leq \mathtt{MI}(X; \mathcal{M}(X)).
\end{equation}
That is, we can control the posterior advantage $\Delta_{\mathrm{KL}}^{\delta}$ by bounding the mutual information between private data and the released output.
Importantly, this mutual information bound holds uniformly over all adversarial inference procedures (including the choice of $\rho$) permitted by PAC Privacy.
Therefore, when $\Delta_f^{\delta} = \Delta_{\mathrm{KL}}^{\delta}$, we can characterize and quantify the PAC Privacy in terms of $\mathtt{MI}(X; \mathcal{M}(X))$ without requiring any adversarial model or $\rho$ tuning, while the semantics of PAC Privacy remains as the impossibility of customized adversarial inference.

Next, we introduce the Auto-PAC.
Consider a deterministic mechanism $\mathcal{M} : \mathcal{X} \rightarrow \mathbb{R}^d$, where the output norm is uniformly bounded: $\|\mathcal{M}(X)\|_2 \leq r$ for all $X$. To guarantee PAC Privacy, the mechanism is perturbed by Gaussian noise $B \sim \mathcal{N}(0, \Sigma_B)$, where $\Sigma_B$ is the covariance.
When $X\sim\mathcal{D}$, let $\Sigma_{\mathcal{M}(X)}$ be the covariance of $\mathcal{M}(X)$.
For any deterministic mechanism $\mathcal{M}$ and any Gaussian noise $B$, define the \textit{Gaussian surrogate bound}

\begin{equation}\label{eq:logdet}
    \mathtt{LogDet}(\mathcal{M}(X), B)\equiv \frac{1}{2} \log \det \left(I_d + \Sigma_{\mathcal{M}(X)} \cdot\Sigma_B^{-1} \right).
\end{equation}

\begin{theorem}[Theorem 3 of \cite{xiao2023pac}]\label{thm:PAC_original_thm3}
For an arbitrary deterministic mechanism $\mathcal{M}$ and Gaussian noise $B \sim \mathcal{N}(0, \Sigma_B)$, the mutual information satisfies
\[
\mathtt{MI}(X; \mathcal{M}(X)+B) \leq \mathtt{LogDet}(\mathcal{M}(X), B).
\]
Moreover, there exists $\Sigma_B$ such that $\mathbb{E}[\|B\|_2^2] = \left( \sum_{j=1}^d \sqrt{\lambda_j} \right)^2$ with $\{\lambda_j\}$ being the eigenvalues of $\Sigma_{\mathcal{M}(X)}$, and $\mathtt{MI}(X; \mathcal{M}(X)+B) \leq \frac{1}{2}$.
\end{theorem}


\begin{algorithm}[t]
\caption{$(1-\gamma)$-Confidence Auto-PAC \cite{xiao2023pac}}
\label{alg:PAC_original}
\begin{algorithmic}[1]
\REQUIRE deterministic mechanism $\mathcal{M}$, dataset $\mathcal{D}$, sample size $m$, security parameter $c$, mutual information quantities $\beta'$ and $v$.
\FOR{$k=1,2,\ldots,m$}
\STATE Generate $X^{(k)}$ from $\mathcal{D}$. Record $y^{(k)}= \mathcal{M}(X^{(k)})$.
\ENDFOR 
\STATE Calculate $\hat{\mu}={\sum_{k=1}^m y^{(k)}}/{m}$ and 
$\hat{\Sigma} = {\sum_{k=1}^m (y^{(k)}-\hat{\mu})(y^{(k)}-\hat{\mu})^{\top}}/{m}.$
\STATE  Apply SVD: $\hat{\Sigma}=\hat{U}\hat{\Lambda} \hat{U}^{\top},$ where $\hat{\Lambda}$ has eigenvalues $\hat{\lambda}_1 \geq \hat{\lambda}_2 \geq \ldots \geq  \hat{\lambda}_d$.   
\STATE Find $j_0 = \arg \max_j \hat{\lambda}_j$ for $\hat{\lambda}_j > c$.
\IF{$\min_{ 1\leq j \leq j_0, 1 \leq l \leq d}  |\hat{\lambda}_j-\hat{\lambda}_l| > r\sqrt{dc}+2c$}
\FOR{$j=1, 2, \ldots, d$}
\STATE Set {\small$\displaystyle\lambda_{B,j} = \frac{2v}{ \sqrt{\hat{\lambda}_j+10cv/\beta'} \cdot \big(\sum_{j=1}^d \sqrt{\hat{\lambda}_j+10cv/\beta'} \big)}.$}
\ENDFOR
\STATE Set $\Sigma_{\bm{B}} = \hat{U}\Lambda^{-1}_{\bm{B}}\hat{U}^{\top}$.
\ELSE
\STATE Set $\Sigma_{\bm{B}}= (\sum_{j=1}^d \hat{\lambda}_j+dc)/(2v) \cdot \bm{I}_d$.
\ENDIF 
\STATE \textbf{Output}: $\Sigma_{\bm{B}}$.  
\end{algorithmic}
\end{algorithm}

Theorem \ref{thm:PAC_original_thm3} establishes a simple upper bound on the mutual information with Gaussian noise perturbation.
Choosing $\Sigma_{B}$ to implement the Gaussian surrogate bound $\mathtt{LogDet}(\mathcal{M}(X), B) = \beta$ for a privacy budget $\beta$ enables \textit{anisotropic} noise as it estimates the eigenvectors of $\mathcal{M}(X)$ to fit the noise to the geometry of the eigenspace of $\mathcal{M}(X)$.
The result extends naturally to randomized mechanisms (Corollary 2 of \cite{xiao2023pac}).
Building on Theorem \ref{thm:PAC_original_thm3}, Algorithm \ref{alg:PAC_original} (we refer to it as \textit{$(1-\gamma)$-Confidence Auto-PAC}) is proposed by \cite{xiao2023pac} to perform automatic PAC privatization.
Algorithm \ref{alg:PAC_original} aims to determine an Gaussian noise covariance $\Sigma_{B}$, so that
$\mathtt{MI}(X; \mathcal{M}(X)+B) \leq \beta$ is satisfied with confidence at least $1 -$ $\gamma$.

\subsection{Differential Privacy}

In addition to the standard PAC Privacy, we also compare our approach to the differential privacy (DP) framework.
Let $x = (x_1, x_2, \ldots, x_n) \in \mathcal{X}=(\mathcal{X}^{\dagger})^{n}$ be the input dataset, where each data point $x_i$ is defined over some measurable domain $\mathcal{X}^{\dagger}$.
We say two datasets $x, x' \in \mathcal{X}$ are \textit{adjacent} if they differ in exactly one data point.

\begin{definition}[$(\epsilon, \Bar{\delta})$-Differential Privacy \cite{dwork2006calibrating}]\label{def:DP}
    A randomized mechanism $\mathcal{M}:\mathcal{X}\mapsto\mathcal{Y}$ is said to be \textup{$(\epsilon, \Bar{\delta})$-differentially private (DP)}, with $\epsilon\geq 0$ and $\Bar{\delta}\in[0,1]$, if for any pair of adjacent datasets $x, x'$, and any measurable $\mathcal{W}\subseteq \mathcal{Y}$, it holds that 
    $\Pr[\mathcal{M}(x)\in \mathcal{W}]\leq e^{\epsilon}\Pr[\mathcal{M}(x')\in\mathcal{W}] + \Bar{\delta}$.
\end{definition}

The parameter $\epsilon$ is usually referred to as the \textit{privacy budget}, and $\delta\in(0,1]$ represents the failure probability. 
DP is an input-independent adversarial worst-case approaches that focus on the sensitivity magnitude, while Auto-PAC is instance-based and adds anisotropic noise tailored to each direction as needed.
Appendix \ref{app:difference_DP_PAC} characterizes the difference between DP, PAC Privacy, and our Residual-PAC (R-PAC) Privacy.

\section{Characterizing The Gaussian Barrier of Automatic PAC Privatization}\label{sec:utility_characterization_PAC}

This section characterizes the utility of Auto-PAC by focusing on the conservativeness of the implemented mutual information bounds.
To distinguish from Algorithm~\ref{alg:PAC_original} ($(1-\gamma)$-confidence Auto-PAC), Auto-PAC refers to the direct implementation of privacy budgets for the bound $\mathtt{LogDet}(\mathcal{M}(X), B)$ without a target conference level.
The Gaussian surrogate bound is conservative due to a nonzero \textit{Gaussianity gap}, the discrepancy between the \textit{true mutual information} and $\mathtt{LogDet}(\mathcal{M}(X), B)$ defined by (\ref{eq:logdet}):
\begin{equation}\label{eq:gap_def}
    \mathtt{Gap}_{\mathtt{d}} \equiv \mathtt{LogDet}(\mathcal{M}(X), B) - \mathtt{MI}(X; \mathcal{M}(X)+B).
\end{equation}
Define $Z = \mathcal{M}(X) + B$ with mean $\mu_Z = \mu_{\mathcal{M}(X)}$ and covariance $\Sigma_Z = \Sigma_{\mathcal{M}(X)} + \Sigma_B$. Let $P_{\mathcal{M},B}$ denote the true distribution of $Z = \mathcal{M}(X) + B$, and define the \textit{Gaussian surrogate distribution} as
\begin{equation}\label{eq:tilde_Q}
    \widetilde{Q}_{\mathcal{M}} \equiv \mathcal{N}(\mu_Z, \Sigma_Z)
\end{equation}
with the same first and second moments as $Z\sim P_{\mathcal{M},B}$.

\begin{proposition}\label{prop:MI_gap}
    Let $B \sim \mathcal{N}(0, \Sigma_{B})$. Then, $\mathtt{Gap}_{\mathtt{d}} = \mathtt{D}_{\mathrm{KL}}(P_{\mathcal{M},B}\| \widetilde{Q}_{\mathcal{M}})\geq 0$.
    Moreover, $\mathtt{Gap}_{\mathtt{d}}=0$ iff $P_{\mathcal{M},B} = \widetilde{Q}_{\mathcal{M}}$.
\end{proposition}

Proposition \ref{prop:MI_gap} shows that the conservativeness of $\mathcal{M}(X)$ in terms of the Gaussianity gap is equivalent to the KL divergence between the true output distribution and the Gaussian surrogate distribution. 
Thus, Auto-PAC tightly implements a privacy budget if and only if the true perturbed output distribution coincides with the Gaussian surrogate distribution in (\ref{prop:MI_gap}).



\begin{proposition}\label{prop:opt_PAC_algorithm}
    For any privacy budget $\beta > 0$, the noise distribution $Q = \mathcal{N}(0,\Sigma_{B})$ obtained by Auto-PAC is the unique solution of the following problem:
\begin{equation}\label{eq:opt_PAC_algorithm}
    \inf_{ Q' = \mathcal{N}(\mu, \Sigma')} \mathbb{E}_{B\sim Q'}\bigl[\|B\|_2^2\bigr] \quad \textup{s.t.} \quad
\mathtt{MI}(X; \widetilde{Z}) \leq \beta \textup{ with } \widetilde{Z} \sim \widetilde{Q}_{\mathcal{M}}.
\end{equation}
\end{proposition}

Proposition \ref{prop:opt_PAC_algorithm} implies that, if we replace $Z\sim P_{\mathcal{M},B}$ by $\widetilde{Z} \sim \widetilde{Q}_{\mathcal{M}}$, Auto-PAC's zero-mean Gaussian noise is the optimal solution to minimize the magnitude of the Gaussian noise subject to the mutual information constraint.

\begin{proposition}\label{prop:gamma_PAC_conservative}
For the same privacy budget $\beta>0$, let $Q$ and $Q_{\gamma}$, respectively, be the Gaussian noise distribution obtained by Auto-PAC and $(1-\gamma)$-Confidence Auto-PAC with any $\gamma\in[0,1]$. 
Let $B\sim Q$ and $B_{\gamma}\sim Q_{\gamma}$.
Then, the following holds.
\begin{itemize}
    \item[(i)] $\mathtt{MI}(X;\mathcal{M}(X) + B_{\gamma})\leq \mathtt{MI}(X;\mathcal{M}(X) + B)$. 
    \item[(ii)] $\mathbb{E}_{Q_{\gamma}}[\|B_{\gamma}\|^2_{2}] \geq \mathbb{E}_{Q}[\|B\|^2_{2}].$
\end{itemize}
\end{proposition}

In Proposition \ref{prop:gamma_PAC_conservative}, part (i) shows that $(1-\gamma)$-confidence Auto-PAC is more conservative than directly implementing $\mathtt{LogDet}(\mathcal{M}(X), B)$ (Auto-PAC) for the same privacy budget. Part (ii) demonstrates that $(1-\gamma)$-confidence Auto-PAC uses larger noise magnitude than Auto-PAC for the same privacy budget. Thus, in subsequent comparisons involving PAC Privacy, we focus on Auto-PAC.

\subsection{Mechanism Comparison in PAC Privacy}

Definition 9 of \cite{xiao2023pac} defines the optimal perturbation for PAC Privacy that tightly implements the privacy budget while maintaining optimal utility, where utility is captured by a loss function $\mathcal{K}$.
An optimal perturbation $Q^{*}$ is a solution of the following optimization problem:
\begin{equation}\label{eq:PAC_org_OPT}
    \inf_{Q}\; \mathbb{E}_{Q,\mathcal{M}, \mathcal{D}}[\mathcal{K}(B;\mathcal{M})] \quad \text{s.t.} \quad \mathtt{MI}(X;\mathcal{M}(X)+B) \leq \beta,  B \sim Q.
\end{equation}
The choice of utility loss function $\mathcal{K}$ is context-dependent. However, in many applications, we are primarily concerned with the expected Euclidean norm of the noise or a convex function thereof, e.g., $\mathbb{E}_{Q,\mathcal{M}, \mathcal{D}}[\mathcal{K}(B;\mathcal{M})] = \mathbb{E}_{Q}\bigl[\|B\|_2^2\bigr]$.

We now show in Proposition \ref{prop:ordering_pac} that using $\mathbb{E}_{Q,\mathcal{M}, \mathcal{D}}[\mathcal{K}(B;\mathcal{M})] = \mathbb{E}_{Q}\bigl[\|B\|_2^2\bigr]$ is sufficient to obtain perturbations that maintain \textit{coherent ordering} of PAC Privacy using mutual information (i.e., larger privacy budgets yield non-decreasing actual mutual information).

\begin{proposition}\label{prop:ordering_pac}
    Fix a mechanism $\mathcal{M}$ and data distribution $\mathcal{D}$.
    Let $\mathcal{Q}$ denote the collection of all zero-mean noise distributions under consideration, and let $\mathtt{I}_{\mathrm{true}} : \mathcal{Q} \mapsto \mathbb{R}_{\geq0}$ be the true mutual information functional; i.e., $\mathtt{I}_{\mathrm{true}}(Q) = \mathtt{MI}(X;\mathcal{M}(X)+B)$ with $B \sim Q$ for $Q \in \mathcal{Q}$.
    For each privacy budget $\beta \geq 0$, define the feasible region $\displaystyle\mathcal{F}(\beta) \equiv \{Q \in \mathcal{Q} : \mathtt{I}_{\mathrm{true}}(Q) \leq \beta\}.$
    Suppose that $\mathcal{F}(\beta)$ is nonempty for all privacy budgets of interest. 
    For each $\beta \geq 0$, let $Q^{*}(\beta)$ be a solution of the problem:
    \begin{equation}
        \min_{Q} \mathbb{E}_{B \sim Q}[\|B\|^{2}_{2}] \quad \text{s.t.} \quad Q \in \mathcal{F}(\beta).
    \end{equation}
    Then, if $\beta_{1} < \beta_{2}$, we have $\mathtt{I}_{\mathrm{true}}(Q^{*}(\beta_{1})) \leq \mathtt{I}_{\mathrm{true}}(Q^{*}(\beta_{2}))$.  
\end{proposition}

However, if Auto-PAC is used to solve the optimization problem (\ref{eq:opt_PAC_algorithm}), we have the conservative implementation of a given privacy budget.
For any mechanism $\mathcal{M}: \mathcal{X} \mapsto \mathcal{Y}$, we let $\mathtt{Gap}_{\mathtt{d}}(Q) = \mathtt{D}_{\mathrm{KL}}(P_{\mathcal{M},B} \| \widetilde{Q}_{\mathcal{M}})$ with $B \sim Q$.
The next result shows that when $\mathtt{Gap}_{\mathtt{d}}(Q)> 0$, Auto-PAC does not, in general, maintain coherent ordering of PAC Privacy.

\begin{theorem}\label{thm:non_coherent_PAC_org}
    Fix a mechanism $\mathcal{M}$ and data distribution $\mathcal{D}$.
    Let $\mathcal{Q}$ denote the collection of all zero-mean Gaussian distributions under consideration, and let $\mathtt{I}_{\mathrm{true}} : \mathcal{Q} \mapsto \mathbb{R}_{\geq0}$ be the true mutual information functional; i.e., $\mathtt{I}_{\mathrm{true}}(Q) = \mathtt{MI}(X;\mathcal{M}(X)+B)$ with $B \sim Q$ for $Q \in \mathcal{Q}$.
    For each $\beta \geq 0$, let $Q^{*}(\beta)$ be a solution of the optimization in Proposition \ref{prop:opt_PAC_algorithm}.
    For any $0 < \beta_{1} < \beta_{2}$, define $\mathtt{G}(\beta_{2}, \beta_{1}) \equiv \mathtt{Gap}_{\mathtt{d}}(Q^*(\beta_{2})) - \mathtt{Gap}_{\mathtt{d}}(Q^*(\beta_{1})).$
    Then:
    \begin{enumerate}
        \item[(i)] If $\mathtt{G}(\beta_{2}, \beta_{1}) \leq \beta_{2} - \beta_{1}$, then $\mathtt{I}_{\mathrm{true}}(Q^{*}(\beta_{1})) \leq \mathtt{I}_{\mathrm{true}}(Q^{*}(\beta_{2}))$.
        \item[(ii)] If $\mathtt{G}(\beta_{2}, \beta_{1}) > \beta_{2} - \beta_{1}$, then $\mathtt{I}_{\mathrm{true}}(Q^{*}(\beta_{1})) > \mathtt{I}_{\mathrm{true}}(Q^{*}(\beta_{2}))$.
    \end{enumerate}
\end{theorem}

Theorem \ref{thm:non_coherent_PAC_org} characterizes when Auto-PAC maintains coherent ordering of actual information leakage $\mathtt{I}_{\mathrm{true}} = \beta - \mathtt{Gap}_{\mathtt{d}}$, and when not.
Increasing the budget from $\beta_1$ to $\beta_2$ permits extra leakage $\beta_{2} - \beta_{1}$ by using Auto-PAC, but part may be wasted if the mechanism output becomes more non-Gaussian. The wasted portion is $\mathtt{G}(\beta_{2}, \beta_{1}) = \mathtt{Gap}_{\mathtt{d}}(Q^*(\beta_{2})) - \mathtt{Gap}_{\mathtt{d}}(Q^*(\beta_{1}))$. If this waste exceeds the budget increase, then $\mathtt{I}_{\mathrm{true}}$ decreases despite a larger nominal budget, violating coherent ordering. This result cautions against comparing mechanisms using Auto-PAC solely by budgets, as identical budgets may yield different true PAC Privacy leakages depending on their respective Gaussianity gaps.

\subsection{$\mathtt{Gap}_{\mathtt{d}}$ Reduction via Non-Gaussianity Correction}\label{sec:gap_reduction}

In this section, we propose two approaches to reduce $\mathtt{Gap}_{\mathtt{d}}$ after a $\mathcal{N}(0, \Sigma_{B})$ is determined by Auto-PAC.
For any deterministic mechanism $\mathcal{M}$ and Gaussian noise $B \sim \mathcal{N}(0, \Sigma_{B})$, recall the Gaussian surrogate distribution $\widetilde{Q}_{\mathcal{M}} = \mathcal{N}(\mu_{Z}, \Sigma_{Z})$ in (\ref{eq:tilde_Q}).
Let $\mathtt{D}_{Z} = \mathtt{D}_{\mathrm{KL}}(P_{\mathcal{M},B} \| \widetilde{Q}_{\mathcal{M}})$. By Proposition \ref{prop:MI_gap}, $\mathtt{Gap}_{\mathtt{d}} = \mathtt{D}_{Z}$.
For any estimator $\widehat{\mathtt{D}}_{Z}$ of $\mathtt{D}_{Z}$, define the \textit{improved mutual information} estimate:
\[
\mathtt{IMI}(\widehat{\mathtt{D}}_{Z}) \equiv \mathtt{LogDet}(\mathcal{M}(X), B) - \widehat{\mathtt{D}}_{Z}.
\]
For $0\leq \widehat{\mathtt{D}}_{Z} \leq \mathtt{D}_{Z}$, we have $\mathtt{MI}(X;\mathcal{M}(X)+B) \leq \mathtt{IMI}(\widehat{\mathtt{D}}_{Z}) \leq \mathtt{LogDet}(\mathcal{M}(X), B).$
Thus, if we can get $\widehat{\mathtt{D}}_{Z}$ between $\mathtt{D}_{Z}$ and $0$
after Auto-PAC privatization is performed, then for any $\Sigma_{B}$ that ensures $\mathtt{LogDet}(\mathcal{M}(X), B) = \beta$, we have $\mathtt{IMI}(\widehat{\mathtt{D}}_{Z}) = \beta - \widehat{\mathtt{D}}_{Z}$ as surrogate upper bound that is tighter than $\mathtt{LogDet}(\mathcal{M}(X), B)$.
Thus, we can have tighter privacy accounting post-hoc to the Auto-PAC privatization to save additional privacy budget, without requiring direct mutual information estimation.

Before describing the approaches, we first introduce two standard discrepancy measures between $P_{\mathcal{M}, B}$ and $\widetilde{Q}_{\mathcal{M}}$.

\begin{definition}[Donsker–Varadhan (DV) Objective \cite{donsker1975asymptotic}]\label{def:DV_objective}
For probability measures $P$ and $Q$ on a common measurable space,
\[
\mathtt{D}_{\mathrm{KL}}(P\|Q)\;=\;\sup_{f:\,\mathcal{Y}\to\mathbb R}\Big\{\,\mathbb{E}_{P}[f(Y)]-\log \mathbb{E}_{Q}\big[e^{f(Y)}\big]\Big\} ,
\]
where the supremum ranges over measurable $f$ such that $\mathbb{E}_Q[e^{f}]<\infty$.
We call 
$
\mathcal{J}(f;P,Q) \equiv \mathbb{E}_{P}[f]-\log \mathbb{E}_{Q}[e^{f}]
$
the \textit{DV objective}.  In our setting,
$\mathtt{D}_Z=\mathtt{D}_{\mathrm{KL}}(P_Z\|\widetilde{Q}_{\mathcal{M}})
=\sup_{f}\mathcal{J} \big(f;P_{\mathcal{M}, B},\widetilde{Q}_{\mathcal{M}}\big)$.
\end{definition}

\begin{definition}[Sliced Wasserstein Distances (SWD) \cite{rabin2011wasserstein,bonneel2015sliced}]
For $p\geq 1$, the $p$-Wasserstein distance between $P$ and $Q$ on $\mathbb{R}^d$ is
$\displaystyle \mathrm{W}_p(P,Q)\;=\;\Big(\inf_{\eta\in\widehat{\Pi}(P,Q)}\,\mathbb{E}_{(X,Y)\sim\eta}\big[\|X-Y\|_2^p\big]\Big)^{1/p},$
%
where $\widehat{\Pi}(P,Q)$ is the set of couplings with marginals $P$ and $Q$.
The \textit{sliced} $p$-Wasserstein distance averages 1-D Wasserstein distances over directions $v$ on the unit sphere $\mathbb{S}^{d-1}$:
\[
\mathrm{SW}_p^p(P,Q)\;=\;\int_{\mathbb{S}^{d-1}} \mathrm{W}_p^p \big(\mathcal{L}(\langle v,X\rangle),\,\mathcal{L}(\langle v,Y\rangle)\big)\,d\sigma(v),
\]
where $\sigma$ is the uniform (Haar) measure on $\mathbb S^{d-1}$ and $\mathcal{L}(\cdot)$ denotes the law of its argument.
In our setting we write $\mathrm{W}_p(P_{\mathcal{M}, B},\widetilde{Q}_{\mathcal{M}})$ and $\mathrm{SW}_p(P_{\mathcal{M}, B},\widetilde{Q}_{\mathcal{M}})$.
\end{definition}

\begin{definition}[Finite-Sample Lower-Confidence DV Estimator]\label{def:lce-dv}
Fix a function class $\mathcal{F}\subset\{f:\mathbb{R}^d \to \mathbb{R}\}$ with $0\in\mathcal{F}$ and let $\widehat{\mathcal{J}}(f;S_P,S_Q)\;\equiv\;
\frac{1}{|S_P|}\sum_{z\in S_P} f(z)\;-\;\log \Big(\frac{1}{|S_Q|}\sum_{z\in S_Q} e^{f(z)}\Big)$ denote the empirical DV objective on samples $S_P$ from $P_Z$ and $S_Q$ from $\widetilde{Q}_{\mathcal{M}}$.
Draw four independent splits
$S_P^{\mathrm{tr}},S_Q^{\mathrm{tr}},S_P^{\mathrm{val}},S_Q^{\mathrm{val}}$
with sizes $n_P^{\mathrm{tr}},n_Q^{\mathrm{tr}},n_P^{\mathrm{val}},n_Q^{\mathrm{val}}$ respectively, and fit $\widehat{f}_{\mathrm{tr}}\;\in\;\arg\max_{f\in\mathcal{F}}\;
\widehat{\mathcal{J}} \big(f;S_P^{\mathrm{tr}},S_Q^{\mathrm{tr}}\big).$

Let $\Gamma_{\hat{\delta}}=\Gamma_{\hat{\delta}}\big(\mathcal{F},n_P^{\mathrm{val}},n_Q^{\mathrm{val}}\big)$ be any valid uniform deviation bound satisfying, with probability at least $1-{\hat{\delta}}$, $\sup_{f\in\mathcal{F}}\Big|
\widehat{\mathcal{J}}\big(f;S_P^{\mathrm{val}},S_Q^{\mathrm{val}}\big)
-\mathcal{J} \big(f;P_{\mathcal{M}, B},\widetilde{Q}_{\mathcal{M}}\big)\Big|
\leq\Gamma_{\hat{\delta}}$,
where $\mathcal{J}(f;P,Q)$ is the DV objective (Definition \ref{def:DV_objective}).
The \textit{finite-sample lower-confidence estimator} of $\mathtt{D}_Z=\mathtt{D}_{\mathrm{KL}}(P_{\mathcal{M}, B}\|\widetilde{Q}_{\mathcal{M}})$ is
\[
\widehat{\mathtt{D}}_{\mathrm{LCE}}
\equiv
\Big[\widehat{\mathcal{J}}\big(\widehat{f}_{\mathrm{tr}};S_P^{\mathrm{val}},S_Q^{\mathrm{val}}\big)-\Gamma_{\hat{\delta}}\Big]_{+}.
\]
\end{definition}

Definition \ref{def:lce-dv} specifies a finite-sample lower-confidence estimator.

\begin{theorem}[DV-Based Correction]\label{thm:DV}
Let $Z=\mathcal{M}(X)+B$ with deterministic $\mathcal{M}$ and $B\sim\mathcal{N}(0$, $\Sigma_B)$, and let $\widetilde{Q}_{\mathcal{M}}$ be defined by (\ref{eq:tilde_Q}). Assume $P_{\mathcal{M},B}\ll \widetilde{Q}_{\mathcal{M}}$. For any measurable $f:\mathbb{R}^d \to \mathbb{R}$ with $\mathbb{E}_{\widetilde{Q}_{\mathcal{M}}}[e^{f(Z)}]<\infty$, define
\[
\widehat{\mathtt{D}}_Z(f)
\equiv
\mathcal{J} \big(f;P_{\mathcal{M}, B},\widetilde{Q}_{\mathcal{M}}\big)
=\mathbb{E}_{P_Z}[f(Z)]-\log\mathbb{E}_{\widetilde{Q}_{\mathcal{M}}} \big[e^{f(Z)}\big].
\]
Let $\widehat{\mathtt{D}}_{\mathrm{LCE}}$ be the finite-sample lower-confidence estimator from Definition~\ref{def:lce-dv}.
Then: 
\begin{itemize}
\item[\textnormal{(i)}] $0 \leq \sup_{f}\widehat{\mathtt{D}}_Z(f) =\mathtt{D}_{\mathrm{KL}} (P_{\mathcal{M}, B}\|\widetilde{Q}_{\mathcal{M}})$.
\item[\textnormal{(ii)}] For every $f$, $\widehat{\mathtt{D}}_Z(f)\leq \mathtt{D}_{\mathrm{KL}} \big(P_{\mathcal{M}, B}\big\|\widetilde{Q}_{\mathcal{M}}\big)\equiv \mathtt{D}_Z$.
\item[\textnormal{(iii)}] With probability at least $1-{\hat{\delta}}$ (over the independent validation splits in Definition~\ref{def:lce-dv}), $0\leq \widehat{\mathtt{D}}_{\mathrm{LCE}}\leq \mathtt{D}_Z$.
\end{itemize}
\end{theorem}

\begin{theorem}[SWD-Based Correction]\label{thm:SWD}
Let $Z=\mathcal{M}(X)+B$ with deterministic $\mathcal{M}$ and $B\sim\mathcal{N}(0,\Sigma_B)$, and let $\widetilde{Q}_{\mathcal{M}}=\mathcal{N}(\mu_Z,\Sigma_Z)$ be defined by (\ref{eq:tilde_Q}), and let $\lambda_{\max}(\Sigma_Z)$ be the largest eigenvalue of $\Sigma_Z$. Assume $P_{\mathcal{M},B}\ll \widetilde{Q}_{\mathcal{M}}$ and $\Sigma_Z\succ 0$. Define
\[
\widehat{\mathtt{D}}_{Z} \equiv \frac{1}{2\lambda_{\max}(\Sigma_Z)}\mathrm{SW}_2^2\big(P_{\mathcal{M}, B},\widetilde{Q}_{\mathcal{M}}\big).
\]
Then $0\leq \widehat{\mathtt{D}}_{Z}\leq \mathtt{D}_Z$.
\end{theorem}

Theorems \ref{thm:DV} and \ref{thm:SWD} lead to Corollary \ref{thm:pac_privacy}.

\begin{corollary}\label{thm:pac_privacy}
Let $\mathcal{M}: \mathcal{X} \mapsto \mathbb{R}^{d}$ be an arbitrary deterministic mechanism and $B \sim \mathcal{N}(0, \Sigma_{B})$ such that $\mathtt{LogDet}(\mathcal{M}(X), B) = \beta$. Under the assumptions of Theorems \ref{thm:DV} and \ref{thm:SWD}, the perturbed mechanism $Z = \mathcal{M}(X) + B$ is PAC private with 
\begin{equation}\label{eq:gap_reduced}
    \mathtt{MI}(X;Z) \leq \beta - \widehat{\mathtt{D}}_Z < \beta,
\end{equation}
    where $\widehat{\mathtt{D}}_Z > 0$ is obtained by Theorem \ref{thm:DV} ($\widehat{\mathtt{D}}_Z(f)$) or Theorem \ref{thm:SWD}.
    In addition, if $\widehat{\mathtt{D}}_Z = \widehat{\mathtt D}_{\mathrm{LCE}}$, then (\ref{eq:gap_reduced}) holds with probability at least $1-\hat{\delta}$.
\end{corollary}

Corollary \ref{thm:pac_privacy} shows that accounting for non-Gaussianity through the correction term $\widehat{\mathtt{D}}_{Z} > 0$ yields $\mathtt{MI}(X; Z) \leq \beta - \widehat{\mathtt{D}}_{Z} < \beta$, where the correction is obtained via DV-based correction or sliced Wasserstein correction.
In practice, $\widehat{\mathtt{D}}_{Z}$ estimates the Gaussianity gap $\mathtt{Gap}_{d}$, capturing the saved privacy budget, which is particularly valuable for budget savings in mechanism composition.
However, this budget-saving approach is post-hoc after Auto-PAC privatization.
Appendix~\ref*{app:non_gaussianity_correction} in \cite{zhang2025breaking} provides additional discussions and interpretations.
Next, we propose a new privacy framework enabling automatic optimal privacy budget implementation.

\section{Residual-PAC Privacy}

Recall that PAC Advantage Privacy (Definition~\ref{def:PAC_advantage}) quantifies the amount of \textit{privacy leaked by} $\mathcal{M}(X)$ in terms of the posterior advantage $\Delta_f^\delta$ encountered by the adversary.
Complementing this perspective, we introduce the notion of \textit{posterior disadvantage} encountered by the adversary, which captures the amount of \textit{residual privacy protection} that persists after leakage by $\mathcal{M}(X)$. 

To formalize this residual protection, we first define the \textit{intrinsic privacy} of a data distribution $\mathcal{D}$ relative to a fixed reference distribution $\mathcal{R}$ on $\mathcal{X}$ such that \textit{(i)} $\mathrm{supp}(\mathcal{D})\subseteq \mathrm{supp}(\mathcal{R})$ and \textit{(ii)} the f-divergence $\mathtt{D}_{f}(\mathcal{D}\|\mathcal{R})$ is finite (when $\mathtt{D}_{f}$ is the KL-divergence, this means the entropy of $\mathcal{R}$ is finite; see Section \ref{sec:foundation_residual_pac} for the formal definition of Shannon/differential entropy).
The intrinsic privacy is then defined based on $f$-divergence as
\[
\mathtt{IntP}_{f}(\mathcal{D}) = -\,\mathtt{D}_{f}(\mathcal{D}\,\|\,\mathcal{R}),
\]
where $\mathtt{D}_f(\mathcal{D}\|\mathcal{R})$ is the $f$-divergence between $\mathcal{D}$ and $\mathcal{R}$, quantifying how much $\mathcal{D}$ deviates from the reference $\mathcal{R}$. Intuitively, $-\mathtt{D}_f(\mathcal{D}\|\mathcal{R})$ rewards distributions that remain close to the "random guess" using $\mathcal{R}$, and by construction $\mathtt{IntP}_{f}(\mathcal{D})\leq 0$, attaining zero exactly when $\mathcal{D} = \mathcal{R}$.

\textbf{Examples of $\mathcal{R}$. }
When $\mathcal{X}$ is bounded, $\mathcal{R}$ can be the uniform law $\mathcal{U}$ on $\mathcal{X}$.
However, on an unbounded $\mathcal{X}$, the uniform reference $\mathcal{R}=\mathcal{U}$ has infinite volume $\int_{\mathcal{X}}dx = \infty$, potentially making $\mathtt{IntP}_{f}(\mathcal{D})$ vacuous or undefined.
To avoid this, we instead require $\mathcal{R}$ to satisfy $\mathtt{D}_{f}(\mathcal{D}\|\mathcal{R})<\infty$.
For example, one can choose $\mathcal{R}$ by: \textit{(i)} truncated uniform on a large but bounded region containing $\mathrm{supp}(\mathcal{D})$, 
\textit{(ii)} maximum-entropy Gaussian matching known moments of $\mathcal{D}$, or 
\textit{(iii)} smooth pullback of uniform on $(0,1)^d$ via bijection (e.g., component-wise sigmoid).
Under any of these constructions, $\mathcal{R}$ retains the "random-guess" semantics yet has finite $\mathtt{D}_{f}(\mathcal{D}\,\|\,\mathcal{R})$, ensuring $\mathtt{IntP}_{f}(\mathcal{D})$ remains meaningful even on unbounded $\mathcal{X}$.
Please see Appendix~\ref*{app:technical_reference} in \cite{zhang2025breaking} for a detailed discussion.

\begin{definition}[$(\mathtt{R}_{f}^{\delta}, \rho, \mathcal{D})$ Residual-PAC (R-PAC) Privacy]\label{def:residual_PAC}
A mechanism $\mathcal{M}$ is said to be $(\mathtt{R}_{f}^{\delta}, \rho, \mathcal{D})$ \textup{Residual-PAC (R-PAC) private} if it is $(\delta, \rho, \mathcal{D})$ PAC private and
\[
\mathtt{R}_{f}^{\delta} \equiv \mathtt{IntP}_{f}(\mathcal{D}) - \mathtt{D}_f(\mathbf{1}_{\delta} \| \mathbf{1}_{\delta^\rho_o}),
\]
is the \textup{posterior disadvantage},
where $\mathbf{1}_\delta$ and $\mathbf{1}_{\delta^\rho_o}$ are indicator distributions representing the adversary's inference success before and after observing the mechanism's output, respectively.
\end{definition}

The posterior disadvantage $\mathtt{R}_{f}^{\delta}$ captures the \textit{residual privacy guarantee}, which is the portion of intrinsic privacy (w.r.t. a reference $\mathcal{R}$) that remains uncompromised after the privacy loss $\Delta_f^\delta=\mathtt{D}_f(\mathbf{1}_{\delta} \| \mathbf{1}_{\delta^\rho_o})$ (Definition \ref{def:PAC_advantage}).
Then, the total intrinsic privacy is precisely decomposed as
\begin{equation}\label{eq:link_PAC_privacy}
    \mathtt{IntP}_f(\mathcal{D}) = \mathtt{R}_f^\delta + \Delta_f^\delta.
\end{equation}
This relationship provides a complete and interpretable quantification of privacy risk, distinguishing between the privacy that is lost and that which endures after information disclosure via $\mathcal{M}(X)$.
Analogous to PAC Privacy, membership inference attacks (MIA) and R-PAC Membership Privacy can be instantiated from R-PAC Privacy. See Appendix~\ref*{app:MIA} in \cite{zhang2025breaking} for detailed constructions.

\subsection{Foundation of Residual-PAC Privacy}\label{sec:foundation_residual_pac}

In this section, we develop general results to support concrete analyses under R-PAC Privacy framework. We begin by introducing key information-theoretic quantities, entropy and conditional entropy.

\textbf{Entropy. }
The \textit{Shannon entropy} of a discrete random variable $X$ on alphabet $\mathcal{X}$ is given by 
\[
\mathcal{H}(X) = -\sum\nolimits_{x \in \mathcal{X}} P_X(x) \log P_X(x)
\]
while for continuous $X$, the \textit{differential entropy} is 
\[
h(X) = -\int\nolimits_{\mathcal{X}} f_X(x) \log f_X(x) \, dx.
\]

\textbf{Conditional Entropy.}
Let $(X,W)$ be jointly distributed random variables. 
When $X$ is discrete, the \textit{conditional entropy} of $X$ given $W$ is defined by
\[
\mathcal{H}(X|W) \equiv \mathbb{E}_{W}\!\left[\mathcal{H}(X|W=w)\right].
\]
When $X$ is continuous, the conditional entropy is $h(X|W) \equiv \mathbb{E}_{W}\!\left[h(X|W=w)\right]$.
Here, the expectation is $\sum_{w\in\mathcal{W}}P_W(w)(\cdot)$ if $W$ is discrete with mass $P_W$, and $\int_{\mathcal{W}} f_W(w)(\cdot)\,dw$ if $W$ is continuous with density $f_W$.

For ease of exposition, we use $\mathcal{H}(X)$ to denote the entropy of $X$, either Shannon or differential depending on the context, and $\mathcal{H}(X|W)$ to denote the corresponding conditional entropy.
When all entropies are finite, mutual information can equivalently be expressed as
\begin{equation}\label{eq:MI_def_CE}
    \mathtt{MI}(X;W)
= \mathcal{H}(X) - \mathcal{H}(X|W).
\end{equation}

Consider any $f$-divergence $\mathtt{D}_f$, Theorem 1 of \cite{xiao2023pac} shows that the posterior advantage $\Delta_f^{\delta}$ is bounded by the minimum $f$-divergence between the joint distribution of $(X, \mathcal{M}(X))$, denoted by $P_{X, \mathcal{M}(X)}$, and the product of the marginal distribution $P_X$ and any auxiliary output distribution $P_W$ independent of $X$:
\begin{equation}\label{eq:PAC_org_first_bound}
    \Delta_f^{\delta}\; \leq \; \inf_{P_W} \; \mathtt{D}_f\big(P_{X, \mathcal{M}(X)} \;\|\; P_X \otimes P_W\big),
\end{equation}
where $P_{X, \mathcal{M}(X)}$ denotes the joint distribution of the data and mechanism output, $P_{X} = \mathcal{D}$, and $P_W$ ranges over all distributions on the output space.
When $\mathtt{D}_{f}$ is instantiated as $\mathtt{D}_{\mathtt{KL}}$ and $P_W=P_{\mathcal{M}(X)}$, we obtain (\ref{eq:KL_upper_MI}).

Thus, for any $f$-divergence $\mathtt{D}_f$, inequality (\ref{eq:PAC_org_first_bound}) implies that a mechanism $\mathcal{M} : \mathcal{X} \rightarrow \mathcal{Y}$ satisfies $(\mathtt{R}_{f}^{\delta}, \rho, \mathcal{D})$ R-PAC Privacy if  
\begin{equation}\label{eq:residual_inequality_f}
    \mathtt{R}_{f}^{\delta} \; \geq \; \mathtt{IntP}_{f}(\mathcal{D}) - \inf_{P_W} \; \mathtt{D}_f\left(P_{X, \mathcal{M}(X)} \,\|\, P_X \otimes P_W\right).
\end{equation}
Let $R$ be a random variable of the reference $\mathcal{R}$ over $\mathcal{X}$.
Corollary \ref{prop:Residual_PAC_basic} follows from Theorem 1 of \cite{xiao2023pac}.
\begin{corollary}\label{prop:Residual_PAC_basic}
Suppose that $\mathcal{H}(X)$ is finite and let $\mathtt{D}_f$ be the KL divergence. A mechanism $\mathcal{M} : \mathcal{X} \rightarrow \mathcal{Y}$ satisfies $(\mathtt{R}_{f}^{\delta}, \rho, \mathcal{D})$ R-PAC Privacy if 
\[
\mathtt{R}_{f}^{\delta} \geq \mathcal{H}(X|\mathcal{M}(X)) - \mathtt{V},
\]
where $\mathtt{V} = \mathcal{H}(R)$ is the entropy of the reference distribution.
\end{corollary}

Corollary \ref{prop:Residual_PAC_basic} establishes that when $\mathcal{H}(X)$ is finite, residual privacy $\mathtt{R}_{f}^{\delta}$ is lower bounded by $\mathcal{H}(X | \mathcal{M}(X)) - \mathtt{V}$, where $\mathtt{V}$ is independent of both data distribution $\mathcal{D}$ and mechanism $\mathcal{M}$. Since $\mathtt{V}$ is constant, $\mathtt{R}_{f}^{\delta} - \mathtt{V}$ effectively provides a privacy quantification lower-bounded by conditional entropy $\mathcal{H}(X|\mathcal{M}(X))$. 
If 
$\mathtt{D}_{f}(\mathcal{D}\|\mathcal{R}) < \infty$, then the inequality (\ref{eq:residual_inequality_f}) holds without requiring $\mathcal{H}(X) < \infty$.

\subsection{Stackelberg Residual-PAC Automatic Privatization}

In this section, we present our algorithms for automatic R-PAC privatization when the $f$-divergence is instantiated with KL divergence, under which the worst-case residual privacy is quantified by conditional entropy.
For a utility loss function $\mathcal{K}$, we define the optimal perturbation problem for any R-PAC privacy budget $\hat{\beta}$ as:
\begin{equation}\label{eq:general_OPT_RPAC}
    \inf_{Q}\; \mathbb{E}_{Q,\mathcal{M}, \mathcal{D}}[\mathcal{K}(B;\mathcal{M})] \quad \text{s.t.} \quad \mathcal{H}(X|\mathcal{M}(X) + B) \geq \hat{\beta}, \; B \sim Q.
\end{equation}
When $\mathcal{H}(X)$ is finite, by (\ref{eq:MI_def_CE}), any solution $Q^*$ to problem (\ref{eq:general_OPT_RPAC}) also solves (\ref{eq:PAC_org_OPT}) with PAC privacy budget $\beta = \mathcal{H}(X) - \hat{\beta}$.
In addition, since $\mathtt{MI}(X; \mathcal{M}(X) + B) = \mathcal{H}(X) - \mathcal{H}(X|\mathcal{M}(X) + B)$ with finite $\mathcal{H}(X)$, solving the optimal perturbation problem (\ref{eq:general_OPT_RPAC}) with conditional entropy constraints presents the same computational challenges as (\ref{eq:PAC_org_OPT}).

To address this limitation, we present a novel automatic privatization algorithm for R-PAC privacy, termed \textit{Stackelberg Residual-PAC (SR-PAC)}. Our approach is based on a Stackelberg game-theoretic characterization of the optimization (\ref{eq:general_OPT_RPAC}). We show that SR-PAC achieves optimal perturbation without wasting privacy budget. Consequently, when $\mathbb{E}_{Q,\mathcal{M}, \mathcal{D}}[\mathcal{K}(B;\mathcal{M})] = \mathbb{E}_{Q}[\|B\|_2^2]$, SR-PAC can achieve superior utility performance compared to Auto-PAC and Efficient-PAC (Appendix~\ref*{app:efficient_pac} in \cite{zhang2025breaking}) for the same mutual information privacy budget.

\begin{algorithm}[ht]
\caption{Monte Carlo SR-PAC}
\label{alg:SR_PAC}
\begin{algorithmic}[1]
\REQUIRE Privacy budget $\hat{\beta}$, decoder family $\Pi_{\phi}$, 
         perturbation rule family $\Gamma_{\lambda}$, utility loss $\mathcal{K}(\cdot)$, 
         learning rates $\eta_{\phi}, \eta_{\lambda}$, penalty weight $\sigma$, 
         iterations $T_{\lambda}, T_{\phi}$, batch size $m$
\STATE Initialize parameters $\lambda, \phi \sim \text{init}()$
\FOR{$t = 1, \ldots, T_{\lambda}$}
    \IF{$t \bmod T_{\phi} = 0$}
        \STATE \textbf{Update Decoder:}
        \FOR{$i = 1, \ldots, T_{\phi}$}
            \STATE Sample $\{(x_j, b_j, y_j)\}_{j=1}^m$ where $x_j \sim \mathcal{D}$, $b_j \sim Q_{\lambda}$, $y_j = \mathcal{M}(x_j) + b_j$
            \STATE $\widehat{W} = \frac{1}{m}\sum_{j=1}^m [-\log \pi_{\phi}(x_j | y_j)]$
            \STATE $\phi \leftarrow \phi - \eta_{\phi} \nabla_{\phi} \widehat{W}$
        \ENDFOR
    \ENDIF
    \STATE \textbf{Update Perturbation Rule:}
    \STATE Sample $\{(x_j, b_j, y_j)\}_{j=1}^m$ where $x_j \sim \mathcal{D}$, $b_j \sim Q_{\lambda}$, $y_j = \mathcal{M}(x_j) + b_j$
    \STATE $H_c = \frac{1}{m}\sum_{j=1}^m [-\log \pi_{\phi}(x_j | y_j)]$ 
    \STATE $\mathcal{L}_{\lambda} = \frac{1}{m}\sum_{j=1}^m \mathcal{K}(b_j) + \sigma (H_c - \hat{\beta})^2$
    \STATE $\lambda \leftarrow \lambda - \eta_{\lambda} \nabla_{\lambda} \mathcal{L}_{\lambda}$
\ENDFOR
\RETURN Optimal parameters $(\lambda^*, \phi^*)$
\end{algorithmic}
\end{algorithm}

Our SR-PAC algorithm recasts the optimal perturbation problem (\ref{eq:general_OPT_RPAC}) as a Stackelberg game between a \textit{Leader} (who chooses the \textit{perturbation rule} $Q$) and a \textit{Follower} (who chooses the \textit{decoder} attempting to infer $X$ from $Y$).
Let $\Gamma$ denote a rich family of noise distributions.
Let $\Pi = \{\pi: \pi(\cdot|y) \in \Delta(\mathcal{X}), y \in \mathcal{Y}\}$ denote a rich family of decoder distributions (e.g., all conditional density functions on $\mathcal{X}$ given $\mathcal{Y}$, or a parameterized neural network family).

\noindent\textbf{Follower's Problem. }
For a fixed perturbation rule $Q$, the Follower chooses decoder $\pi$ to minimize the expected log score
\[
W(Q,\pi)\equiv \mathbb{E}_{X\sim\mathcal{D},B\sim Q}\left[-\log\pi(X|\mathcal{M}(X) + B)\right].
\]
That is, the follower aims to find $\pi^{*}(Q)\in \arg\inf_{\pi\in\Pi}W(Q,\pi)$.

\noindent\textbf{Leader's Problem.}
Given a privacy budget $\hat{\beta}$, the Leader chooses $Q$ to solve
\[
\begin{aligned}
    \inf_{\,Q\,\in\Gamma\,}
  \mathbb{E}_{X\sim \mathcal{D}, B\sim Q}\bigl[\mathcal{K}(B;\mathcal{M})\bigr],
  \text{ s.t. }
  \inf_{\pi\in\Pi} W(Q,\pi) \geq \hat{\beta}.
\end{aligned}
\]
Therefore, a profile $(Q^*, \pi^*)$ is a \textit{Stackelberg equilibrium} if it satisfies
\begin{equation}\label{eq:stackelberg_equilibrium}
    \begin{cases}
 Q^*\in\arg\inf_{Q\in\Gamma} 
  \mathbb{E}[\mathcal{K}(B;\mathcal{M})], \text{ s.t. }
  W\bigl(Q,\,\pi^*(Q)\bigr)\geq\hat{\beta},\\
  \pi^*(Q) \in \arg\inf_{\pi\in\Pi} W(Q,\pi).
\end{cases}
\end{equation}

When we consider output perturbation and the utility loss $\mathcal{K}$ is chosen such that $Q\mapsto \mathbb{E}_{X\sim P_X, B\sim Q}\bigl[\mathcal{K}(B;\mathcal{M})\bigr]$ is convex in $Q$, the problem (\ref{eq:stackelberg_equilibrium}) is convex in both $Q$ and $\pi$.
Specifically, for each fixed perturbation rule $Q$, the map $\pi\mapsto W(Q,\pi)$ is a convex function of $\pi$. 
Similarly, for each fixed decoder $\pi$, the function $Q\mapsto W(Q,\pi)$ is convex in $Q$.
Because these two convexity properties hold simultaneously, $(Q,\pi)\mapsto W(Q,\pi)$ is jointly convex on $\Gamma\times \Pi$.
By the partial minimization theorem \cite[Section 3.2.5]{boyd2004convex}, taking the pointwise infimum over $\pi$ preserves convexity in $Q$. Thus, $Q \mapsto \inf_{\pi\in\Pi} W(Q,\pi)$ is a convex function of $Q$. 
Consequently, once the Follower replaces $\pi$ by its best response $\pi^{*}(Q)$, the Leader's feasible set $\{Q\in\Gamma: \inf_{\pi\in\Pi}W(Q,\pi)\geq \hat{\beta}\}$ is convex, and minimizing the convex utility loss function $Q\mapsto \mathbb{E}_{X\sim P_X, B\sim Q}\bigl[\mathcal{K}(B;\mathcal{M})\bigr]$ over this set remains a convex program in $Q$.
Meanwhile, the Follower's problem $\inf_{\pi\in\Pi}W(Q,\pi)$ is convex in $\pi$ for any fixed $Q$.
Thus, the Stackelberg game reduces to a single-level convex optimization in $Q$, with the inner decoder problem convex in $\pi$.

Proposition \ref{prop:Stackelberg_convergence} shows that the Stackelberg equilibrium perturbation rule solves (\ref{eq:general_OPT_RPAC}).

\begin{proposition}\label{prop:Stackelberg_convergence}
    Let $(Q^*, \pi^*)$ be a Stackelberg equilibrium satisfying (\ref{eq:stackelberg_equilibrium}) for any given $\hat{\beta}$.
    Then, $Q^*$ solves (\ref{eq:general_OPT_RPAC}) with privacy budget $\hat{\beta}$.
    In addition, in any Stackelberg equilibrium $(Q^*, \pi^*)$, $\pi^*=\pi^*(Q^*)$ is unique.
\end{proposition}

Algorithm~\ref{alg:SR_PAC} provides a Monte-Carlo-based approach to solve the Stackelberg equilibrium (\ref{eq:stackelberg_equilibrium}).
By Monte Carlo sampling, the algorithm trains the decoder by minimizing reconstruction loss on perturbed data, allowing it to adapt to the current noise distribution.
It then updates the perturbation rule by minimizing utility loss subject to the privacy constraint, implemented via a penalty term that drives the privacy cost toward the target budget.
For scalability, the online extended version \cite{zhang2025breaking} (Appendix~\ref*{app:sliced_rpac}) also presents two variants, \textit{Sliced R-PAC Privacy} and \textit{Sliced SR-PAC} algorithm, based on \textit{sliced mutual information} \cite{goldfeld2021sliced}.
The online extended version \cite{zhang2025breaking} also provides finite-sample and approximate-optimization error analyses for the Follower (Appendix~\ref*{app:finite_errors}).

\section{Properties of SR-PAC Privatization}

This section presents some important properties of SR-PAC.

\subsection{Anisotropic Noise Perturbation}

The Auto-PAC perturbs the mechanism using \textit{anisotropic} Gaussian noise as much as needed in each direction of the output. 
This direction-dependent noise addition yields better privacy-utility tradeoffs than isotropic perturbation.
SR-PAC also supports anisotropic perturbation under Assumption \ref{assp:anisotropic}.

\begin{assumption}\label{assp:anisotropic}
   For an arbitrary deterministic mechanism $\mathcal{M}$, we assume the following.
    \begin{itemize}
        \item[(i)] Every $Q\in \Gamma$ is log-concave.
        \item[(ii)] For any orthonormal direction $w\in \mathbb{R}^{d}$, $\langle \mathcal{M}(X), w\rangle$ is non-degenerate.
        \item[(iii)] The utility function $\mathcal{K}$ is radial (depends only on $\|B\|_{2}$) and strictly convex in the eigenvalues of covariance matrix $\Sigma_{Q}$ of $Q$. For example, $\kappa(B) = \|B\|^2_{2}$.
        \item[(iv)] There exist orthonormal $u, v\in\mathbb{R}^{d}$ such that the marginal entropy gain per unit variance along $u$ exceeds that along $v$. That is, for any $\sigma^2>0$, $\frac{\partial}{\partial \sigma^2_{u}} \mathcal{H}(X|Z_u)|_{\sigma^{2}} > \frac{\partial}{\partial\sigma^2_{v}} \mathcal{H}(X|Z_v)|_{\sigma^2},$
        where $Z_w = \mathcal{M}_{w}(X)+B_{w}$, with $A_{w}(X)=\langle A(X), w\rangle$ for $A\in\{\mathcal{M}, B\}$, $w\in\{u,v\}$.
    \end{itemize}
\end{assumption}

Assumption 1 ensures that SR-PAC's optimization is convex and admits a genuinely anisotropic solution: requiring each noise distribution in $\Gamma$ to be log-concave makes the feasible set convex and tractable; non-degeneracy of $< \mathcal{M}(X),w>$ for every unit vector $w$ guarantees that every direction affects information leakage; a strictly convex, radial utility $K$ yields a unique cost-to-noise mapping; and the existence of two orthonormal directions whose marginal entropy gain per unit variance differs implies that allocating noise unevenly strictly outperforms isotropic noise.

\begin{proposition}\label{prop:anisotropic}
    Under Assumption \ref{assp:anisotropic}, any Stackelberg equilibrium perturbation rule $Q^*$ is anisotropic. That is, its covariance matrix $\Sigma_{Q^*}$ satisfies 
    \[
    r_{\max}(\Sigma_{Q^*}) > r_{\min}(\Sigma_{Q^*}),
    \]
    where $r_{\max}(\Sigma_{Q^*})$ and $r_{\min}(\Sigma_{Q^*})$ are the maximum and the minimum eigenvalues of $\Sigma_{Q^*}$.
\end{proposition}

Proposition \ref{prop:anisotropic} demonstrates that SR-PAC allocates noise exclusively to privacy-sensitive directions, with high-leakage dimensions receiving proportionally more noise than low-leakage dimensions. This targeted approach achieves desired privacy levels with minimal total perturbation, preserving task-relevant information with reduced noise.

\subsection{Directional-Selectivity of SR-PAC}
\label{sec:dirsel-opt}

Let $Z$ be a $d$-dimensional real-valued \textit{output vector} produced by a
deterministic mechanism $\mathcal M(X)$. 
Throughout we assume
$\Sigma_Z\succ0$ and finite differential entropy $\mathcal{H}(Z)$.
For any application, let
$S_{\mathrm{task}}\subseteq\mathbb{R}^{d}$ denote a practitioner‑chosen
\textit{task‑critical sub‑space} (the directions whose preservation
matters most) and write
$\Pi_{\mathrm{task}}$ for the orthogonal projector onto it.

\noindent\textbf{Classification tasks. }
In what follows we illustrate the theory with multi‑class
classification, where $Z$ is the \textit{logit} vector,
$\hat y=\arg\max_i Z_i$, and $S_{\mathtt{lab}}
   \equiv\operatorname{span}\{e_\ell-e_j:\,j\neq\ell\}$, where $\mathtt{lab}$ means "label".
Let $\Pi_{\mathtt{lab}}$ be the projector onto $S_{\mathtt{lab}}$.
The analysis for a general $S_{\mathrm{task}}$ is identical after
replacing $\mathtt{lab}$ by $\mathrm{task}$.

For any privacy budget $0<\beta< \mathcal{H}(Z)$, consider $Q^*$ that solves
\[
\inf_{Q: \mathtt{MI}(Z;Z+B)=\beta}\mathbb{E}[\|B\|^2_2].
\]
For every unit vector $w$, let $g(w)\equiv \frac{1}{2}\mathtt{mmse}(\langle Z,w\rangle)$, where $\mathtt{mmse}(\langle Z,w\rangle)\equiv\mathbb{E}\Bigl[\bigl\langle Z,w\bigr\rangle
\;-\;\mathbb{E}\bigl[\langle Z,w\rangle |Y\bigr]
\Bigr]^{2}$ is the \textit{minimum mean‐squared error} of estimating the scalar random variable $\langle Z,w\rangle$ from the noisy observation $Y = Z + B$.

\begin{proposition}\label{prop:directional}
Suppose $\mathcal{H}(Z)$ is finite.
    Fix any $0<\beta<\mathcal{H}(Z)$. The following holds.
    \begin{itemize}
        \item[(i)] Let $\mathcal{N}(0, \Sigma_{\mathrm{PAC}})$ be the Gaussian noise distribution used by the Auto-PAC such that $\mathtt{LogDet}(Z, B_{\mathrm{PAC}}) = \beta$.
    If $Z$ is non-Gaussian, then $\mathbb{E}_{Q^{*}}[\|B\|^2_2]< \mathbb{E}[\|B_{\mathrm{PAC}}\|^2_2]$.
    
    %

    \item[(ii)] Suppose $\sup_{v\in S_{\mathtt{lab}},\|v\|=1} g(v)<\inf_{w\perp S_{\mathtt{lab}},\|w\|=1} g(w)$.
    Let $\beta_{\mathtt{lab}}\equiv \frac{1}{2}\int_{w\bot S_{\mathtt{lab}}}g(w) d\sigma^2_{w}$ denote the largest privacy budget that can be satisfied using noise supported entirely on $S^{\perp}_{\mathtt{lab}}$.
  Then, for every $\beta\leq \beta_{\mathtt{lab}}$, we have $\Pi_{\mathtt{lab}}B^{*}=0\ \text{ a.s.},\;
       \arg\max_{i}(Z_{i}+B^{*}_{i})=\hat y\ \text{ a.s.}$
    \end{itemize}
\end{proposition}

In Proposition \ref{prop:directional}, part (i) shows that SR-PAC always uses strictly less noise magnitude than any Auto-PAC (regardless of how anisotropic the Auto-PAC noise covariance may be) because Auto-PAC treats $Z$ as Gaussian and thus overestimates the required variance when $Z$ is non-Gaussian.
Part (ii) demonstrates that, under the natural ordering of directional sensitivities, SR-PAC allocates its noise budget exclusively in directions orthogonal to the label sub-space until a critical threshold $\beta_{\mathtt{lab}}$ is reached.
In practice, this means SR-PAC perturbs only "utility-harmless" dimensions first, preserving the predicted class and concentrating protection where it is most needed, thereby outperforming Auto-PAC in any scenario where certain directions leak more information than others.

\subsection{Sensitivity to $\beta$}

Sensitivity to the privacy parameter $\beta$ is crucial for predictable and accurate control of privacy-utility trade-off.
Let $\mathtt{Priv}_{\beta}$ and $\mathtt{Util}_{\beta}$, respectively, denote the sensitivities of privacy and utility (for certain measures).
If $\mathtt{Priv}_{\beta} = 1$, then any infinitesimal increase $\Delta\beta$ in the privacy budget raises the true mutual information $\mathtt{MI}(X;Y)$ by exactly $\Delta\beta$.
Thus, no part of the privacy budget is "wasted" or "over-consumed".
By contrast, if $\mathtt{Priv}_{\beta}< 1$, then increasing $\beta$ may force additional noise without achieving the full allowed leakage; and if $\mathtt{Priv}_{\beta}>1$, then increasing the budget by $\Delta\beta$ can increase the true leakage by more than $\Delta\beta$. 
In particular, if the mechanism is calibrated to be tight at $\beta$ (i.e., $\mathtt{MI}(X;Y)=\beta$), then it may become over-budget, i.e., $\mathtt{MI}(X;Y)>\beta + \Delta\beta$. 
%
%
Similarly, if $\mathtt{Util}{\beta}$ is high, then an infinitesimal increase $\Delta\beta$ in the privacy budget yields a large improvement in utility; if $\mathtt{Util}_{\beta}$ is low, the same increase yields a small improvement, indicating inefficient conversion of the privacy budget into utility gains.

Let $V_{\mathrm{SR}}(\beta) \equiv \min_{Q: \mathtt{MI}(X;\mathcal{M}(X)+B) \leq \beta} \mathbb{E}_Q\big[\|B\|_2^2\big]$
be the optimal noise-power curve attained by SR-PAC, and let $\mathtt{MI}_{\mathrm{SR}}(\beta)$ as the corresponding true mutual information attained by SR-PAC.
Let $V_{\mathrm{PAC}}(\beta)\equiv \mathrm{tr}(\Sigma_{B_{\mathrm{PAC}}}(\beta))$, where $Q(\beta)=\mathcal{N}(0,\Sigma_{B_{\mathrm{PAC}}}(\beta))$ solves $\mathtt{LogDet}(\mathcal{M}(X), B_{\mathrm{PAC}}) =\beta$.
In addition, let
$\displaystyle \mathtt{MI}_{\mathrm{PAC}}(\beta) \equiv \beta - \mathtt{Gap}_{\mathtt{d}}(Q(\beta)),$
where $\mathtt{Gap}_{\mathtt{d}}(Q)=\mathtt{D}_{\mathrm{KL}}(P_{\mathcal{M},B}\|\widetilde{Q}_{\mathcal M})$ with $B\sim Q$, and $\widetilde{Q}_{\mathcal M}$ given by (\ref{eq:tilde_Q}).
Define $\mathtt{Priv}^{\mathrm{SR}}_{\beta}\equiv\frac{d}{d\beta}\mathtt{MI}_{\mathrm{SR}}(\beta)$, $\mathtt{Priv}^{\mathrm{PAC}}_{\beta} \equiv \frac{d}{d \beta} \mathtt{MI}_{\mathrm{PAC}}(\beta)$, $\mathtt{Util}^{\mathrm{SR}}_{\beta}\equiv \frac{d}{d\beta}(-V_{\mathrm{SR}}(\beta))$, and $\mathtt{Util}^{\mathrm{PAC}}_{\beta}\equiv \frac{d}{d\beta}(-V_{\mathrm{PAC}}(\beta))$.

\begin{theorem}\label{thm:speed_general}
    For any data distribution $\mathcal{D}$, let $\mathcal{M}$ be an arbitrary deterministic mechanisms such that $\mathcal{M}(X)$ is non-Gaussian with $\Sigma_M \succ 0$.
    The following holds.
    \begin{itemize}
        \item[(i)] $\mathtt{Priv}^{\mathrm{PAC}}_{\beta} \leq\mathtt{Priv}^{\mathrm{SR}}_{\beta} = 1$, with strict inequality for non-Gaussian $\mathcal{M}(X)$.

        \item[(ii)] $\mathtt{Util}^{\mathrm{SR}}_{\beta}\geq \mathtt{Util}^{\mathrm{PAC}}_{\beta}$, with equality only for Gaussian $\mathcal{M}(X)$.
    \end{itemize}
\end{theorem}

Theorem \ref{thm:speed_general} proves that SR-PAC with arbitrary noise distributions achieves:
(i) \textit{Exact leakage-budget alignment} ($\mathtt{Priv}^{\mathrm{SR}}_{\beta} = 1$),
(ii) \textit{Stricter utility decay} for Auto-PAC ($\mathtt{Util}^{\mathrm{SR}}_{\beta}\geq \mathtt{Util}^{\mathrm{PAC}}_{\beta}$).
This holds for all non-Gaussian $\mathcal{M}(X)$ when increasing privacy strength (i.e., $\beta$ decreasing).
A robustness analysis under finite-sample calibration and optimization effects is given in Appendix~\ref*{app:implemented_curves} of \cite{zhang2025breaking}.



\begin{figure*}[!t]
    \centering
    \subfloat[\label{fig:c10_acc} CIFAR-10: Accuracy vs.\ $\beta$]{%
        \includegraphics[width=0.24\textwidth]{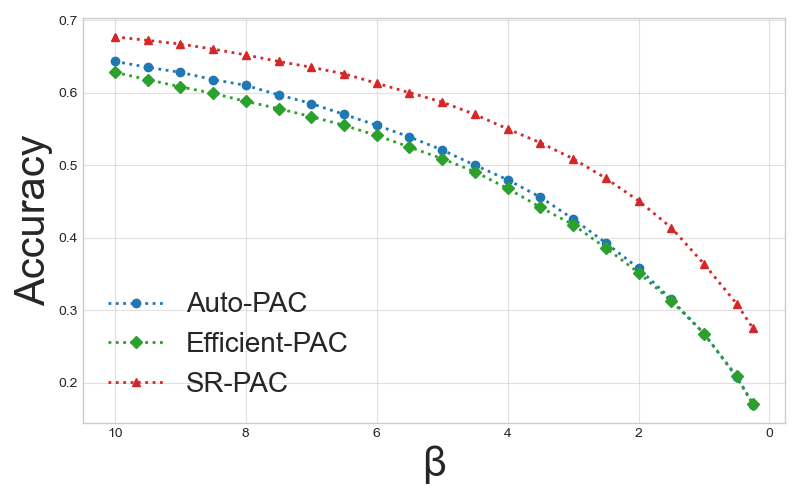}}%
    \hfill
    \subfloat[\label{fig:c100_acc} CIFAR-100: Accuracy vs.\ $\beta$]{%
        \includegraphics[width=0.24\textwidth]{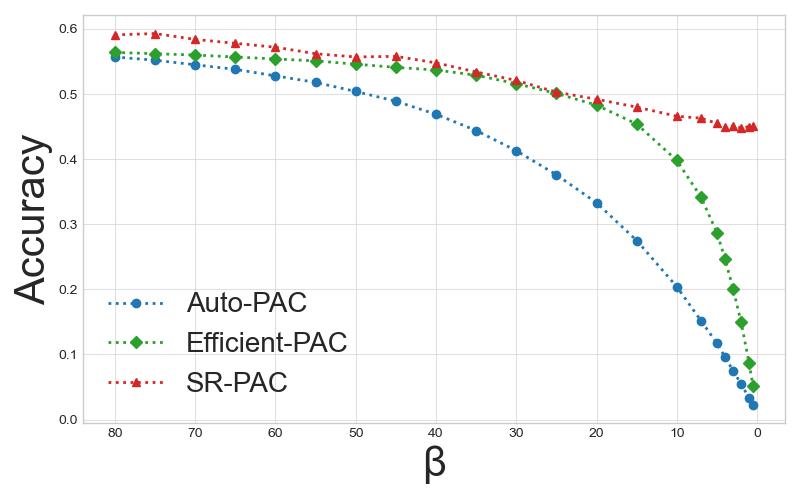}}%
    \hfill
    \subfloat[\label{fig:mnist_acc} MNIST: Accuracy vs.\ $\beta$]{%
        \includegraphics[width=0.24\textwidth]{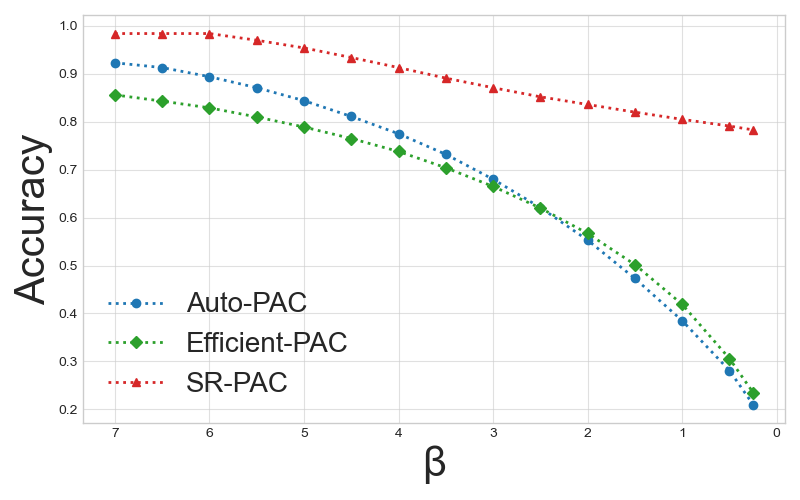}}%
    \hfill
    \subfloat[\label{fig:ag_acc} AG-News: Accuracy vs.\ $\beta$]{%
        \includegraphics[width=0.24\textwidth]{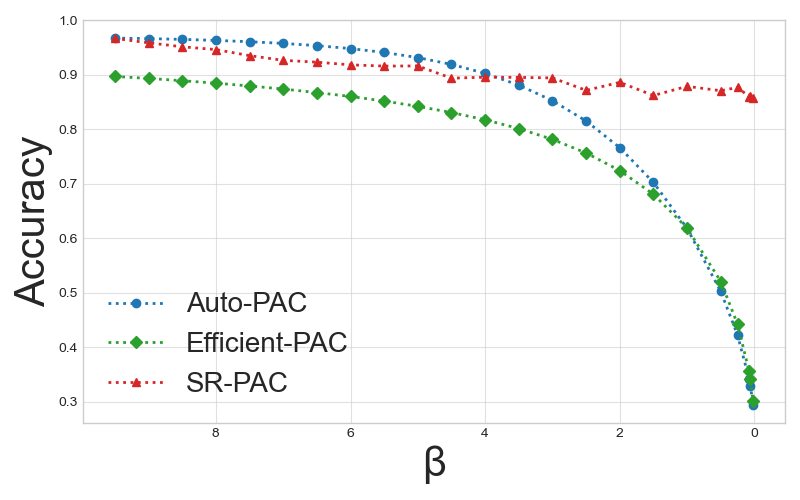}}%


    \subfloat[\label{fig:c10_noise} CIFAR-10: Noise Magnitude vs.\ $\beta$]{%
        \includegraphics[width=0.24\textwidth]{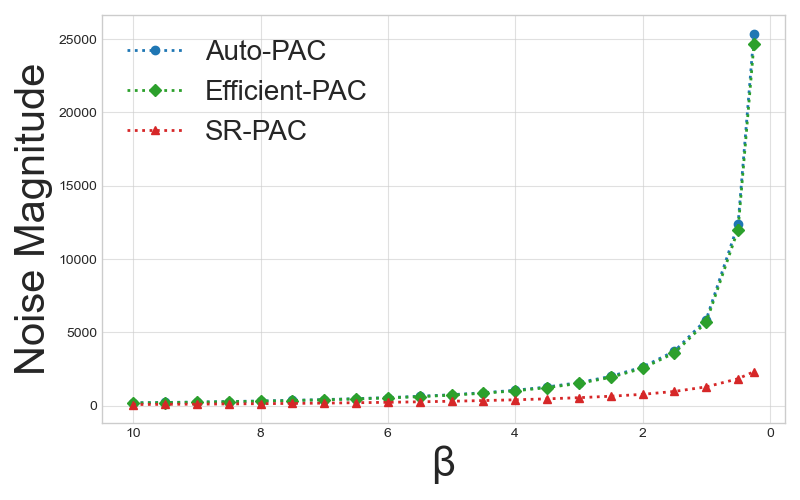}}%
    \hfill
    \subfloat[\label{fig:c100_noise} CIFAR-100: Noise vs.\ $\beta$]{%
        \includegraphics[width=0.24\textwidth]{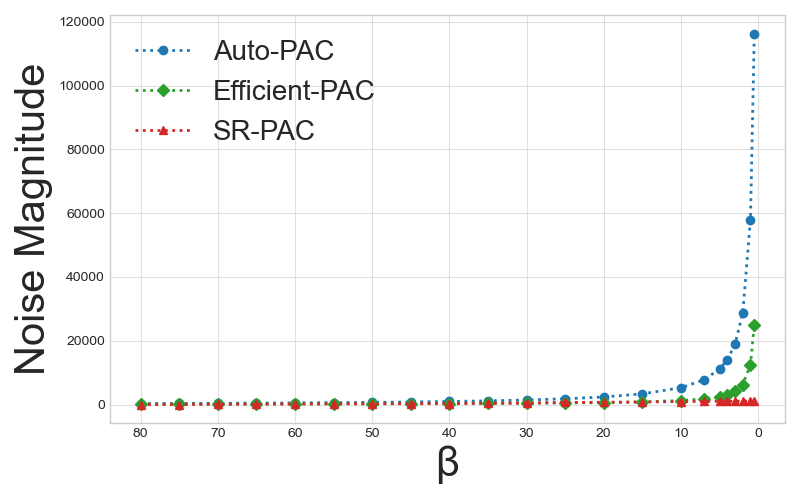}}%
    \hfill
    \subfloat[\label{fig:mnist_noise} MNIST: Noise vs.\ $\beta$]{%
        \includegraphics[width=0.24\textwidth]{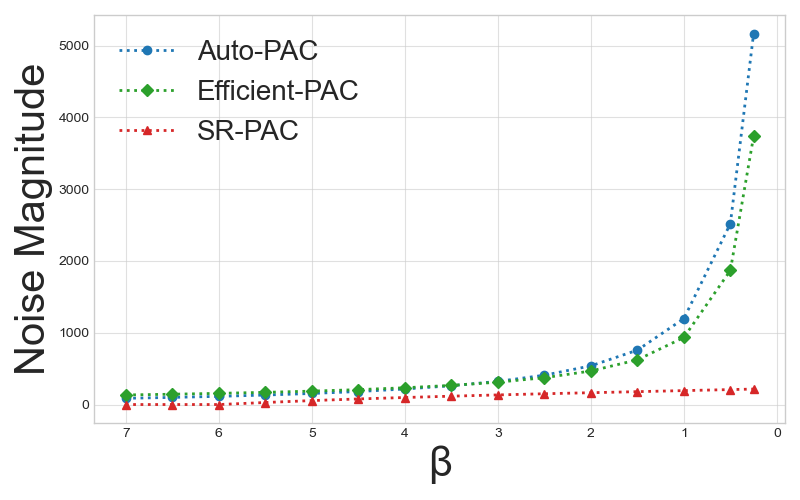}}%
    \hfill
    \subfloat[\label{fig:ag_noise} AG-News: Noise vs.\ $\beta$]{%
        \includegraphics[width=0.24\textwidth]{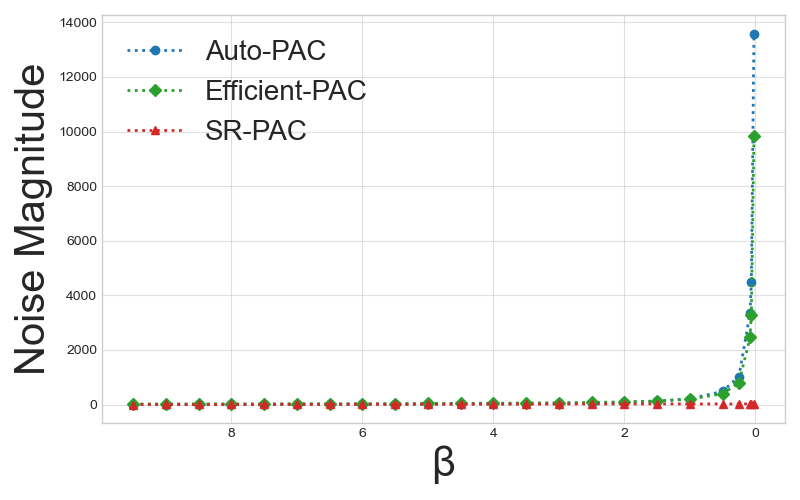}}%


    \subfloat[\label{fig:c10_match} CIFAR-10: Budgets vs. \ $\beta$]{%
        \includegraphics[width=0.24\textwidth]{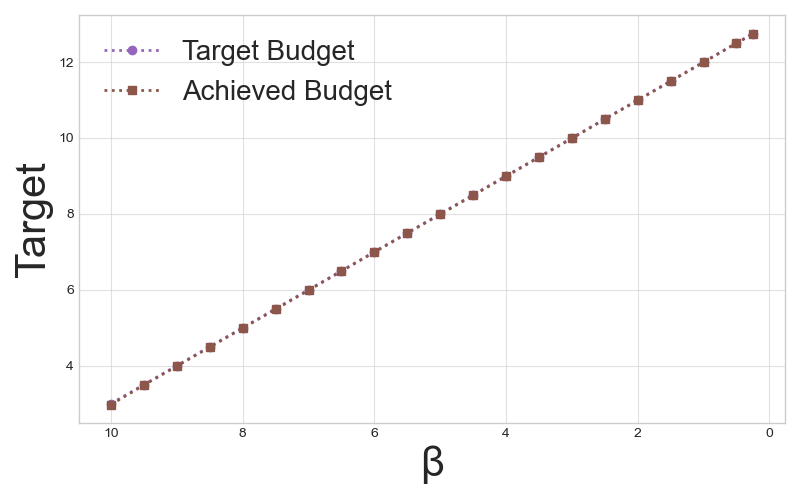}}%
    \hfill
    \subfloat[\label{fig:c100_match} CIFAR-100: Budgets vs. \ $\beta$]{%
        \includegraphics[width=0.24\textwidth]{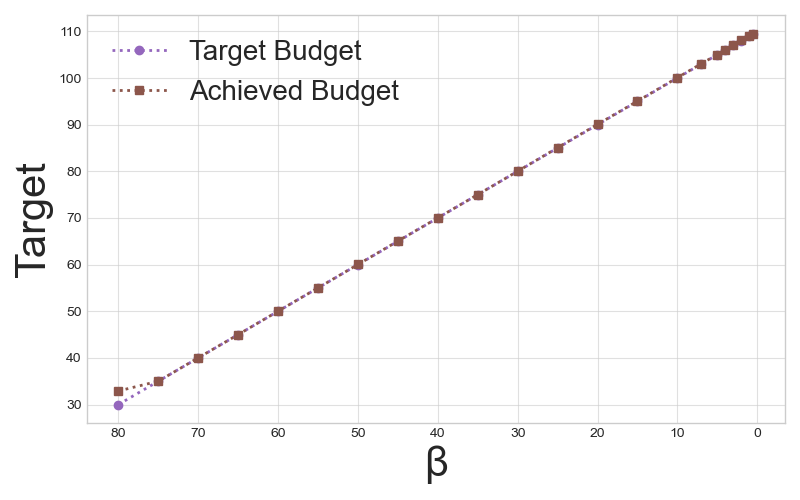}}%
    \hfill
    \subfloat[\label{fig:mnist_match} MNIST: Budgets vs.\ $\beta$]{%
        \includegraphics[width=0.24\textwidth]{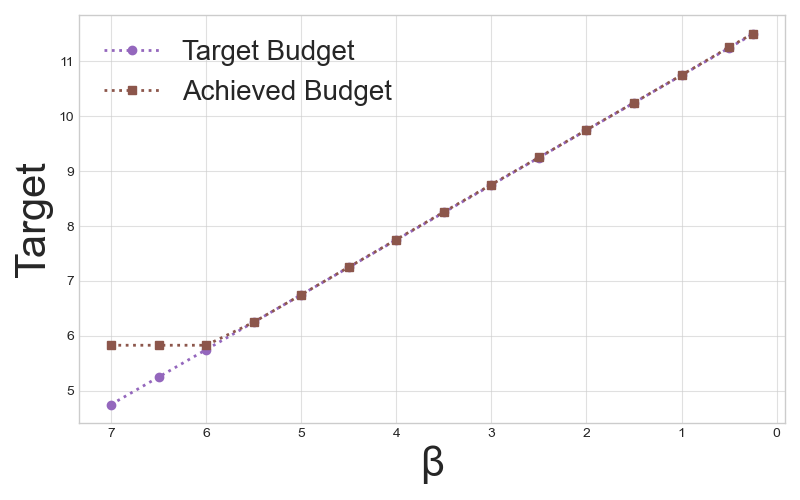}}%
    \hfill
    \subfloat[\label{fig:ag_match} AG-News: Budgets vs.\ $\beta$]{%
        \includegraphics[width=0.24\textwidth]{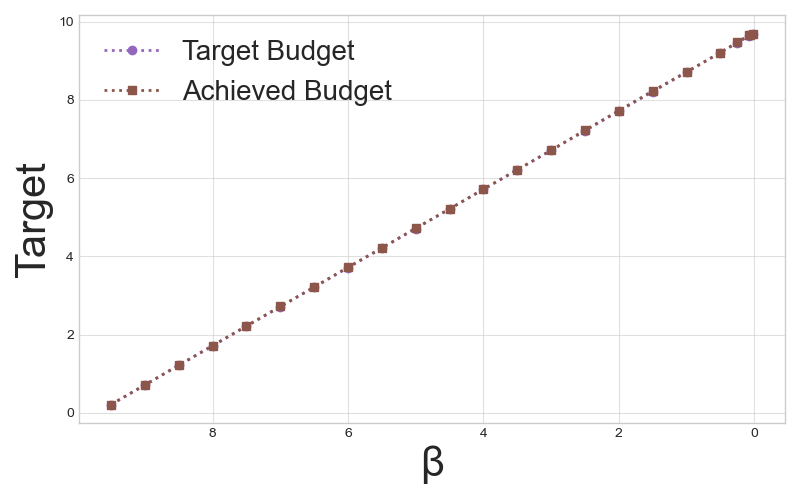}}%

    \caption{Empirical comparisons of SR-PAC, Auto-PAC (Algorithm \ref{alg:PAC_original}), and Efficient-PAC (Algorithm~\ref*{alg:PAC_alg_original} in \cite{zhang2025breaking}) on CIFAR-10, CIFAR-100, MNIST, and AG-News as $\beta$ varies. Each column corresponds to one dataset; within each column, the three panels report (top) classification accuracy of the perturbed model versus the target budget $\beta$, (middle) the average noise magnitude $\mathbb{E}[\lVert B\rVert_{2}^{2}]$ used by each method, and (bottom) the 
    "target versus achieved" privacy budget (conditional entropy) for our SR-PAC. }
    \label{fig:all_metrics}
\end{figure*}

\subsection{Composition}

Graceful composition properties in privacy definitions like DP make privacy loss quantifiable under multiple operations on datasets.
This enables modular system design: each component can be tuned to a local privacy–utility trade-off, while composition rules provide an explicit bound on the overall (global) privacy risk.
Consider $k$ mechanisms $\mathcal{M}_{1}, \mathcal{M}_{2},\dots \mathcal{M}_{k}$, where each $\mathcal{M}_{i}(\cdot, \theta_{i}):\mathcal{X}\mapsto \mathcal{Y}_{i}$ with $\theta_{i}\in\Theta_{i}$ as the random seed.
Let $\vec{\mathcal{Y}}=\prod^{k}_{i=1}\mathcal{Y}_{i}$ and let $\vec{\Theta} = \prod^{k}_{i=1}\Theta_{i}$.
The composition $\overrightarrow{\mathcal{M}}(\cdot, \vec{\theta}):\mathcal{X}\mapsto \vec{\mathcal{Y}}$ is defined as $\overrightarrow{\mathcal{M}}(X, \vec{\theta}) = \left(\mathcal{M}_{1}(X, \theta_{1}), \dots, \mathcal{M}_{k}(X, \theta_{k})\right).$
PAC Privacy composes gracefully \cite{xiao2023pac}. 
In particular, for independent mechanisms applied to the same dataset, mutual information bounds compose additively: if each $\mathcal{M}_{i}$ is PAC Private with bound $\beta_{i}$, then $\overrightarrow{\mathcal{M}}$ has bound $\sum^{k}_{i=1}\beta_{i}$.

R-PAC Privacy also enjoys additive composition with respect to conditional entropy bounds. Suppose each mechanism $\mathcal{M}_{i}$ is R-PAC private with conditional entropy lower bound $\hat{\beta}_{i}$. By (\ref{eq:MI_def_CE}), this implies that $\mathcal{M}_{i}$ is PAC private with privacy budget $\beta_{i} = \mathcal{H}(X) - \hat{\beta}_{i}$. 
Then, by Theorem 7 of \cite{xiao2023pac}, the composition $\overrightarrow{\mathcal{M}}(X, \vec{\theta})$ is PAC private with total mutual information upper bounded by $\sum_{i=1}^{k} (\mathcal{H}(X) - \hat{\beta}_{i})$. Equivalently, the composition $\overrightarrow{\mathcal{M}}(X, \vec{\theta})$ is R-PAC private with overall conditional entropy lower bounded by $\sum_{i=1}^{k} \hat{\beta}_{i} - (k - 1)\mathcal{H}(X)$.
%

However, this additive composition property for mutual information yields conservative aggregated privacy bounds \cite{xiao2023pac}, and utility degradation compounds when each mechanism $\mathcal{M}_i$ uses conservative privacy budgets $\beta_i$. 
To address this limitation, we employ an optimization-based approach within the SR-PAC framework to compute tighter conditional entropy bounds. 
Consider $k$ mechanisms $\mathcal{M}_{1}, \mathcal{M}_{2},\dots, \mathcal{M}_{k}$, where each $\mathcal{M}_{i}$ is privatized by the perturbation rule $Q_i$ to satisfy R-PAC privacy with bounds $\hat{\beta}_{i}$. The Leader designs these perturbation rules $Q_{1},..., Q_{k}$, while the Follower finds the optimal decoder for the joint composition $\overrightarrow{\mathcal{M}}(X, \vec{\theta}) = (\mathcal{M}_1(X), \dots, \mathcal{M}_k(X))$: 
$\displaystyle \inf_{\pi\in\Pi} W(\pi;\overrightarrow{\mathcal{M}}) \equiv \mathbb{E}_{X\sim\mathcal{D}}\left[-\log\pi(X|\overrightarrow{\mathcal{M}}(X),\vec{\theta})\right].$
This game-theoretic formulation allows for tighter privacy-utility trade-offs in composed systems by optimizing the joint privatization strategy.
This joint SR-PAC formulation also extends to adaptive composition, where each $Q_i$ (and the corresponding decoder update) may be chosen sequentially based on previously released privatized outputs, in the same spirit as the adaptive composition procedure of PAC Privacy \cite{xiao2023pac}.

\section{Experiments}

\begin{figure}[t] %
  \centering
  \begin{subfigure}{0.495\columnwidth}
    \includegraphics[width=\linewidth]{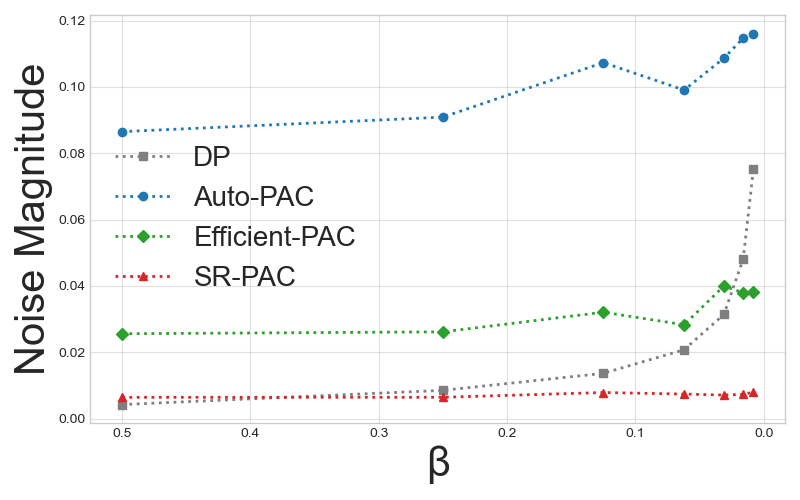}
    \caption{Iris}
    \label{fig:Iris}
  \end{subfigure}
  \hfill
  \begin{subfigure}{0.495\columnwidth}
    \includegraphics[width=\linewidth]{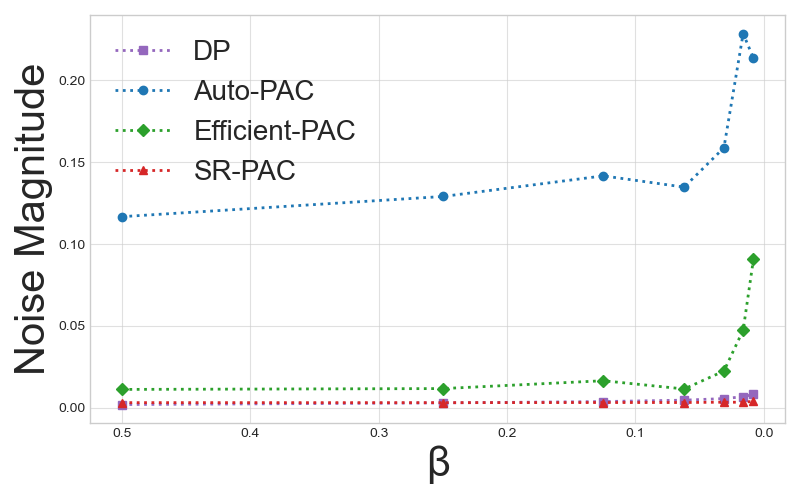}
    \caption{Rice}
    \label{fig:Rice}
  \end{subfigure}
  \caption{Empirical comparisons of DP, Auto-PAC, Efficient-PAC, and SR-PAC on mean estimations, using Iris and Rice datasets, in terms of average noise magnitude $\mathbb{E}[\|B\|^2_2]$. All the numerical values are shown in Tables \ref{tab:Iris_value} and \ref{tab:Rice_value}.}
  \label{fig:DP_comparisons}
\end{figure}

We conduct two sets of experiments to evaluate our approach. First, we compare SR-PAC against Auto-PAC and Efficient-PAC (Appendix~\ref*{app:efficient_pac} in \cite{zhang2025breaking}) using CIFAR-10 \cite{Krizhevsky09learningmultiple}, CIFAR-100 \cite{Krizhevsky09learningmultiple}, MNIST \cite{lecun2010mnist}, and AG-News \cite{zhang2015character} datasets, with results presented in Section \ref{sec:experiment_PACs}. 
We use (R-)PAC to refer to the family of SR-PAC, Auto-PAC, and Efficient-PAC.
Second, we extend this comparison to include DP by equalizing optimal posterior success rates of membership inference (Appendix \ref{app_sub:fair_MIA}) across all methods, making their privacy budgets comparable. For this comparison, we use Iris \cite{fisher1988iris} and Rice \cite{cinar2019classification} datasets, with results shown in Section \ref{S:expdp}. All experiments focus on output perturbation.
Appendix~\ref*{app:more_on_experiments} in \cite{zhang2025breaking} provides more details.

\noindent\textbf{CIFAR-10 and Base Classifier.} 
CIFAR-10 and base classifier. We evaluate on CIFAR-10 ($32\times 32$ RGB, 10 classes) with standard per-channel normalization (mean $0.5$, std $0.5$). 
The unperturbed classifier is a small CNN with two Conv–ReLU–MaxPool blocks ($2\times 2$ pooling; 32 and 64 channels), followed by a 128-unit fully connected ReLU layer and a 10-logit output layer. 
It is trained with cross-entropy loss, and predicts the argmax logit at inference. 
The unperturbed classifier achieves $0.7181\pm 0.0088$ accuracy.

\noindent\textbf{CIFAR-100 and Base Classifier.} We evaluate on CIFAR-100 ($32\times 32$ RGB, 100 classes) with standard per-channel normalization (mean 0.5, std 0.5). The unperturbed classifier is a deeper CNN with three convolutional blocks (two $3\times 3$ Conv–BatchNorm–ReLU layers per block, followed by $2\times 2$ max-pooling), with channel widths 64/128/256. A three-layer MLP head (4096→512→256→100) with ReLU and dropout (0.5) outputs 100 logits, and prediction is by argmax. The unperturbed classifier achieves $0.5914\pm 0.0090$ accuracy.

\noindent\textbf{MNIST dataset and Base Classifier.}
We evaluate on MNIST (60,000 train / 10,000 test) consisting of $28\times 28$ grayscale digit images.
Each image is loaded as a $1\times 28\times 28$ tensor and normalized per channel (mean 0.1307, std 0.3081).
The unperturbed classifier is a CNN with two Conv2d$\to$BatchNorm$\to$ReLU$\to$MaxPool ($2\times 2$) blocks (channels $1\to32\to64$), yielding a $64\times 7\times 7$ feature map, followed by a two-layer fully connected head (128 units with ReLU+Dropout, then 10 logits).
At inference it outputs a 10-dimensional logit vector and predicts by argmax.
The unperturbed classifier achieves $0.9837$ accuracy.


\noindent\textbf{AG-News dataset and Base Classifier.}
We evaluate on AG-News with 120,000 training and 7,600 test articles evenly split across four classes (World, Sports, Business, Sci/Tech), i.e., 30,000 training and 1,900 test examples per class. Each example's title and description are concatenated, lowercased, and tokenized by whitespace (truncated/padded to 64 tokens). We build a 30,000-word vocabulary from the training split, map tokens to indices (out-of-vocabulary as 0), and feed the indices into an \texttt{nn.EmbeddingBag} layer (embedding size 300, mean-pooling) to obtain a 300-dimensional document vector. A two-layer MLP head ($300\to256$ with ReLU and 0.3 dropout, then $256\to4$) produces a 4-dimensional logit vector, and prediction is by argmax. 
The unperturbed mechanism achieves $0.9705$ accuracy.

Recall that $\beta$ upper-bounds $\mathtt{MI}\bigl(X;\mathcal{M}(X)+B\bigr)$, and SR-PAC enforces the equivalent constraint
$\mathcal{H}\bigl(X \mid \mathcal{M}(X)+B\bigr)\ge \widehat{\beta}=\mathcal{H}(X)-\beta$.
Although $\mathcal{H}(X)$ is unknown, we estimate for the purpose of evaluation to verify the tightness of the privacy bounds.
Let $\mathtt{MI}_{0}=\mathtt{MI}\bigl(X;\mathcal{M}(X)\bigr)$. By data processing inequality,
$\mathtt{MI} \bigl(X;\mathcal{M}(X)+B\bigr)\le \mathtt{MI}_{0}$ for any independent $B$, so the feasible budgets are $0<\beta\le \mathtt{MI}_{0}$ and this interval is common to Auto-PAC, Efficient-PAC, and SR-PAC.
At $\beta=\mathtt{MI}_{0}$, the optimal choice is $B=0$, and all methods coincide at the noiseless accuracy. This shared endpoint and feasible domain ensure that comparisons at a common target $\beta$ are well-defined even without the exact $\mathcal{H}(X)$.
Moreover, reparameterizing by achieved mutual information preserves the endpoint and domain, and, together with the small budget errors observed in panels (i–l), does not affect our empirical ordering.
Under additive $\ell_{2}$ output noise, the ordering by total noise magnitude $\mathbb{E} \left[\|B\|_2^2\right]$ coincides with the ordering by accuracy, consistent with the $\ell_{2}$-based behavior reported in prior work.
Hence, the accuracy– and noise–vs.–$\beta$ panels convey the same conclusion in our experiments.

\subsection{(R-)PAC Comparison}\label{sec:experiment_PACs}

For each dataset and its pretrained base classifier $\mathcal{M}$, we plot (i) the test accuracy of the perturbed model as a function of $\beta$, (ii) the average noise magnitude $\mathbb{E}[\lVert B\rVert_{2}^{2}]$ required to achieve each $\beta$, and (iii) SR-PAC's ability to hit the target $\mathcal{H}(X) - \beta$.

\textbf{Accuracy vs. $\beta$ (a–d of Figure \ref{fig:all_metrics}).}
As $\beta$ decreases (moving right), privacy increases and all methods lose test accuracy.
For large $\beta$ (near the no-privacy case), all three algorithms attain accuracies close to the noiseless model.
As $\beta$ tightens, the SR-PAC curve remains strictly above the Auto-PAC and Efficient-PAC curves across datasets.
On \textbf{CIFAR-10} and \textbf{CIFAR-100}, Auto-PAC and Efficient-PAC are visibly separated from each other (not merely from SR-PAC), reflecting their different Gaussian calibrations.
On \textbf{MNIST} and \textbf{AG-News}, the three methods cluster near the top for larger $\beta$, but SR-PAC retains a measurable accuracy edge at matched $\beta$.

\textbf{Noise magnitude vs. $\beta$ (e–h of Figure \ref{fig:all_metrics}).}
As $\beta$ decreases, each algorithm must add more noise, so all three curves rise.
Across all datasets, SR-PAC uses the smallest $\mathbb{E} \bigl[\|B\|^{2}_{2}\bigr]$ at each $\beta$.
Auto-PAC and Efficient-PAC both overshoot—they inject more noise than SR-PAC at matched $\beta$—and on CIFAR-100, MNIST and AG-News, they diverge from each other as well.

The empirical ordering in both accuracy and noise magnitude matches Theorem~\ref{thm:speed_general}.
Moreover, Figure~\ref{fig:all_metrics} (c–d, g–h) exhibits the behavior predicted by Proposition~\ref{prop:directional} on \textbf{MNIST} and \textbf{AG-News}:
for $\beta \le \beta_{\mathtt{lab}}$, SR-PAC allocates noise predominantly in directions (approximately) orthogonal to the label subspace, preserving the predicted class over a wide budget range.
Concurrently, its total noise remains substantially smaller than Auto-PAC and Efficient-PAC, whose conservative Gaussian calibrations overestimate the required variance on heavy-tailed (non-Gaussian) logits.

\textbf{Budgets vs. $\beta$ (i–l of Figure \ref{fig:all_metrics}).}
These panels plot SR-PAC's target privacy budgets in terms of mutual-information bounds $\beta$ (horizontal) against the achieved empirical conditional-entropy budget (vertical).
In every dataset, the red points lie tightly along the $y=x$ line, confirming that SR-PAC solves its follower problem with high accuracy and enforces the desired privacy level with negligible budget error.
This provides a reliable, data-driven guarantee that the privacy constraint is satisfied.

\subsection{Comparison with Differential Privacy}
\label{S:expdp}

We calibrate DP and (R-)PAC to the same (optimal) posterior success rate for membership inference attacks, then compare their utility in terms of noise magnitudes (i.e., $\ell_2$-norm of the difference between original and perturbed outputs).
The base mechanism is a mean estimator.
Appendix~\ref{app_sub:fair_MIA} provides the conversions between (optimal) posterior success rates, DP parameters (DP-to-posterior mapping), and mutual information budgets (MI-to-posterior mapping).
Concretely, for DP we select $(\epsilon,\bar{\delta})$ that yields the target posterior bound via the DP-to-posterior mapping, and for (R-)PAC we choose the budget $\beta$ that yields the same posterior via the MI-to-posterior mapping; with subsampling rate $r=0.5$ we have prior $p=0.5$.
In each trial, we construct a membership vector $m\in\{0,1\}^P$ by i.i.d.\ Bernoulli$(0.5)$ draws, so the member count $S=\sum_i m_i$ is random.
We follow similar treatments for DP as Section 6.3 of \cite{sridhar2024pac}:
the DP baseline clips each row in $\ell_2$ to radius $C$, adds Gaussian noise to the clipped sum, and divides by $S$ to produce the privatized mean; (R-)PAC injects noise calibrated to $\beta$.
We report $\mathbb{E}\left[\|B\|_2^2\right]$ at matched posterior success rates. In our output-perturbed mean setting, this quantity equals the expected squared $\ell_2$ error of the released statistic.
Hence the ordering by $\mathbb{E}\left[\|B\|_2^2\right]$ is identical to the ordering by $\ell_2$ accuracy. 
Appendix~\ref{app_sub:fair_MIA} gives detailed DP vs. (R-)PAC discussion.

\textbf{Figure~\ref{fig:DP_comparisons}.}
On the Iris and Rice mean–estimation tasks, SR-PAC attains the smallest average noise magnitude
$\mathbb{E}\!\left[\lVert B\rVert_2^2\right]$ across privacy budgets $\beta$.
As $\beta$ decreases (stricter privacy), the noise required by Auto-PAC and Efficient-PAC rises much more steeply,
whereas SR-PAC grows gently; see the zoomed view in Fig.~\ref{fig:DPvs_SR_PAC} (Appendix~\ref*{app:more_on_experiments} in \cite{zhang2025breaking}).
The DP baseline remains well above SR-PAC and, at small budgets on Iris, also exceeds Efficient-PAC.
Appendix~\ref*{app:more_on_experiments} in \cite{zhang2025breaking} further reports empirical membership–inference results
for SR-PAC, DP, Auto-PAC, and Efficient-PAC on these privatized mechanisms.

Auto-PAC and Efficient-PAC allocate \textit{anisotropic} noise, but their shapes are task-agnostic and depend only on second-order structure, via covariance scaling (Auto-PAC) or eigen--allocation (Efficient-PAC).
In small-sample regimes (e.g., Iris and Rice), the covariance spectrum is noisy and often ill-conditioned, and these moment-based rules propagate that instability into the noise design, yielding conservative noise levels, especially for small $\beta$.
By contrast, SR-PAC enforces the conditional-entropy budget directly, leading to tighter budget implementation and lower required noise.
Empirically (Figure \ref{fig:DP_comparisons}), SR-PAC attains smaller average noise magnitudes across $\beta$ with smoother scaling.


\section{Conclusion}

This work introduced R-PAC Privacy, an enhanced framework that guarantees privacy beyond Gaussian assumptions while overcoming the conservativeness of existing PAC Privacy algorithms. Our SR-PAC algorithm casts the privacy–utility trade-off as a Stackelberg problem, efficiently using the privacy budget and learning data- and mechanism-specific anisotropic noise. Extensive experiments show that SR-PAC consistently attains tighter privacy guarantees and higher utility than prior approaches, providing a rigorous and practical foundation for scalable privacy assurance in complex applications.


\newpage
\section*{Ethical Considerations}

We propose Residual-PAC Privacy (R-PAC) and its privatization scheme Stackelberg Residual-PAC (SR-PAC) as a privacy protection framework. Our goal is to reduce the conservativeness of prior PAC Privacy mechanisms so that, for a fixed privacy budget, practitioners can achieve better utility without relaxing stated privacy guarantees. All experiments use standard, publicly available benchmark datasets. We do not collect new data, interact with live production systems or APIs, scrape data, or discover/disclose vulnerabilities.
We organize ethical considerations around the generic data-processing pipeline
\[
\textbf{Data} \rightarrow \textbf{Mechanisms} \rightarrow \textbf{Downstream Decision-Making},
\]
and assess benefits and harms using the Menlo Report principles \cite{dittrich2012menlo}: Beneficence, Respect for Persons, Justice, and Respect for Law and Public Interest.
Appendix~\ref*{app:ethical_trade_offs} in \cite{zhang2025breaking} provides ethical considerations about the trade-offs across different privacy standards.

\subsection*{Stakeholders}

\textbf{Data} are collections of records about people (e.g., medical, financial, behavioral). In our experiments, we use widely adopted public benchmarks, but in practice similar mechanisms could be deployed on sensitive real-world data. Relevant stakeholders include data subjects and dataset curators.

\noindent\textbf{Mechanisms} are data-driven systems (e.g., statistical analyses, ML models) operating under R-PAC or other frameworks (e.g., DP). Stakeholders include the privacy research community, data scientists/ML engineers, privacy/security teams, and auditors/defenders who pressure-test privacy claims.

\noindent\textbf{Downstream decision-making} comprises automated or human decisions relying on mechanism outputs (e.g., risk scoring, recommendation, decision support). Stakeholders include affected individuals (patients, applicants, platform users), as well as regulators, standards bodies, and advocacy organizations that rely on formal privacy statements.

\noindent\textbf{Cross-cutting societal stakeholders} (taxpayers, communities broadly subject to algorithmic systems, environmental stakeholders) are indirectly affected by how privacy-preserving data analysis becomes common practices.

\subsection*{Harms and Mitigations}

\textbf{Data stage.} The primary privacy risk arises from how data are used, not from their mere existence. R-PAC does not introduce new collection/scraping/linkage activity; rather, it calibrates and interprets residual privacy risk for data use in mechanisms. A key harm is \textit{misaligned expectations}: curators or deployers may choose parameters that are too weak for a given context, or may overgeneralize scenario-specific guarantees, creating a false sense of protection. 
We mitigate this by restricting experiments to public benchmarks and by emphasizing that guarantees are scenario-specific and depend on the chosen prior, model, and calibration.
Real deployments require context-aware parameter selection and transparent communication of scope and limits.

\noindent\textbf{Mechanism stage.} The primary harm arises from the inevitable privacy--utility trade-off: stronger privacy typically requires injecting randomness that reduces output fidelity.
Additional risks include (i) miscalibration or misuse from incorrect priors, assumptions, or implementations; (ii) privacy-washing or miscommunication if residual-risk metrics are presented as context-free guarantees; and (iii) increased engineering burden that raises the risk of implementation errors, especially for less-resourced teams.
We provide finite-sample analysis and recommend documenting priors, modeling, and calibration choices, using conservative settings when assumptions are uncertain, and treating R-PAC bounds as one input to broader privacy/security review that includes independent auditing and pressure-testing.

\noindent\textbf{Downstream stage.} Noise can degrade decision quality and cause harmful errors in high-stakes contexts (e.g., healthcare, finance, admissions, risk scoring). 
Residual-risk metrics may also be overly relied upon as blanket approval for aggressive data use, even in domains where any non-zero residual risk or utility degradation is unacceptable.
Our work does not deploy real systems; we use benchmarks to study trade-offs. For high-stakes deployments, we recommend conservative parameters, domain-specific evaluation of decision quality, and governance processes that do not treat formal bounds as the sole determinant of acceptability.

\noindent\textbf{Cross-cutting societal impacts.} If widely adopted, R-PAC-style methods may help align formal claims with empirical behavior and make residual risk explicit, supporting more realistic interpretation by auditors and regulators. However, they could be misused to justify aggressive data use if assumptions are obscured. Clear documentation, independent auditing, and appropriate oversight are important safeguards.

\subsection*{Decision to Conduct and Publish This Work}

From a Beneficence perspective, our goal is to improve utility \textit{at a fixed residual-privacy budget}, reducing harms associated with overly conservative or difficult-to-use mechanisms that can make formal privacy less practical. From a Respect for Persons and Justice perspective, making residual risk explicit and tying guarantees to concrete modeling and calibration choices supports more truthful, context-aware communication of privacy properties, while avoiding claims of absolute protection. 
From a Respect for Law and Public Interest perspective, more transparent and realistic privacy accounting can help regulators, standards bodies, and auditors evaluate systems where privacy claims depend on explicitly stated assumptions such as priors and model classes.
We consider it ethically justified to conduct and publish this work, but deployment remains context-dependent.

\section*{Open Science}

All artifacts necessary to evaluate our contribution consist solely of source code, available at the repository on Zenodo:
\url{https://doi.org/10.5281/zenodo.17871622}. 

The repository contains implementations of SR-PAC and all baseline methods. All required benchmarks are either public
datasets that are automatically downloaded by the scripts from their official sources, or small data files
included directly in the repository. Please refer to \texttt{README.md} for further details.

\section*{Acknowledgements}

This work was partially supported by the NSF (IIS-2214141, ITE-2452834), ARO (W911NF-25-1-0059), ONR (N000142412663), and Amazon.

\bibliographystyle{plain}
\bibliography{reference}

@inproceedings{balle2022reconstructing,
  title={Reconstructing training data with informed adversaries},
  author={Balle, Borja and Cherubin, Giovanni and Hayes, Jamie},
  booktitle={2022 IEEE Symposium on Security and Privacy (SP)},
  pages={1138--1156},
  year={2022},
  organization={IEEE}
}

@article{cinar2019classification,
  title={Classification of rice varieties using artificial intelligence methods},
  author={Cinar, Ilkay and Koklu, Murat},
  journal={International Journal of Intelligent Systems and Applications in Engineering},
  volume={7},
  number={3},
  pages={188--194},
  year={2019},
  publisher={International Journal of Intelligent Systems and Applications in Engineering}
}

@article{fisher1988iris,
  title={Iris. UCI machine learning repository},
  author={Fisher, Ronald Aylmer},
  journal={DOI: https://doi. org/10.24432/C56C76},
  year={1988}
}

@article{otto2000generalization,
  title={Generalization of an inequality by Talagrand and links with the logarithmic Sobolev inequality},
  author={Otto, Felix and Villani, C{\'e}dric},
  journal={Journal of Functional Analysis},
  volume={173},
  number={2},
  pages={361--400},
  year={2000},
  publisher={Elsevier}
}

@article{talagrand1996transportation,
  title={Transportation cost for Gaussian and other product measures},
  author={Talagrand, Michel},
  journal={Geometric \& Functional Analysis GAFA},
  volume={6},
  number={3},
  pages={587--600},
  year={1996},
  publisher={Springer}
}

@article{bonneel2015sliced,
  title={Sliced and radon wasserstein barycenters of measures},
  author={Bonneel, Nicolas and Rabin, Julien and Peyr{\'e}, Gabriel and Pfister, Hanspeter},
  journal={Journal of Mathematical Imaging and Vision},
  volume={51},
  number={1},
  pages={22--45},
  year={2015},
  publisher={Springer}
}

@article{donsker1975asymptotic,
  title={Asymptotic evaluation of certain Markov process expectations for large time, I},
  author={Donsker, Monroe D and Varadhan, SR Srinivasa},
  journal={Communications on pure and applied mathematics},
  volume={28},
  number={1},
  pages={1--47},
  year={1975},
  publisher={Wiley Online Library}
}

@article{palomar2005gradient,
  title={Gradient of mutual information in linear vector Gaussian channels},
  author={Palomar, Daniel P and Verd{\'u}, Sergio},
  journal={IEEE Transactions on Information Theory},
  volume={52},
  number={1},
  pages={141--154},
  year={2005},
  publisher={IEEE}
}

@inproceedings{
  zhang2015character,
  title={Character-level Convolutional Networks for Text Classification},
  author={Xiang Zhang and Junbo Zhao and Yann LeCun},
  booktitle={Advances in Neural Information Processing Systems 28 (NIPS 2015)},
  year={2015}
}

@article{lecun2010mnist,
  title={MNIST handwritten digit database},
  author={LeCun, Yann and Cortes, Corinna and Burges, CJ},
  journal={ATT Labs [Online]. Available: http://yann.lecun.com/exdb/mnist},
  volume={2},
  year={2010}
}

@TECHREPORT{Krizhevsky09learningmultiple,
  author = {Alex Krizhevsky},
  title = {Learning multiple layers of features from tiny images},
  institution = {},
  year = {2009}
}

@article{park2012equivalence,
  title={On the equivalence between Stein and de Bruijn identities},
  author={Park, Sangwoo and Serpedin, Erchin and Qaraqe, Khalid},
  journal={IEEE Transactions on Information Theory},
  volume={58},
  number={12},
  pages={7045--7067},
  year={2012},
  publisher={IEEE}
}

@article{selvi2025differential,
  title={Differential privacy via distributionally robust optimization},
  author={Selvi, Aras and Liu, Huikang and Wiesemann, Wolfram},
  journal={Operations Research},
  year={2025},
  publisher={INFORMS}
}

@article{saeidian2023pointwise,
  title={Pointwise maximal leakage},
  author={Saeidian, Sara and Cervia, Giulia and Oechtering, Tobias J and Skoglund, Mikael},
  journal={IEEE Transactions on Information Theory},
  volume={69},
  number={12},
  pages={8054--8080},
  year={2023},
  publisher={IEEE}
}

@article{issa2019operational,
  title={An operational approach to information leakage},
  author={Issa, Ibrahim and Wagner, Aaron B and Kamath, Sudeep},
  journal={IEEE Transactions on Information Theory},
  volume={66},
  number={3},
  pages={1625--1657},
  year={2019},
  publisher={IEEE}
}

@inproceedings{shokri2017membership,
  title={Membership inference attacks against machine learning models},
  author={Shokri, Reza and Stronati, Marco and Song, Congzheng and Shmatikov, Vitaly},
  booktitle={2017 IEEE symposium on security and privacy (SP)},
  pages={3--18},
  year={2017},
  organization={IEEE}
}

@inproceedings{carlini2022membership,
  title={Membership inference attacks from first principles},
  author={Carlini, Nicholas and Chien, Steve and Nasr, Milad and Song, Shuang and Terzis, Andreas and Tramer, Florian},
  booktitle={2022 IEEE Symposium on Security and Privacy (SP)},
  pages={1897--1914},
  year={2022},
  organization={IEEE}
}

@inproceedings{xiao2023pac,
  title={Pac privacy: Automatic privacy measurement and control of data processing},
  author={Xiao, Hanshen and Devadas, Srinivas},
  booktitle={Annual International Cryptology Conference},
  pages={611--644},
  year={2023},
  organization={Springer}
}

@article{sridhar2024pac,
  title={PAC-Private Algorithms},
  author={Sridhar, Mayuri and Xiao, Hanshen and Devadas, Srinivas},
  journal={Cryptology ePrint Archive},
  year={2024}
}

@inproceedings{humphries2023investigating,
  title={Investigating membership inference attacks under data dependencies},
  author={Humphries, Thomas and Oya, Simon and Tulloch, Lindsey and Rafuse, Matthew and Goldberg, Ian and Hengartner, Urs and Kerschbaum, Florian},
  booktitle={2023 IEEE 36th Computer Security Foundations Symposium (CSF)},
  pages={473--488},
  year={2023},
  organization={IEEE}
}

@article{lebanon2009beyond,
  title={Beyond k-anonymity: A decision theoretic framework for assessing privacy risk},
  author={Lebanon, Guy and Scannapieco, Monica and Fouad, Mohamed and Bertino, Elisa},
  journal={Transactions on Data Privacy},
  year={2009}
}

@inproceedings{hannun2021measuring,
  title={Measuring data leakage in machine-learning models with fisher information},
  author={Hannun, Awni and Guo, Chuan and van der Maaten, Laurens},
  booktitle={Uncertainty in Artificial Intelligence},
  pages={760--770},
  year={2021},
  organization={PMLR}
}

@inproceedings{guo2022bounding,
  title={Bounding training data reconstruction in private (deep) learning},
  author={Guo, Chuan and Karrer, Brian and Chaudhuri, Kamalika and van der Maaten, Laurens},
  booktitle={International Conference on Machine Learning},
  pages={8056--8071},
  year={2022},
  organization={PMLR}
}

@inproceedings{bun2018composable,
  title={Composable and versatile privacy via truncated cdp},
  author={Bun, Mark and Dwork, Cynthia and Rothblum, Guy N and Steinke, Thomas},
  booktitle={Proceedings of the 50th Annual ACM SIGACT Symposium on Theory of Computing},
  pages={74--86},
  year={2018}
}

@article{dwork2016concentrated,
  title={Concentrated differential privacy},
  author={Dwork, Cynthia and Rothblum, Guy N},
  journal={arXiv preprint arXiv:1603.01887},
  year={2016}
}

@inproceedings{bun2016concentrated,
  title={Concentrated differential privacy: Simplifications, extensions, and lower bounds},
  author={Bun, Mark and Steinke, Thomas},
  booktitle={Theory of cryptography conference},
  pages={635--658},
  year={2016},
  organization={Springer}
}

@article{lopuhaa2024mechanisms,
  title={Mechanisms for Robust Local Differential Privacy},
  author={Lopuha{\"a}-Zwakenberg, Milan and Goseling, Jasper},
  journal={Entropy},
  volume={26},
  number={3},
  pages={233},
  year={2024},
  publisher={MDPI}
}

@inproceedings{goseling2022robust,
  title={Robust optimization for local differential privacy},
  author={Goseling, Jasper and Lopuha{\"a}-Zwakenberg, Milan},
  booktitle={2022 IEEE International Symposium on Information Theory (ISIT)},
  pages={1629--1634},
  year={2022},
  organization={IEEE}
}

@article{farokhi2017fisher,
  title={Fisher information as a measure of privacy: Preserving privacy of households with smart meters using batteries},
  author={Farokhi, Farhad and Sandberg, Henrik},
  journal={IEEE Transactions on Smart Grid},
  volume={9},
  number={5},
  pages={4726--4734},
  year={2017},
  publisher={IEEE}
}

@article{sankar2013utility,
  title={Utility-privacy tradeoffs in databases: An information-theoretic approach},
  author={Sankar, Lalitha and Rajagopalan, S Raj and Poor, H Vincent},
  journal={IEEE Transactions on Information Forensics and Security},
  volume={8},
  number={6},
  pages={838--852},
  year={2013},
  publisher={IEEE}
}

@inproceedings{chatzikokolakis2010statistical,
  title={Statistical measurement of information leakage},
  author={Chatzikokolakis, Konstantinos and Chothia, Tom and Guha, Apratim},
  booktitle={International Conference on Tools and Algorithms for the Construction and Analysis of Systems},
  pages={390--404},
  year={2010},
  organization={Springer}
}

@inproceedings{cuff2016differential,
  title={Differential privacy as a mutual information constraint},
  author={Cuff, Paul and Yu, Lanqing},
  booktitle={Proceedings of the 2016 ACM SIGSAC Conference on Computer and Communications Security},
  pages={43--54},
  year={2016}
}

@article{xiao2008output,
  title={Output perturbation with query relaxation},
  author={Xiao, Xiaokui and Tao, Yufei},
  journal={Proceedings of the VLDB Endowment},
  volume={1},
  number={1},
  pages={857--869},
  year={2008},
  publisher={VLDB Endowment}
}

@inproceedings{ghosh2009universally,
  title={Universally utility-maximizing privacy mechanisms},
  author={Ghosh, Arpita and Roughgarden, Tim and Sundararajan, Mukund},
  booktitle={Proceedings of the forty-first annual ACM symposium on Theory of computing},
  pages={351--360},
  year={2009}
}

@inproceedings{murtagh2015complexity,
  title={The complexity of computing the optimal composition of differential privacy},
  author={Murtagh, Jack and Vadhan, Salil},
  booktitle={Theory of Cryptography Conference},
  pages={157--175},
  year={2015},
  organization={Springer}
}

@inproceedings{gupte2010universally,
  title={Universally optimal privacy mechanisms for minimax agents},
  author={Gupte, Mangesh and Sundararajan, Mukund},
  booktitle={Proceedings of the twenty-ninth ACM SIGMOD-SIGACT-SIGART symposium on Principles of database systems},
  pages={135--146},
  year={2010}
}

@inproceedings{geng2020tight,
  title={Tight analysis of privacy and utility tradeoff in approximate differential privacy},
  author={Geng, Quan and Ding, Wei and Guo, Ruiqi and Kumar, Sanjiv},
  booktitle={International Conference on Artificial Intelligence and Statistics},
  pages={89--99},
  year={2020},
  organization={PMLR}
}

@inproceedings{mironov2017renyi,
  title={R{\'e}nyi differential privacy},
  author={Mironov, Ilya},
  booktitle={2017 IEEE 30th computer security foundations symposium (CSF)},
  pages={263--275},
  year={2017},
  organization={IEEE}
}

@inproceedings{dwork2006calibrating,
  title={Calibrating noise to sensitivity in private data analysis},
  author={Dwork, Cynthia and McSherry, Frank and Nissim, Kobbi and Smith, Adam},
  booktitle={Theory of cryptography conference},
  pages={265--284},
  year={2006},
  organization={Springer}
}

@article{huang2018generative,
  title={Generative adversarial privacy},
  author={Huang, Chong and Kairouz, Peter and Chen, Xiao and Sankar, Lalitha and Rajagopal, Ram},
  journal={arXiv preprint arXiv:1807.05306},
  year={2018}
}

@inproceedings{chen2018understanding,
  title={Understanding compressive adversarial privacy},
  author={Chen, Xiao and Kairouz, Peter and Rajagopal, Ram},
  booktitle={2018 IEEE Conference on Decision and Control (CDC)},
  pages={6824--6831},
  year={2018},
  organization={IEEE}
}

@inproceedings{abadi2016deep,
  title={Deep learning with differential privacy},
  author={Abadi, Martin and Chu, Andy and Goodfellow, Ian and McMahan, H Brendan and Mironov, Ilya and Talwar, Kunal and Zhang, Li},
  booktitle={Proceedings of the 2016 ACM SIGSAC conference on computer and communications security},
  pages={308--318},
  year={2016}
}

@inproceedings{jordon2018pate,
  title={PATE-GAN: Generating synthetic data with differential privacy guarantees},
  author={Jordon, James and Yoon, Jinsung and Van Der Schaar, Mihaela},
  booktitle={International conference on learning representations},
  year={2018}
}

@inproceedings{dwork2006differential,
  title={Differential privacy},
  author={Dwork, Cynthia},
  booktitle={International colloquium on automata, languages, and programming},
  pages={1--12},
  year={2006},
}

@inproceedings{nasr2018machine,
  title={Machine learning with membership privacy using adversarial regularization},
  author={Nasr, Milad and Shokri, Reza and Houmansadr, Amir},
  booktitle={Proceedings of the 2018 ACM SIGSAC conference on computer and communications security},
  pages={634--646},
  year={2018}
}

@inproceedings{du2012privacy,
  title={Privacy against statistical inference},
  author={du Pin Calmon, Fl{\'a}vio and Fawaz, Nadia},
  booktitle={2012 50th annual Allerton conference on communication, control, and computing (Allerton)},
  pages={1401--1408},
  year={2012},
  organization={IEEE}
}

@inproceedings{alghamdi2022cactus,
  title={Cactus mechanisms: Optimal differential privacy mechanisms in the large-composition regime},
  author={Alghamdi, Wael and Asoodeh, Shahab and Calmon, Flavio P and Kosut, Oliver and Sankar, Lalitha and Wei, Fei},
  booktitle={2022 IEEE International Symposium on Information Theory (ISIT)},
  pages={1838--1843},
  year={2022},
}

@inproceedings{kairouz2015composition,
  title={The composition theorem for differential privacy},
  author={Kairouz, Peter and Oh, Sewoong and Viswanath, Pramod},
  booktitle={International conference on machine learning},
  pages={1376--1385},
  year={2015},
  organization={PMLR}
}

@article{stam1959some,
  title={Some inequalities satisfied by the quantities of information of Fisher and Shannon},
  author={Stam, Aart J},
  journal={Information and Control},
  volume={2},
  number={2},
  pages={101--112},
  year={1959},
  publisher={Elsevier}
}

@techreport{dittrich2012menlo,
  title={The Menlo Report: Ethical principles guiding information and communication technology research},
  author={Dittrich, David and Kenneally, Erin and others},
  year={2012},
  institution={US Department of Homeland Security}
}

@article{xiao2025pac,
  title={PAC Privacy and Black-Box Privatization},
  author={Xiao, Hanshen and Devadas, Srinivas},
  journal={IEEE Security \& Privacy},
  volume={23},
  number={4},
  pages={92--97},
  year={2025},
  publisher={IEEE}
}

@inproceedings{rabin2011wasserstein,
  title={Wasserstein barycenter and its application to texture mixing},
  author={Rabin, Julien and Peyr{\'e}, Gabriel and Delon, Julie and Bernot, Marc},
  booktitle={International conference on scale space and variational methods in computer vision},
  pages={435--446},
  year={2011},
  organization={Springer}
}

@inproceedings{balle2020hypothesis,
  title={Hypothesis testing interpretations and renyi differential privacy},
  author={Balle, Borja and Barthe, Gilles and Gaboardi, Marco and Hsu, Justin and Sato, Tetsuya},
  booktitle={International Conference on Artificial Intelligence and Statistics},
  pages={2496--2506},
  year={2020},
  organization={PMLR}
}

@article{zhang2025sliced,
  title={Sliced R\'enyi Pufferfish Privacy: Directional Additive Noise Mechanism and Private Learning with Gradient Clipping},
  author={Zhang, Tao and Vorobeychik, Yevgeniy},
  journal={arXiv preprint arXiv:2512.01115},
  year={2025}
}

@article{kifer2014pufferfish,
  title={Pufferfish: A framework for mathematical privacy definitions},
  author={Kifer, Daniel and Machanavajjhala, Ashwin},
  journal={ACM Transactions on Database Systems (TODS)},
  volume={39},
  number={1},
  pages={1--36},
  year={2014},
  publisher={ACM New York, NY, USA}
}

@inproceedings{pierquin2024renyi,
  title={R{\'e}nyi pufferfish privacy: General additive noise mechanisms and privacy amplification by iteration via shift reduction lemmas},
  author={Pierquin, Cl{\'e}ment and Bellet, Aur{\'e}lien and Tommasi, Marc and Boussard, Matthieu},
  booktitle={International Conference on Machine Learning (ICML 2024)},
  year={2024}
}

@inproceedings{Song2017,
author = {Song, Shuang and Wang, Yizhen and Chaudhuri, Kamalika},
title = {Pufferfish Privacy Mechanisms for Correlated Data},
year = {2017},
isbn = {9781450341974},
url = {https://doi.org/10.1145/3035918.3064025},
doi = {10.1145/3035918.3064025},
booktitle = {Proceedings of the 2017 ACM International Conference on Management of Data},
pages = {1291–1306},
numpages = {16},
series = {SIGMOD '17}
}

@book{boyd2004convex,
  title={Convex optimization},
  author={Boyd, Stephen and Vandenberghe, Lieven},
  year={2004},
  publisher={Cambridge university press}
}

@inproceedings{zhang2025differential,
  title={Differential Confounding Privacy and Inverse Composition},
  author={Zhang, Tao and Malin, Bradley A and Raviv, Netanel and Vorobeychik, Yevgeniy},
  booktitle={2025 IEEE International Symposium on Information Theory (ISIT)},
  pages={1--6},
  year={2025},
  organization={IEEE}
}

@article{li2021large,
  title={Large language models can be strong differentially private learners},
  author={Li, Xuechen and Tramer, Florian and Liang, Percy and Hashimoto, Tatsunori},
  journal={arXiv preprint arXiv:2110.05679},
  year={2021}
}

@article{goldfeld2021sliced,
  title={Sliced mutual information: A scalable measure of statistical dependence},
  author={Goldfeld, Ziv and Greenewald, Kristjan},
  journal={Advances in Neural Information Processing Systems},
  volume={34},
  pages={17567--17578},
  year={2021}
}

@article{zhang2025breaking,
  title={Breaking the Gaussian Barrier: Residual-PAC Privacy for Automatic Privatization},
  author={Zhang, Tao and Vorobeychik, Yevgeniy},
  journal={arXiv preprint arXiv:2506.06530},
  year={2025}
}

\appendix

\section{Discussion: PAC/R-PAC Privacy vs. Differential Privacy}\label{app:difference_DP_PAC}

In this section, we discuss the difference and the relationship between PAC/R-PAC Privacy and DP (Definition \ref{def:DP}).

DP and PAC (and R-PAC) Privacy use different semantics and different privacy quantification metrics.
DP considers the presence or absence of individual records as secrets and ensures that, regardless of external knowledge, an adversary with access to mechanism outputs makes similar conclusions whether or not any individual data record is included in the dataset \cite{dwork2006differential}.
DP employs the worst-case probabilistic \textit{input-independent indistinguishability} to quantify the privacy risk, by using the max-divergence for pure DP and the hockey stick divergence for the approximated DP.

Unlike DP, PAC Privacy can protect secrets that go beyond individual data records.
PAC Privacy measures privacy in terms of the adversary's difficulty in achieving accurate reconstruction, capturing the semantics of the \textit{impossibility of customized adversarial inference}~\cite{xiao2025pac}, where the adversary is assumed to be computationally-unbounded. 
PAC Advantage Privacy uses the posterior advantage $\Delta^{\delta}_{f}$ (Definition \ref{def:PAC_advantage}) to quantify privacy risk, where $\Delta^{\delta}_{f}$ depends on a chosen $f$-divergence, secret (e.g., data) entropy determined by $\mathcal{D}$, an attack model in terms of $\rho$.
When $f$-divergence is instantiated as KL divergence, $\Delta^{\delta}_{f}$ is upper bounded by mutual information that is uniform over all adversaries and admissible $\rho$.

R-PAC and PAC Privacy are two sides of the same coin: PAC quantifies leakage (e.g., via $\Delta^{\delta}_{f}$), while R-PAC quantifies the remaining privacy (e.g., via $\mathtt{R}^{\delta}_{f}$), linked exactly by $\mathtt{IntP}_f(\mathcal{D}) = \mathtt{R}_f^\delta + \Delta_f^\delta$ (equation (\ref{eq:link_PAC_privacy})). 
R-PAC Privacy uses the same semantics as PAC Privacy uses the posterior disadvantage $\mathtt{R}_{f}^{\delta}$ (Definition \ref{def:residual_PAC}) to quantify the remaining privacy after leakage.
When KL divergence is used, the posterior disadvantage $\mathtt{R}_{f}^{\delta}$ is lower bounded by the conditional entropy, which is uniform over all adversaries and admissible $\rho$.

PAC (and R-PAC) Privacy provides a more general framework that quantifies the reconstruction hardness for any sensitive information that an adversary might seek to infer. 
This encompasses not only individual membership inference (a special case), but also broader privacy concerns such as data reconstruction within specified error bounds, identification of multiple participants, or recovery of sensitive attributes.
Crucially, PAC Privacy operates under distributional assumptions about the data (or general secrets) generation process $\mathcal{D}$, enabling instance-based analysis that can potentially require less noise than worst-case DP guarantees.
However, the automatic privatization procedures (e.g., Auto-PAC and Efficient-PAC) proposed to realize PAC Privacy certify the privacy guarantee via an MI budget $\beta$, enforcing Gaussian surrogate bound $\mathtt{LogDet}(\mathcal{M}(X), B)\leq \beta$ as a sufficient condition.

Consequently, the delivered mechanisms inherit MI–based properties and caveats (e.g., data processing, distributional/average–case nature, standard composition scaling, and lack of worst–case indistinguishability unless additional constraints are imposed), even though the abstract PAC Privacy notion itself does not rely on MI. 
Our SR-PAC follows the MI principle but implements an MI bound that is tighter than the Gaussian surrogate bound $\mathtt{LogDet}(\mathcal{M}(X), B)$.


While PAC (resp. R-PAC) privacy uses the semantics of impossibility of customized adversarial inference and is independent of mutual information (resp. conditional entropy), Auto-PAC uses mutual information (resp. conditional entropy) to quantify privacy risk.
In this section, we discuss the difference between mutual information (MI) as privacy quantification and the input-independent indistinguishability of DP.

\textbf{What each notion protects. }
Now, we discuss what each privacy notion protects.
Let $\mathcal{M}:\mathcal{X}\mapsto \mathcal{Y}$ be a randomized mechanism. DP, independent of input distribution, protects \textit{worst-case, per-individual input-independen indistinguishability} by ensuring a uniform bound $\ell(x,y) \leq \varepsilon$ almost surely (up to $\delta$ in the $(\varepsilon,\delta)$ case), where $\ell(x,y) = \log \frac{P_{X|Y}(x|y)}{P_X(x)}$ is the privacy-loss random variable. However, DP does not, in general, ensure that the \textit{average} leakage $\mathtt{MI}(X;Y)$ is small—indeed, $\mathtt{MI}(X;Y)$ can scale with the dataset size unless $\varepsilon$ shrinks appropriately. In contrast, MI-based privacy constrains the \textit{average information} leaked from inputs to outputs under a specific input distribution $\mathcal{D} \in \Delta(\mathcal{X})$. MI controls average leakage: $\mathtt{MI}(X;Y) = \mathbb{E}_{P_{XY}}[\ell(X,Y)] \leq \beta$ upper-bounds the expected log-likelihood gain of an optimal Bayesian adversary. However, MI \textit{does not} by itself bound the worst-case leakage $L \equiv \operatorname*{ess\,sup}\ell$; in particular, $\mathtt{MI}(X;Y) \leq \beta$ is compatible with $L = \infty$ (rare but arbitrarily large disclosures).

\textbf{Worst-case vs. average-case guarantees. }
DP is a distribution-free, worst-case guarantee that must hold for all neighboring datasets and all adversaries, independent of any input distribution. By contrast, MI-based privacy is \textit{distributional}: it controls \textit{expected} leakage under an input distribution $P_X$, typically via $\mathtt{MI}(X;Y) \leq \beta$ for $Y = \mathcal{M}(X)$. Because $\mathtt{MI}(X;Y) = \mathbb{E}_{P_{XY}}[\ell(X,Y)]$, the noise needed to enforce $\mathtt{MI}(X;Y) \leq \beta$ depends on $P_X$: when most probability mass lies on inputs for which $\ell(X,Y)$ is typically small, less perturbation can suffice, and the resulting noise shape may be tailored to that distribution. At the same time, an MI budget does not by itself preclude rare high-leakage cases: if there exists a measurable event $E \subseteq \mathcal{X} \times \mathcal{Y}$ with $P_{XY}(E) = p$ and $\ell(x,y) \geq L$ for all $(x,y) \in E$, then $\mathtt{MI}(X;Y) \geq pL$; hence the constraint $\mathtt{MI}(X;Y) \leq \beta$ forbids such a case only when $pL > \beta$ (and any "perfect disclosure" with $L = \infty$ is incompatible for all $p > 0$).


\textbf{Name-and-shame example. }
One example of "rare high–leakage cases" is the \textit{name-and-shame}.
Let $E=\{(x,y):y=x\}$ denote the event in which the mechanism reveals the input directly, occurring with probability $p$. On $E$, the per–sample leakage is $\ell(x,y)=\log \tfrac{p_{X\mid Y}(x\mid x)}{p_X(x)}=-\log p_X(x)$, which can be very large (and unbounded when $p_X$ has heavy tails or continuous support). Thus this is a small–probability, high–leakage branch. In the discrete case with finite support, one has $\mathtt{MI}(X;Y)=p \mathcal{H}(X)$. Choosing $p=\beta/\mathcal{H}(X)$ makes $\mathtt{MI}(X;Y)=\beta$, which saturates the heuristic $I\mathtt{MI}(X;Y)\ge pL$ when $L$ is interpreted as the average leakage $\mathcal{H}(X)$ on $E$. If one insists on the pointwise form from the paragraph, taking $L_0=\operatorname*{ess\,inf}_{x}(-\log p_X(x))$ yields $\mathtt{MI}(X;Y)(X;Y)\ge p\,L_0$, which still places the example in the same regime. Finally, if "name–and–shame" is modeled as perfect disclosure with continuous $X$, then $\ell=\infty$ on $E$ and the constraint rules it out immediately, since $L=\infty$ is incompatible with any $p>0$.

\textbf{DP perspective on the name-and-shame example. }
To see why this example fundamentally conflicts with the DP notion of rare-but-exact disclosure, consider the per–record "name–and–shame" mechanism $M$ that, independently for each index $i$, outputs $(i,x_i)$ with probability $p$ and $\bot$ otherwise. Let $x$ and $x'$ be neighboring databases that differ only in record $i$, and define the event $E={(i,x_i)}$. Then
\[
\Pr[M(x)\in E]=p,\quad \Pr[M(x')\in E]=0.
\]
The $(\epsilon,\bar{\delta})$–DP inequality for $E$ reads $p\le e^\epsilon\cdot 0+\bar{\delta}=\bar{\delta}$, hence any $(\epsilon,\bar{\delta})$ satisfied by $M$ must obey $\bar{\delta}\ge p$. In particular, with the standard regime $\bar{\delta}\ll 1/n$ (negligible failure probability), such a mechanism is \textit{not} DP for any finite $\epsilon$; conversely, allowing $\bar{\delta}\ge p$ makes the guarantee vacuous on the $p$–fraction of runs that reveal $(i,x_i)$ exactly.

\subsection{Fair Comparison Under MIA}\label{app_sub:fair_MIA}

PAC Privacy and R-PAC Privacy (and also MI-based privacy) address complementary notions of privacy to DP. 
Neither framework dominates the other.
To perform a fair comparison, we focus on the cases when the privacy budgets of DP and PAC/R-PAC Privacy are "equalized".
In particular, we consider Membership Inference Attack (MIA) defined by Definition~\ref*{def:MIA} in Appendix~\ref*{app:MIA} in \cite{zhang2025breaking}.

DP can be understood through the lens of membership inference success rates.
Consider the membership inference scenario from Definition \ref*{def:MIA} (Appendix~\ref*{app:MIA} in \cite{zhang2025breaking}), where we have a dataset of size $n=\frac{N}{2}$ (i.e., each individual data record has a $50\%$ probability of being included in the selected subset $X$).
If a mechanism $\mathcal{M}$ is $(\epsilon, \bar{\delta})$-DP, then by \cite{kairouz2015composition,humphries2023investigating}, an adversary's ability to successfully infer whether a specific individual record $i$ is included in the dataset (i.e., posterior success rate $p_{o} = 1-\delta_i$) is fundamentally limited: 
\begin{equation}\label{eq:posterior_DP}
    p_{o}\leq 1 - \frac{1 - \bar{\delta}}{1 + e^\epsilon}.
\end{equation}
This bound demonstrates how DP parameters directly translate into concrete limits on an adversary's inference capabilities in MIA.
Thus, the maximal posterior success rate permitted by $(\epsilon,\bar{\delta})$-DP is $1 - \frac{1 - \bar{\delta}}{1 + e^\epsilon}$.

In addition, there is a relationship between the posterior success rate $p_{o}$ and the mutual information \cite{sridhar2024pac} (derived from (\ref{eq:KL_upper_MI})):
\begin{equation}
    p_{o}\log\frac{p_{o}}{\bar{p}} + (1-p_{o})\log\frac{1-p_{o}}{1-\bar{p}} \leq \mathtt{MI}(X;\mathcal{M}(X)),
\end{equation}
where $\bar{p}$ is the optimal prior success rate, which is $\max(r, 1-r)$ with $r$ as the subsampling rate that selects the dataset from a data pool.
Thus, given a privacy budget $\mathtt{MI}(X;\mathcal{M}(X))=\beta$ and a prior success rate $\bar{p}$, we can calculate the posterior success level $p_{o}$ and, by (\ref{eq:posterior_DP}), pin down $\epsilon$ for a chosen $\bar{\delta}$ so that DP has an "equivalent" budget to PAC.
The corresponding R-PAC budget is $\mathcal{H}(X) - \beta$.

For per-individual membership, the relevant secret is the membership indicator for person $i$, and the mechanism output is $Y=\mathcal{M}(X)$.
Let $U_i \in \{0,1\}$ denote the membership indicator as specified in Appendix \ref*{app:MIA} of \cite{zhang2025breaking}.
Since $U_i \to X \to Y$ forms a Markov chain, the data-processing inequality gives $\mathtt{MI}(U_i;Y) \leq \mathtt{MI}(X;Y)=\beta.$
The Bernoulli–KL inequality used above applies equally with $\mathtt{MI}(U_i;Y)$ on the right-hand side; replacing it by $\mathtt{MI}(X;Y)$ is therefore conservative and still yields a valid upper bound on the Bayes-optimal membership posterior success $p_o$.
This validates using $\mathtt{MI}(X;Y)$ to compute $p_o(\beta,\bar p)$ for MIA and then selecting $(\epsilon,\bar\delta)$ so that (\ref{eq:posterior_DP}) enforces the same $p_o$ for a fair, like-for-like comparison between DP and PAC/R-PAC.

\subsection{Noise Magnitude}

In this section, we discuss how they differ in \textit{noise magnitude} under an \textit{equalized privacy budget}. Concretely, we fix a mutual-information budget $\beta$ for PAC/R-PAC; when contrasting with DP, we use the $(\epsilon,\bar\delta)$ that induces the \textit{same} posterior-success level via the MI$ \leftrightarrow$DP conversion described in Section \ref{app_sub:fair_MIA}. Even at this matched budget, the required noise can vary substantially. We measure it by the total noise magnitude $V(\beta) \equiv \mathbb{E}\|B\|_2^2$ for outputs $Y=\mathcal{M}(X)+B$. Let the centered output covariance have eigenvalues $\lambda_1 \ge \cdots \ge \lambda_p > 0$ on its informative $p$-dimensional subspace ($p\le d$), and write $R=\max_j \lambda_j$. We first present the \textit{ideal} Auto-PAC baseline derived from the log-det MI bound, then the (SR-PAC) optimizer that tightens noise under the same $\beta$, and finally contrast both with classical DP mechanisms that must mask worst-case sensitivity in $d$ dimensions. (Throughout, Auto-PAC refers to this ideal log-det calibration; the practical Algorithm~\ref{alg:PAC_original} uses estimated eigenvalues and a stabilization $10cv/\beta$, yielding total noise magnitude $(\sum_j \sqrt{\hat\lambda_j + 10cv/\beta})^2/(2v)$, a conservative upper envelope of the ideal baseline.)

\textbf{Auto-PAC. } Let the (centered) mechanism output have covariance eigenvalues $\lambda_1 \ge \dots \ge \lambda_p>0$ (in its informative $p$-dimensional subspace). Auto-PAC calibrates \textit{Gaussian} noise $B\sim\mathcal N(0,\Sigma_B)$ under an MI budget $\beta$, yielding the total noise magnitude
\[
V_{\text{PAC}}(\beta)
\;=\;
\mathbb E\|B\|_2^2
\;=\;
\frac{\Big(\sum_{j=1}^{p}\sqrt{\lambda_j}\Big)^2}{2\beta}.
\]
(When the exhibited calibration targets $\mathtt{MI} \le  \tfrac12$, this specializes to $V_{\text{PAC}}=(\sum_j\sqrt{\lambda_j})^2$.) A general bound is
\[
\Big(\textstyle\sum_{j=1}^p \sqrt{\lambda_j}\Big)^2
\;\le\;
p\sum_{j=1}^p \lambda_j
\;\leq\;
p^2 R,
\]
where $R\equiv \max_j \lambda_j$.
Hence $V_{\text{PAC}}(\beta)=O(p^2R/\beta)$ in the worst case (and improves to $O(pR/\beta)$ if $\sum_j\lambda_j=O(R)$).
\textit{(In practice, Algorithm~1 uses estimated eigenvalues and a stabilization $10cv/\beta$, yielding total noise magnitude $(\sum_j \sqrt{\hat\lambda_j + 10cv/\beta})^2/(2v)$, which is a conservative upper envelope of the ideal log-det calibration.)}
Because differential privacy (DP) must mask \textit{worst-case} changes in all $d$ coordinates, the required noise for $d$-dimensional outputs typically grows like $\sqrt d$ (e.g., $O(\sqrt d/n)$ for mean queries with dataset size $n$)—the classic "curse of dimensionality." Thus, when the data are effectively low-rank ($p\ll d$), Auto-PAC already mitigates this dimensional blow-up.

\textbf{SR-PAC. } SR-PAC \textit{optimizes} the full noise distribution under the same MI budget $\beta$ and strictly improves (or matches) the Gaussian baseline:
\begin{itemize}
    \item \textit{Universal gain (non-Gaussian outputs):} For any non-Gaussian output $Z=\mathcal M(X)$, SR-PAC achieves
    \[
    \mathbb E\|B_{\text{SR}}\|_2^2 \;<\; \mathbb E\|B_{\text{PAC}}\|_2^2
    \quad\text{at the same }\beta,
    \]
    closing the conservativeness of Auto-PAC. (If $Z$ is exactly Gaussian, the gap can vanish.)
    \item \textit{Anisotropic allocation:} The Stackelberg-optimal covariance is provably anisotropic; variance is shifted toward directions with high leakage and away from benign ones, improving utility without violating the MI budget.
    \item \textit{Zero-noise subspaces:} Under a mild separation of directional sensitivities, there exists a threshold $\beta_{\mathrm{lab}}$ such that for all $\beta\le \beta_{\mathrm{lab}}$ SR-PAC injects \textit{no} noise on an $s$-dimensional task-critical subspace (e.g., the $k{-}1$ label directions in classification), reducing the order from $O(p)$ to $O(p{-}s)$ in those regimes.
\end{itemize}

\textbf{Comparison to DP. } Since SR-PAC pointwise dominates Auto-PAC for every $\beta$ and Auto-PAC already avoids DP's $\sqrt d$-type growth, SR-PAC inherits—and sharpens—the dimensional advantage. Writing
\[
V_{\text{SR}}(\beta)
\;=\;
V_{\text{PAC}}(\beta) \;-\; \Delta(\beta),
\qquad
0 \le \Delta(\beta) \le V_{\text{PAC}}(\beta),
\]
we have $\Delta(\beta)>0$ whenever $Z$ is non-Gaussian. In high-dimensional tasks with modest informative rank $p$ and harmless directions ($s>0$), SR-PAC reduces noise from $O(p)$ down to $O(p{-}s)$ (at fixed $\beta$), yielding a strictly better privacy–utility trade-off than both Auto-PAC and classical DP.

\subsection{Computational Complexity}\label{app:computational_cost}

In this section, we concisely characterize the computational complexity of $(1-\gamma)$-Confidence Auto-PAC (i.e., Algorithm \ref{alg:PAC_original}) and SR-PAC.
For simplicity, we still use Auto-PAC to refer to Algorithm \ref{alg:PAC_original}.
Let $d$ be the dimension of the mechanism output $\mathcal{M}(X)\in\mathbb{R}^{d}$.

\textbf{Auto-PAC. }
Let $m$ be the number of Monte Carlo simulations (samples) used by Auto-PAC.
In addition, let $C_{\mathcal{M}}$ denote the cost of one evaluation of the (black-box) mechanism $\mathcal{M}(\cdot)$.
Auto-PAC first draws $m$ i.i.d. samples $X^{(1)}, X^{(2)}, ..., X^{(m)}$.
For each sample $X^{(k)}$, Auto-PAC then evaluates $y^{(k)} = \mathcal{M}(X^{(k)})$; let $O(m C_{\mathcal{M}})$ denote the corresponding costs.
It then forms the empirical mean and full empirical covariance $\hat{\Sigma}\in\mathbb{R}^{d\times d}$, which costs $O(md^{2})$ time and $O(d^2)$ memory.
Finally, it performs an SVD/eigendecomposition of $\hat{\Sigma}$ and constructs $\Sigma_{B}$, which costs $O(d^{3})$ time (full decomposition) and $O(d^{2})$ memory. Overall, the cost of Auto-PAC is 
\begin{itemize}
    \item Time: $O(mC_{\mathcal{M}} + md^{2} + d^{3})$;
    \item Memory: $O(d^{2})$.
\end{itemize}
Here, the $d^{3}$ SVD/eigendecomposition step dominates at large output dimension.

\textbf{SR-PAC (Monte Carlo Stackelberg optimization). }
SR-PAC uses the same black-box sampling access to $\mathcal{M}(\cdot)$ but avoids $d\times d$ SVD operations.
In Algorithm~\ref{alg:SR_PAC}, each update uses a fresh Monte-Carlo batch of size $m$ (lines~6--8 and 12--15).
Let $C_{\mathcal{M}}$ be the cost of one evaluation of $\mathcal{M}(\cdot)$, let $C_{\pi}$ be the cost of one forward/backward pass of the decoder,
and let $C_{g}$ be the cost of sampling and differentiating through the perturbation rule.
Then each decoder-gradient step costs $O\!\left(m\,(C_{\mathcal{M}} + C_{\pi} + C_{g})\right)$ time (lines~6--9),
and each leader update costs $O\!\left(m\,(C_{\mathcal{M}} + C_{\pi} + C_{g})\right)$ time (lines~12--15).
Over $T_{\lambda}$ leader iterations with decoder-update phases triggered every $T_{\phi}$ iterations (line~3),
the total runtime is
\[
O\!\left(T_\lambda\, m\,(C_{\mathcal{M}} + C_{\pi} + C_{g})\;+\;N_{\text{dec}}\cdot m\,(C_{\mathcal{M}} + C_\pi + C_g)\right),
\]
where $N_{\text{dec}}$ is the total number of decoder-gradient steps (e.g., $N_{\text{dec}} \approx \lfloor T_\lambda/T_\phi\rfloor\cdot T_\phi$ in Algorithm~\ref{alg:SR_PAC}).
The memory cost is dominated by storing model parameters and one mini-batch:
\[
O\!\left(p_\pi + p_g + m d\right),
\]
where $p_\pi$ and $p_g$ are the parameter counts of the decoder and perturbation rule, respectively,
and $md$ accounts for holding a batch of $m$ outputs in $\mathbb{R}^d$ during a step.
In particular, SR-PAC avoids the $O(d^2)$ memory footprint and $O(d^3)$ matrix-decomposition bottleneck of Auto-PAC.

\textbf{DP. }
Sensitivity (i.e., the maximal possible change
on the output when a single data record changes) is the key component of DP privatization via noise perturbation.
However, computing sensitivity is, in general, NP-hard \cite{xiao2008output}. 
DP-SGD \cite{abadi2016deep} is a \textit{decompose-then-compose} privatization scheme for DP, which avoids explicit sensitivity computation.
However, DP-SGD adds per-iteration per-example gradient clipping and noise, inducing utility drop and computational overhead in large-scale applications \cite{li2021large}.
Since DP and PAC/R-PAC Privacy adopt fundamentally different semantics and different privatization schemes, it is not self-evident how to compare them fairly in terms of computational complexity.

\textbf{Scalable sliced variants. }
When $d$ is large, the SR-PAC principle also applies to \textit{sliced} objectives (e.g., sliced conditional entropy via sliced mutual information; see Appendix \ref*{app:sliced_rpac} of \cite{zhang2025breaking}),
which replaces high-dimensional estimation with $r$ one-dimensional projections.
This yields an additional factor $O(r)$ over sampling while avoiding $O(d^3)$ operations, i.e., $O(mC_{\mathcal{M}}+rmd)$
for estimating sliced leakage terms (plus the decoder-training term if used).

\newpage

\section{Trade-offs Across Privacy Standards and Ethical Considerations}\label{app:ethical_trade_offs}

Our work fits into a broader ecosystem of privacy frameworks, including differential privacy (DP) and its variants (e.g., R\'enyi DP), PAC-style privacy, and information-theoretic notions. These standards should be viewed as \textit{different languages for expressing and reasoning about the privacy of a mechanism}, with different semantics, parameters, and audiences, rather than as a single total order of "stronger" versus "weaker" privacy. In particular, it is generally not meaningful to compare numerical parameters across different notions directly (e.g., comparing an $\varepsilon$ value in one framework to an $\varepsilon$ value in another) without accounting for the notion's semantics and the relevant conversions or bounds.
For the same underlying mechanism, multiple privacy statements may hold simultaneously under different frameworks, but they summarize risk in different ways; selecting a framework is therefore also a communication, governance, and implementation choice.

\noindent\textbf{Benefits and costs of DP.}
Different privacy standards differ not only in semantics, but also in how privacy can be \textit{quantified, implemented, and accounted for under composition}. Worst-case notions such as DP provide distribution-free guarantees with strong robustness to auxiliary information, and they come with well-developed composition and accounting tools. However, implementing DP in complex pipelines can be challenging: classical DP mechanisms often require bounding (global) sensitivity, which may be difficult to compute tightly for realistic workloads, while DP-SGD avoids explicit sensitivity calculation but introduces per-example clipping and iterative noise injection that can increase computational overhead and degrade utility if not carefully tuned. 
DP is often an excellent choice in practice, particularly when distribution-free guarantees, robustness to unknown auxiliary information, and mature accounting/composition tools are priorities.

\noindent\textbf{Benefits and costs of PAC/R-PAC Privacy.}
PAC-style frameworks (including R-PAC) offer an alternative way to express and calibrate privacy risk based on stated modeling assumptions and calibration procedures; this can make some forms of distribution-aware calibration and auditing more directly expressible within the framework. The corresponding cost is that guarantees are scenario-specific and must be communicated together with their assumptions (e.g., the instantiated prior/model and calibration procedure), and their composition/interpretation must be handled within that same framework rather than by informal cross-framework parameter comparisons.

\noindent\textbf{Ethical risks of misinterpretation.}
A key ethical risk is \textit{false comparability}: stakeholders may treat a notion with "stricter" worst-case semantics as automatically "more private" than another notion, or may treat reported parameters as interchangeable across frameworks. This can lead to privacy-washing (overstating protection), inappropriate deployment decisions, or misguided regulatory comparisons. We therefore emphasize that privacy claims should be reported together with the underlying assumptions and calibration choices, and interpreted within the chosen framework rather than via informal cross-framework comparisons.

\section{Membership Inference Attack}\label{app:MIA}

We first recall the standard definition of membership inference attacks formalized to match PAC Privacy \cite{xiao2023pac,sridhar2024pac}.

\begin{definition}[Membership Inference Attack \cite{xiao2023pac,sridhar2024pac}]\label{def:MIA}
Given a finite data pool $\mathcal{U} = \{u_1, u_2, \ldots, u_N\}$ and some processing mechanism $\mathcal{M}$, $X$ is an $n$-subset of $\mathcal{U}$ randomly selected. An informed adversary is asked to return an $n$-subset $\hat{X}$ as the membership estimation of $X$ after observing $\mathcal{M}(X)$. We say $\mathcal{M}$ is resistant to $(1-\delta_i)$ individual membership inference for the $i$-th datapoint $u_i$, if for an arbitrary adversary,
$$\Pr_{X \leftarrow \mathcal{U}, \tilde{X} \leftarrow \mathcal{M}(X)}(\mathbf{1}_{u_i \in X} = \mathbf{1}_{u_i \in \hat{X}}) \leq 1 - \delta_i$$
Here, $\mathbf{1}_{u_i \in X}$ ($\mathbf{1}_{u_i \in \hat{X}}$) is an indicator which equals 1 if $u_i$ is in $X$ ($\hat{X}$).
\end{definition}

Building on this attack model, we now introduce the corresponding R-PAC membership privacy notion:

\begin{definition}[R-PAC Membership Privacy]
For a data processing mechanism $\mathcal{M}$, given some measure $\rho$ and a data set $\mathsf{U}=(u_1, u_2, \ldots , u_{N})$, we say $\mathcal{M}$ satisfies $(\mathtt{R}_{f}^{\delta}, \rho, \mathsf{U}, \mathcal{D})$-R-PAC Membership Privacy if it is $(\delta, \rho, \mathsf{U}, \mathcal{D})$ PAC Membership private and:
\begin{equation}
\mathtt{R}_{f}^{\delta} \equiv \mathtt{IntP}_{f}(\mathcal{D}) - \mathtt{D}_f(\mathbf{1}_{\delta} \,\|\, \mathbf{1}_{\delta_\rho^o})
\end{equation}
is the \textup{posterior disadvantage}, where:
\begin{itemize}
\item $\mathtt{IntP}_{f}(\mathcal{D}) = -\mathtt{D}_f(\mathcal{D} \,\|\, \mathsf{U})$ is the intrinsic membership privacy of the sampling distribution $\mathcal{D}$ relative to the uniform distribution over $\mathsf{U}$,
\item $\mathbf{1}_{\delta}$ and $\mathbf{1}_{\delta_\rho^o}$ denote the \textit{posterior} and \textit{prior} inference outcomes, respectively (thought of as binary success/failure indicators; equivalently, Bernoulli distributions with success parameters $1-\delta$ and $1-\delta_\rho^o$),
\item $\delta_\rho^o = \inf_{\tilde{\bm{1}}_{\mathsf{U}}} \Pr_{X \sim \mathcal{D}}[\rho(\tilde{\bm{1}}_{\mathsf{U}}, \bm{1}_{\mathsf{U}}) \neq 1]$ is the optimal \textit{prior} error level (so the optimal prior success is $1-\delta_\rho^o$).
\end{itemize}
\label{def:residual_pac_membership}
\end{definition}

$\mathtt{R}_{f}^{\delta}$ quantifies the R-PAC membership privacy that persists after adversarial inference. The total intrinsic membership privacy is decomposed as:
\begin{equation}
\mathtt{IntP}_{f}(\mathcal{D}) = \mathtt{R}_{f}^{\delta} + \Delta_{f}^{\delta},
\qquad
\text{where }\;
\Delta_{f}^{\delta} \;=\; \mathtt{D}_f(\mathbf{1}_{\delta} \,\|\, \mathbf{1}_{\delta_\rho^o})
\end{equation}
is the PAC Membership Privacy loss, providing a complete accounting of membership privacy risk.

\paragraph{KL case and Markov-chain justification.}
When $\mathtt{D}_f$ is the KL divergence, write $Y=\mathcal{M}(X)$ and let $U_i\in\{0,1\}$ be the membership indicator for an individual $i$, and $J\in\{1,\dots,N\}$ the one-hot index with distribution $\mathcal{D}$ (so $\bm{1}_{\mathsf{U}}$ is the one-hot representation of $J$). Then
\[
\Delta_{f}^{\delta}
\;=\;
\mathrm{KL} \big(\mathbf{1}_{\delta}\,\big\|\,\mathbf{1}_{\delta_\rho^o}\big)
\;\le\; I(U_i;Y)
\;\le\; I(J;Y),
\]
where the first inequality is the Bernoulli–KL information–risk bound for membership, and the second follows because $U_i$ is a deterministic function of $J$ and $U_i \to J \to X \to Y$ is a Markov chain (data processing).
Moreover,
\[
\mathtt{IntP}_{\mathrm{KL}}(\mathcal{D})
\;=\;
\mathcal{H}(\mathcal{D}) - \log |\mathsf{U}|\,.
\]
Combining these,
\[
\begin{aligned}
    \mathtt{R}_{f}^{\delta}
&\;=\;
\mathtt{IntP}_{\mathrm{KL}}(\mathcal{D}) - \Delta_{f}^{\delta}
\;\ge\;
\mathcal{H}(J\,|\,Y) - \log |\mathsf{U}|\\
&\;=\;
\mathcal{H}(\bm{1}_{\mathsf{U}} \mid \mathcal{M}(X)) - \mathtt{V},
\end{aligned}
\]
with $\mathtt{V} \equiv \log |\mathsf{U}|$ independent of both $\mathcal{D}$ and $\mathcal{M}$. This shows that the residual term lower-bounds the conditional uncertainty of the one-hot membership indicator given the mechanism output, up to a constant that depends only on the universe size.

\section{Automatic Efficient PAC Privatization}\label{app:efficient_pac}

PAC privacy (Auto-PAC) provides a framework for measuring privacy risk through simulation-based proofs that bound the mutual information between inputs and outputs of black-box algorithms. While this approach offers rigorous privacy guarantees without requiring white-box algorithm modifications, the original implementation faced computational and practical challenges. The initial algorithm required computing the full covariance matrix and performing Singular Value Decomposition (SVD) across the entire output dimension, which becomes prohibitively expensive for high-dimensional outputs. Additionally, black-box privacy mechanisms suffer from output instability caused by random seeds, arbitrary encodings, or non-deterministic implementations, leading to inconsistent noise calibration and suboptimal utility.

Recent work by Sridhar et al.~\cite{sridhar2024pac} addresses these limitations through Efficient-PAC (Algorithm \ref{alg:PAC_alg_original}), which introduces two key improvements. First, they develop an anisotropic noise calibration scheme that avoids full covariance estimation by projecting mechanism outputs onto a unitary basis and estimating only per-direction variances. This leads to a more scalable and sample-efficient algorithm while maintaining rigorous mutual information guarantees. Second, they propose methods for reducing output instability through regularization and canonicalization techniques, enabling more consistent noise calibration and better overall utility. These refinements are particularly impactful in high-dimensional or structure-sensitive learning tasks, where the original PAC scheme may incur unnecessary noise due to variability not intrinsic to the learning objective.

\begin{algorithm}[t]
\caption{Efficient-PAC \cite{sridhar2024pac}}\label{alg:PAC_alg_original}
\begin{algorithmic}[1]
\REQUIRE deterministic mechanism $\mathcal{M}$, data distribution $\mathcal{D}$, precision parameter $\tau$, convergence function $f_{\tau}$, privacy budget $\beta$, unitary projection matrix $A\in\mathbb{R}^{d\times d}$.
\STATE Initialize $m \gets 1$, $\boldsymbol{\sigma}_0 \gets \text{null}$, $\mathbf{G} \gets \text{null}$
\WHILE{$m \le 2$ or $f_{\tau}(\boldsymbol{\sigma}_{m-1}, \boldsymbol{\sigma}_m) \ge \tau$}
    \STATE Sample $X_m \sim \mathcal{D}$, compute $y_m \gets \mathcal{M}(X_m)$
    \STATE Set $g_m \gets [y_m \cdot A_1, \dots, y_m \cdot A_d]$, append to $\mathbf{G}$
    \STATE Set $\boldsymbol{\sigma}_m[k]$ to empirical variance of column $k$ in $\mathbf{G}$, increment $m \gets m + 1$
\ENDWHILE
\FOR{$i = 1$ to $d$}
    \STATE Set $e_i \gets \frac{\sqrt{\boldsymbol{\sigma}_m[i]}}{2\beta} \sum_{j=1}^d \sqrt{\boldsymbol{\sigma}_m[j]}$
\ENDFOR
\RETURN $\Sigma_B$ with $\Sigma_B[i][i] = e_i$
\end{algorithmic}
\end{algorithm}

Theorem \ref{thm:PAC_Alg_thm1} establishes the privacy guarantee of Efficient-PAC.

\begin{theorem}[Theorem 1 of \cite{sridhar2024pac}]\label{thm:PAC_Alg_thm1}
Let $\mathcal{M} : \mathcal{X} \rightarrow \mathbb{R}^d$ be a deterministic mechanism, and let $A \in \mathbb{R}^{d \times d}$ be a unitary projection matrix. Let $\boldsymbol{\sigma} \in \mathbb{R}^d$ be the variance vector of the projected outputs $\mathcal{M}(X) \cdot A$, and let $B \sim \mathcal{N}(0, \Sigma_B)$ be the additive noise with covariance $\Sigma_B = \operatorname{diag}(e_1, \ldots, e_d)$, where $e_i = \frac{\sqrt{\sigma_i}}{2\beta} \sum_{j=1}^d \sqrt{\sigma_j}$.
Then, the mutual information between the input and privatized output satisfies $\mathtt{MI}(X; \mathcal{M}(X) + B) \leq \beta.$
\end{theorem}

\subsection{Auto-PAC vs.\ Efficient-PAC: Conservativeness}\label{app:Auto_vs_Eff}

Efficient-PAC induces additional conservativeness relative to Auto-PAC. When Efficient-PAC enforces $\mathrm{MI}(X;\mathcal{M}(X)+B)\le \beta$, the proof of Theorem~\ref{thm:PAC_Alg_thm1} in \cite{sridhar2024pac} (Theorem 1) yields
\[
\begin{aligned}
\mathtt{MI}(X;\mathcal{M}(X)+B)
&= \mathtt{MI}(X;\mathcal{M}(X) \cdot  A + B \cdot  A)\\
&\le \tfrac{1}{2}\log\det \Big(I_d + \Sigma_{\mathcal{M}(X)\cdot A}\,\Sigma_B^{-1}\Big)\\
&\le \tfrac{1}{2}\log\det \Big(I_d + \operatorname{diag} \big(\Sigma_{\mathcal{M}(X)\cdot A}\big)\,\Sigma_B^{-1}\Big)\\
&= \tfrac{1}{2}\log \prod_{i} \Big(1+\tfrac{\sigma_i}{e_i}\Big)\\
&= \tfrac{1}{2}\sum_{i}\log \Big(1+\tfrac{\sigma_i}{e_i}\Big)\\
&\le \tfrac{1}{2}\sum_{i}\tfrac{\sigma_i}{e_i}\\
&= \beta,
\end{aligned}
\]
where $\sigma_i=[\operatorname{diag}(\Sigma_{\mathcal{M}(X)\cdot A})]_i$ and $\Sigma_B=\operatorname{diag}(e_1,\ldots,e_d)$. The second inequality is Hadamard's inequality (tight only if $\Sigma_{\mathcal{M}(X)\cdot A}$ is diagonal in the chosen basis), and the last inequality uses $\log(1+x)\le x$ (tight only at $x=0$). Minimizing $\sum_i e_i$ under the linearized constraint $\tfrac12\sum_i \sigma_i/e_i=\beta$ gives the closed form
$e_i=\frac{\sum_j\sqrt{\sigma_j}}{2\beta}\,\sqrt{\sigma_i}$, so that $\sum_{i=1}^d\frac{\sigma_i}{2e_i}=\beta$.
Thus, Efficient-PAC is \textit{weakly more conservative} than Auto-PAC, which (approximately) works in the eigenbasis of $\Sigma_{\mathcal{M}(X)}$ and avoids the Hadamard slack.

Moreover, since $B\cdot A$ is constructed with covariance $\Sigma_{B\cdot A}=A^\top\Sigma_B A$ and $\Sigma_{\mathcal{M}(X)\cdot A}=A^\top\Sigma_{\mathcal M}A$, the exact log-det term is basis-invariant under joint congruence:
\[
\begin{aligned}
    \frac12\log\det \Big(I_d+\Sigma_{\mathcal{M}(X)\cdot A}\,\Sigma_{B\cdot A}^{-1}\Big)
    &=\frac12\log\frac{\det(\Sigma_{B\cdot A}+\Sigma_{\mathcal{M}(X)\cdot A})}{\det(\Sigma_{B\cdot A})}\\
    &=\frac12\log\det \Big(I_d+\Sigma_{\mathcal M}\,\Sigma_B^{-1}\Big)\\
    &= \mathtt{LogDet}(\mathcal{M}(X), B).
\end{aligned}
\]
Therefore, Efficient-PAC implements a budget $\beta$ that upper-bounds the exact Gaussian $\mathtt{LogDet}(\mathcal{M}(X), B)$, with conservativeness decomposing into the Hadamard step and the $\log(1+x)\le x$ linearization.

\begin{remark}
    All our comparisons of Auto-PAC and SR-PAC that rely on the conservativeness of $\mathtt{LogDet}(\mathcal{M}(X), B)$ carry over verbatim for Efficient-PAC because Efficient-PAC implements $\mathtt{LogDet}(\mathcal{M}(X), B)\leq \beta$, where the inequality is in general non-attainable. 
    Thus, conservativeness-related results for $\mathtt{LogDet}(\mathcal{M}(X), B)$ remain valid \textit{a fortiori} for the $\beta$ implemented by Efficient-PAC.
\end{remark}

\begin{remark}
Given any privacy budget, the upper bound implemented by Efficient-PAC induces more conservativeness than directly implementing $\mathtt{LogDet}(\mathcal{M}(X), B)$.
However, there is no universal ordering between the \textit{true} mutual informations $\mathrm{MI}(X;\mathcal{M}(X)+B_{\mathrm{Auto}})$ and $\mathrm{MI}(X;\mathcal{M}(X)+B_{\mathrm{Eff}})$, where $B_{\mathrm{Auto}}$ and $B_{\mathrm{Eff}}$ are the Gaussian noise determined by Auto-PAC and Efficient-PAC for the same privacy budget.
This is because the Gaussianity gaps (explicitly formulated by (\ref{eq:gap_def})) of Auto-PAC and Efficient-PAC can be in general different magnitudes.
\end{remark}

\section{Technical Constructions of Reference Distributions for Intrinsic Privacy
}\label{app:technical_reference}

Intrinsic privacy is defined as
\[
\mathtt{IntP}_f(\mathcal D\Vert\mathcal R)\;=\;-\;\mathtt{D}_f(\mathcal D\Vert\mathcal R),
\]
where $\mathcal R$ is a \textit{reference distribution} that plays the role of an a priori baseline and $\mathtt{D}_f$ is an $f$–divergence (KL in our evaluations).
To make $\mathtt{IntP}_f$ well-defined and mechanism-independent, one must choose $\mathcal R$ so that (i) $\mathrm{supp}(\mathcal D)\subseteq \mathrm{supp}(\mathcal R)$ and (ii) $\mathtt{D}_f(\mathcal D\Vert\mathcal R)<\infty$.
This appendix gives three canonical constructions of $\mathcal R$ together with conditions that guarantee finiteness, and brief practical advice on when to use each choice.

We write $X\sim\mathcal D$ for the data distribution on $\mathbb R^d$.
A reference distribution $\mathcal R$ has density $r(\cdot)$ w.r.t.\ Lebesgue measure (whenever it exists).
For KL, $\mathtt{D}_{\mathrm{KL}}(\mathcal D\Vert\mathcal R)=\mathbb E_{\mathcal D} \big[\ln\frac{d\mathcal D}{d\mathcal R}(X)\big]$ and $H(\mathcal R)$ denotes the (differential) entropy of $\mathcal R$ (log base as in the main text).

\begin{proposition}[Finiteness criteria for $\mathrm{KL}$]
\label{prop:KL_finiteness}
If $\mathcal D\ll\mathcal R$ and $\mathbb E_{\mathcal D} \big[\,|\ln r(X)|\,\big]<\infty$, then $\mathtt{D}_{\mathrm{KL}}(\mathcal D\Vert\mathcal R)<\infty$.
Consequently, any construction of $\mathcal R$ that ensures full support on $\mathbb R^d$ and mild tail control on $r$ suffices for finiteness of $\mathtt{IntP}_{\mathrm{KL}}$.
\end{proposition}

\begin{proof}
Since $\mathcal R$ has Lebesgue density $r$ and $\mathcal D\ll\mathcal R$, we also have $\mathcal D\ll$ Lebesgue; let $p$ denote the Lebesgue density of $\mathcal D$.
By the chain rule for Radon–Nikodym derivatives,
\[
\frac{d\mathcal D}{d\mathcal R}(x) \;=\; \frac{d\mathcal D/dx}{d\mathcal R/dx}(x) \;=\; \frac{p(x)}{r(x)} \quad\text{a.e.}
\]
Hence
\[
\begin{aligned}
    \mathtt{D}_{\mathrm{KL}}(\mathcal D\Vert\mathcal R)
&=
\int p(x)\,\ln \frac{p(x)}{r(x)}\,dx\\
&=
\underbrace{\int p(x)\ln p(x)\,dx}_{-\;\mathcal{H}(\mathcal D)} - \underbrace{\int p(x)\ln r(x) dx}_{\mathbb E_{\mathcal D}[\ln r(X)]}.
\end{aligned}
\]
By assumption, $H(\mathcal D)>-\infty$ and $\mathbb E_{\mathcal D}[\,|\ln r(X)|\,]<\infty$, so both terms on the right-hand side are finite (the first from below, the second in absolute value), and their difference is finite. Therefore $\mathtt{D}_{\mathrm{KL}}(\mathcal D\Vert\mathcal R)<\infty$.
\end{proof}

In the KL case, our residual privacy lower bound involves a constant offset $V=H(\mathcal R)$ (independent of both $\mathcal D$ and the mechanism), so we also highlight when $H(\mathcal R)<\infty$.

\subsubsection*{(a) Maximum-entropy Gaussian}

The maximum-entropy Gaussian is defined as
\[
  \mathcal R \;=\;\mathcal N(\mu,\Sigma),
  \qquad
  \mu=\mathbb E_{\mathcal D}[X],\;
  \Sigma=\mathrm{Cov}_{\mathcal D}(X).
\]
The density function is
\[
  r(x)=\frac{1}{\sqrt{(2\pi)^d\det\Sigma}}\;
  \exp \Bigl(-\tfrac12(x-\mu)^\top\Sigma^{-1}(x-\mu)\Bigr),
\]
with the support $\mathrm{supp}(\mathcal R)=\mathbb R^d$.
The corresponding entropy is 
\[
  H(\mathcal R)=\tfrac12\,\ln \big((2\pi e)^d\det\Sigma\big)\;<\infty.
\]
If $\mathcal D$ is absolutely continuous and $\mathbb E_{\mathcal D} \big[\|X\|^2\big]<\infty$, then $\mathtt{D}_{\mathrm{KL}}(\mathcal D\Vert\mathcal R)<\infty$.

It is a natural default when second moments exist; full support guarantees $\mathrm{supp}(\mathcal D)\subseteq\mathrm{supp}(\mathcal R)$ automatically.
In practice, ensure $\Sigma\succ0$ via standard shrinkage if needed.

\subsubsection*{(b) Smooth pull-back of the unit-cube uniform}

The smooth pull-back construction is defined as follows: let $U\sim\mathrm{Unif}((0,1)^d)$ and choose a $C^1$ bijection
\[
  T:(0,1)^d\to\mathbb R^d,\qquad \det J_T(u)>0.
\]
The reference is the push-forward $\mathcal R=T_{\#}U$ with density
\[
  r(x)=\bigl|\det J_{T^{-1}}(x)\bigr|,
\]
and support $\mathrm{supp}(\mathcal R)=\mathbb R^d$.
The corresponding entropy is
\[
  H(\mathcal R)=\mathbb E_U \big[\ln|\det J_T(U)|\big]\;<\infty
\]
whenever $\ln|\det J_T|$ is integrable on $(0,1)^d$.
If $\mathcal D\ll\mathcal R$ and $\mathbb E_{\mathcal D}[\,|\ln r(X)|\,]<\infty$, then $\mathtt{D}_{\mathrm{KL}}(\mathcal D\Vert\mathcal R)<\infty$.

It is useful when one wishes to encode geometry or tail behavior via the map $T$ while retaining full support and finite $H(\mathcal R)$ through an integrability check on $\ln|\det J_T|$.

\subsubsection*{(c) Truncated uniform on a bounded set}

The truncated uniform is defined as follows: let $B\subset\mathbb R^d$ be compact with $\mathrm{supp}(\mathcal D)\subseteq B$, and set
\[
  \mathcal R=\mathrm{Unif}(B),\qquad
  r(x)=
  \begin{cases}
    1/\mathrm{vol}(B), & x\in B,\\[2pt]
    0, & x\notin B.
  \end{cases}
\]
The support is $\mathrm{supp}(\mathcal R)=B$.
The corresponding entropy is
\[
  H(\mathcal R)=\ln \big(\mathrm{vol}(B)\big)\;<\infty.
\]
Moreover,
\[
  \mathtt{D}_{\mathrm{KL}}(\mathcal D\Vert\mathcal R)
  = -\,H(\mathcal D) + \ln \big(\mathrm{vol}(B)\big),
\]
so finiteness requires $H(\mathcal D)<\infty$.

It is appropriate only when the domain is naturally bounded and the data distribution has finite entropy; otherwise, one should prefer the Gaussian or pull-back constructions.

\section{More on Non-Gaussianity Correction}\label{app:non_gaussianity_correction}

In Section \ref{sec:gap_reduction}, we propose two approaches to approximate the Gaussianity gap $\mathtt{Gap}_{\mathtt{d}}$, which are certified replacements of $\mathtt{D}_{Z}$ to find a tighter mutual information after Auto-PAC privatization.
Theorem \ref{thm:DV} uses Donsker–Varadhan (DV) representation $\mathtt{D}_Z=\sup_f\{\mathbb{E}_{P_{\mathcal{M},B}}f-\log\mathbb E_{\widetilde Q_{\mathcal{M}}}e^f\}$, so that any value of the DV objective at a trained critic $f_\psi$ is a valid lower bound on $\mathtt{D}_{Z}$. 
Under a mild transport condition for $P_{\mathcal{M}, B}$ and $\widetilde{Q}_{\mathcal{M}}$, Theorem \ref{thm:SWD} use the sliced Wasserstein distance (SWD) as the estimation $\widehat{\mathtt{D}}_{Z}$, which is unbiased in the minibatch limit. 
In addition, the estimation achieves a certified $0\leq \widehat{\mathtt{D}}_{Z}\leq \mathtt{D}_{Z}$.

Consequently, our improved mutual information estimate
\[
\mathtt{IMI}(\widehat{\mathtt{D}}_Z) = \mathtt{LogDet}(\mathcal{M}(X),B) - \widehat{\mathtt{D}}_Z
\]
is a provable upper bound on $\mathtt{MI}(X;Z)$ whenever $\widehat{\mathtt{D}}_Z$ is one of the certified corrections above.

\begin{algorithm}[ht]
\caption{DV Gap Correction (minibatch lower bound on $\mathtt{D}_Z$)}
\label{alg:dv_gap_minimal}
\begin{algorithmic}[1]
\REQUIRE Oracle for i.i.d. samples $Z\sim P_{\mathcal{M},B}$; function class $\mathcal{F}=\{f_{\phi}\}$; steps $T$; batch size $m$; step size $\eta$; confidence penalty $c_{m,\delta}$
\STATE Draw an initial batch $\{Z_i\}_{i=1}^{m_0}\sim P_{\mathcal{M},B}$ and estimate $\widehat{\mu}_Z,\widehat{\Sigma}_Z$; define $\widetilde{Q}_{\mathcal{M}}\equiv\mathcal{N}(\widehat{\mu}_Z,\widehat{\Sigma}_Z)$
\STATE Initialize $\phi$
\FOR{$t=1,\dots,T$}
  \STATE Sample $\{Z^{(P)}_j\}_{j=1}^{m}\sim P_Z$
  \STATE Sample $\{Z^{(Q)}_j\}_{j=1}^{m}\sim \widetilde{Q}_{\mathcal{M}}$
  \STATE $\widehat{b} \leftarrow \frac{1}{m}\sum_{j=1}^{m} f_{\phi}(Z^{(P)}_j) - \log\!\Big(\frac{1}{m}\sum_{j=1}^{m} e^{f_{\phi}(Z^{(Q)}_j)}\Big)$
  \STATE $\phi \leftarrow \phi + \eta\,\nabla_{\phi}\widehat{b}$
\ENDFOR
\STATE Evaluate $\widehat{b}_{\mathrm{val}}$ on held-out minibatches; set $\underline{\mathtt{D}}_Z \leftarrow \max\{\widehat{b}_{\mathrm{val}} - c_{m,\delta},0\}$
\STATE \textbf{Return} $\underline{\mathtt{D}}_Z$
\end{algorithmic}
\end{algorithm}

\begin{algorithm}[ht]
\caption{Sliced Wasserstein Gap Correction (training-free lower bound on $\mathtt{D}_Z$)}
\label{alg:sw_gap_minimal}
\begin{algorithmic}[1]
\REQUIRE Oracle for i.i.d. samples $Z\sim P_{\mathcal{M},B}$; number of projections $M$; samples per slice $n$; confidence penalty $\xi_{n,\delta}$
\STATE Draw an initial batch $\{Z_i\}_{i=1}^{n_0}\sim P_{\mathcal{M},B}$ and estimate $\widehat{\mu}_Z,\widehat{\Sigma}_Z$; set $W\leftarrow \widehat{\Sigma}_Z^{-1/2}$
\FOR{$m=1,\dots,M$}
  \STATE Draw $\theta_m$ uniformly on $\mathbb{S}^{d-1}$
  \STATE Draw $n$ fresh samples $Z_i\sim P_Z$ and set $u_i=\theta_m^{\top}W(Z_i-\widehat{\mu}_Z)$
  \STATE Draw $n$ i.i.d. samples $s_i\sim \mathcal{N}(0,1)$
  \STATE Sort $u_{(1)}\le\cdots\le u_{(n)}$ and $s_{(1)}\le\cdots\le s_{(n)}$; set $w_m^{2}\leftarrow \frac{1}{n}\sum_{i=1}^{n}(u_{(i)}-s_{(i)})^{2}$
\ENDFOR
\STATE $\widehat{\mathrm{SW}}_{2}^{2}\leftarrow \frac{1}{M}\sum_{m=1}^{M} w_m^{2}$
\STATE $\underline{\mathtt{D}}_Z \leftarrow \max\{\frac{1}{2}\widehat{\mathrm{SW}}_{2}^{2}-\xi_{n,\delta},0\}$
\STATE \textbf{Return} $\underline{\mathtt{D}}_Z$
\end{algorithmic}
\end{algorithm}

Both approaches admit short, minibatch estimators:
\begin{itemize}
  \item \textbf{DV Correction.} Train a critic $f_{\phi}$ by maximizing
  \[
    \widehat{\mathcal{J}}
    =
    \frac{1}{m}\sum_{i=1}^{m} f_{\phi}(Z_i)
    - \log\!\Big(\frac{1}{m}\sum_{i=1}^{m} e^{\,f_{\phi}(\widetilde{Z}_i)}\Big),
  \]
  where $Z_i\sim P_Z$ and $\widetilde{Z}_i\sim \widetilde{Q}_{\mathcal{M}}$.
  After $T$ steps, set $\widehat{\mathtt{D}}_Z \leftarrow \widehat{\mathcal{J}}(f_{\phi})$.
  Algorithm \ref{alg:dv_gap_minimal} shows an example.
  \item \textbf{SWD Correction.} Sample $K$ directions $v_k \sim \mathrm{Unif}(\mathbb{S}^{d-1})$, project both batches, sort each projection, and average 1D squared distances:
  \[
    \widehat{\mathrm{SW}}_2^2
    =
    \frac{1}{K}\sum_{k=1}^{K}
    \frac{1}{m}\sum_{j=1}^{m}
    \big(\langle v_k,Z\rangle_{(j)} - \langle v_k,\widetilde{Z}\rangle_{(j)}\big)^2.
  \]
  Convert $\widehat{\mathrm{SW}}_2^2$ to $\widehat{\mathtt{D}}_Z$ using the calibration stated in Theorem~\ref{thm:SWD}.
  Algorithm \ref{alg:sw_gap_minimal} gives an example.
\end{itemize}
Each iteration uses a single minibatch pass and either a small critic update (DV) or $K$ sorts of length $m$ (SWD); no backpropagation through $\mathcal{M}$ and no nested inner loops.

\section{Finite-Sample Guarantees and Robustness for SR-PAC}

\subsection{Follower Generalization and Approximate Optimization}\label{app:finite_errors}

Fix a perturbation rule $Q\in\Gamma$. The Follower's objective is
\[
\pi^*(Q)\in\arg\min_{\pi\in\Pi} W(Q,\pi),
\]
where
\[
\begin{aligned}
    W(Q,\pi) \equiv \mathbb{E}_{X\sim \mathcal{D},\,B\sim Q}\big[-\log \pi(X\mid \mathcal{M}(X)+B)\big].
\end{aligned}
\]
Given $m$ i.i.d. samples $(X_i,B_i,Y_i)_{i=1}^m$ with $X_i\sim\mathcal{D}$, $B_i\sim Q$, and
$Y_i=M(X_i)+B_i$, define the empirical risk
\[
\widehat W(Q,\pi)\;=\;\frac{1}{m}\sum_{i=1}^m \big[-\log \pi(X_i\mid Y_i)\big],
\quad
\hat\pi\in\arg\min_{\pi\in\Pi}\widehat W(Q,\pi).
\]
Let $G_\Pi \equiv \{\, g_\pi(x,y) = -\log \pi(x\mid y):\pi\in\Pi\}$ and denote by
$\widehat{\mathcal{R}}_m(G_\Pi)$ the empirical Rademacher complexity of $G_\Pi$ on $m$ samples.

\begin{assumption}[bounded log-likelihood]\label{assp:bounded_log_likelihood}
There exists $B>0$ such that for all $\pi\in\Pi$ and all $(x,y)$ in the support,
$-\log \pi(x\mid y)\in[0,B]$.
\end{assumption}
When densities are unbounded, Assumption \ref{assp:bounded_log_likelihood} is enforced by standard truncation or by lower-bounding the decoder's variance/softmax temperature over a bounded input domain.

\begin{lemma}[Follower's Decoder PAC generalization]\label{lem:decoder-pac}
Fix a Leader's perturbation rule $Q$. Draw i.i.d. samples $(X_{i}, B_{i}, Y_{i})$ with $X_{i}\sim \mathcal{D}, B_{i}\sim Q$, $Y_{i} = \mathcal{M}(X_{i}) + B_{i}$.
Under Assumption \ref{assp:bounded_log_likelihood}, for any $\delta\in(0,1)$, with probability at least $1-\delta$ over the draw of the $m$ samples,
\[
\begin{aligned}
    \bigg|\inf_{\pi\in\Pi} W(Q,\pi) -\widehat W\big(Q,\hat\pi\big)\bigg| 
    &\leq \underbrace{4\,\widehat{\mathcal{R}}_m(G_\Pi)}_{\textup{capacity}}
+
\underbrace{2B\sqrt{\tfrac{2\log(1/\delta)}{m}}}_{\textup{concentration}}\\
&\equiv \varepsilon_{m,\delta}.
\end{aligned}
\]
\end{lemma}

\begin{proof}

Let $G_{\Pi}=\{g_\pi(x,y)=-\log \pi(x\mid y):\pi\in\Pi\}$ with $g_\pi\in[0,B]$ by assumption, and let
\[
\widehat{\mathcal{R}}_m(G_\Pi)\;\equiv\;
\mathbb{E}_{\sigma}\Big[\sup_{g\in G_\Pi}\frac{1}{m}\sum_{i=1}^m \sigma_i\, g(X_i,Y_i)\Big]
\]
be the (empirical) Rademacher complexity on the sample $(X_i,Y_i)_{i=1}^m$, where $\sigma_i\in\{\pm1\}$ are i.i.d. Rademacher variables.
By standard symmetrization and McDiarmid's inequality (bounded differences $B/m$), with probability at least $1-\delta$,
\begin{equation}\label{eq:uniform-deviation}
\sup_{\pi\in\Pi}\big|W(Q,\pi)-\widehat W(Q,\pi)\big|
\;\le\; 2\,\widehat{\mathcal{R}}_m(G_\Pi)\;+\; B\sqrt{\tfrac{2\log(1/\delta)}{m}}.
\end{equation}
On the same event, let $\pi^*\in\arg\min_{\pi}W(Q,\pi)$ and $\hat\pi\in\arg\min_{\pi}\widehat W(Q,\pi)$. Then
\[
\begin{aligned}
    \inf_{\pi}W(Q,\pi)-\widehat W(Q,\hat\pi)
&= W(Q,\pi^*)-\widehat W(Q,\hat\pi)\\
&\geq - \sup_{\pi} \big|W-\widehat W\big| -\sup_{\pi}\big|W-\widehat W\big|\\
&\geq -2\Delta,
\end{aligned}
\]
where $\Delta=2\,\widehat{\mathcal{R}}_m(G_\Pi)+B\sqrt{2\log(1/\delta)/m}$ is the right-hand side of \eqref{eq:uniform-deviation}.
Thus
\[
\begin{aligned}
    \Bigl|\inf_{\pi\in\Pi} W(Q,\pi)\;-\;\widehat W\big(Q,\hat\pi\big)\Bigr|
&\leq 2\Delta\\
&=4\,\widehat{\mathcal{R}}_m(G_\Pi)\;+\;2B\sqrt{\tfrac{2\log(1/\delta)}{m}}.
\end{aligned}
\]

\end{proof}

\begin{algorithm}[ht]
\caption{Monte Carlo SR-PAC (with PAC-adjusted Penalty)}\label{alg_app_SR_PAC}
\begin{algorithmic}[1]
\REQUIRE Privacy budget $\hat{\beta}$, parametrized decoder family $\Pi_{\phi}$, 
         perturbation rule family $\Gamma_{\lambda}$, utility loss $\mathcal{K}(\cdot)$, 
         learning rates $\eta_{\phi}, \eta_{\lambda}$, penalty weight $\sigma$, 
         iterations $T_{\lambda}, T_{\phi}$, batch size $m$
\STATE Initialize parameters $\lambda, \phi \sim \text{init}()$
\FOR{$t = 1, \ldots, T_{\lambda}$}
    \IF{$t \bmod T_{\phi} = 0$}
        \STATE \textbf{Update Decoder:}
        \FOR{$i = 1, \ldots, T_{\phi}$}
            \STATE Sample $\{(x_j, b_j, y_j)\}_{j=1}^m$ where $x_j \sim \mathcal{D}$, $b_j \sim Q_{\lambda}$, $y_j = \mathcal{M}(x_j) + b_j$
            \STATE $\widehat{W} = \frac{1}{m}\sum_{j=1}^m [-\log \pi_{\phi}(x_j | y_j)]$
            \STATE $\phi \leftarrow \phi - \eta_{\phi} \nabla_{\phi} \widehat{W}$
        \ENDFOR
    \ENDIF
    \STATE \textbf{Update Perturbation Rule:}
    \STATE Sample $\{(x_j, b_j, y_j)\}_{j=1}^m$ where $x_j \sim \mathcal{D}$, $b_j \sim Q_{\lambda}$, $y_j = \mathcal{M}(x_j) + b_j$
    \STATE $H_c = \frac{1}{m}\sum_{j=1}^m [-\log \pi_{\phi}(x_j | y_j)]$ 
    \STATE $\mathcal{L}_{\lambda} = \frac{1}{m}\sum_{j=1}^m \mathcal{K}(b_j) + \sigma \big(H_c - (\hat{\beta} + \varepsilon_{m,\delta})\big)^2_{+}$
    \STATE $\lambda \leftarrow \lambda - \eta_{\lambda} \nabla_{\lambda} \mathcal{L}_{\lambda}$
\ENDFOR
\RETURN Optimal parameters $(\lambda^*, \phi^*)$
\end{algorithmic}
\end{algorithm}

\paragraph{Approximate follower optimization.}
In practice, the decoder (Follower) update may return an $\varepsilon_{\mathrm{opt}}$-approximate minimizer
$\tilde\pi$ of the empirical objective, i.e.,
$\widehat W(Q,\tilde\pi)\le \inf_{\pi\in\Pi}\widehat W(Q,\pi)+\varepsilon_{\mathrm{opt}}$.
On the same event as Lemma~\ref{lem:decoder-pac}, we then have
\[
\inf_{\pi\in\Pi} W(Q,\pi)
\;\ge\;
\widehat W(Q,\tilde\pi) - \varepsilon_{m,\delta}-\varepsilon_{\mathrm{opt}} .
\]

\begin{corollary}[Finite-Sample Feasibility for Leader]\label{cor:finite_sample_leader}
Let $\hat\beta$ be the residual-PAC budget in the Leader's constraint
$\inf_{\pi\in\Pi}W(Q,\pi)\geq \hat\beta$.
If the batch cross-entropy $H_c=\widehat W(Q,\tilde{\pi})$, where $\tilde{\pi}$ satisfies $\widehat{W}(Q,\tilde\pi)\leq \inf_{\pi\in\Pi}\widehat{W}(Q,\pi)+\varepsilon_{\mathrm{opt}}$ (we take $\tilde{\pi}$ to be the decoder $\pi_{\phi}$ used in Algorithm~\ref{alg_app_SR_PAC}), satisfies
\[
H_c \;\geq\; \hat\beta + \varepsilon_{m,\delta} +\varepsilon_{\mathrm{opt}},
\]
then, with probability at least $1-\delta$,
$\inf_{\pi\in\Pi}W(Q,\pi)\geq \hat\beta$.
\end{corollary}

\paragraph{PAC-adjusted penalty.}
Define the PAC-adjusted threshold
\[
\hat{\beta}_{\mathrm{PAC}} \equiv\hat\beta+\varepsilon_{m,\delta} + \varepsilon_{\mathrm{opt}}.
\]
In the ideal ERM case, we have $\varepsilon_{\mathrm{opt}}$ and then $\hat\beta_{\mathrm{PAC}} =\hat\beta+\varepsilon_{m,\delta}$.
A convenient implementation is to use $\hat\beta_{\mathrm{PAC}}$ \textit{in place of} $\hat\beta$ inside the Leader's penalty; i.e., set
\[
\textup{penalty}= \sigma \bigl(H_c-\hat{\beta}_{\mathrm{PAC}}\bigr)_+^2,
\]
where $\sigma$ is the penalty weight and $(\cdot)_+=\max\{\cdot,0\}$.
Algorithm \ref{alg_app_SR_PAC} shows the SR-PAC with the PAC-adjusted penalty.
By Corollary~\ref{cor:finite_sample_leader}, any iterate with zero penalty (or sufficiently small penalty when smooth proxies are used) satisfies the population constraint with probability at least $1-\delta$ at sample size $m$.

\paragraph{Sample complexity (reading $\varepsilon_{m,\delta}$).}
If $\widehat{\mathcal{R}}_m(G_\Pi)\le C/\sqrt m$, then any
\[
m\;\;\geq\;\;\Bigl(\tfrac{4C}{\eta}+\tfrac{2B\sqrt{2\log(1/\delta)}}{\eta}\Bigr)^2
\]
ensures $\varepsilon_{m,\delta}\le\eta$, so enforcing $H_c\geq \hat{\beta}+\eta$ certifies feasibility with probability $\geq 1-\delta$.
By Corollary~\ref{cor:finite_sample_leader}, this variant certifies feasibility with probability at least $1-\delta$ at sample size $m$.

\subsection{Robustness of Sensitivity Conclusions in Theorem~\ref{thm:speed_general}}
\label{app:implemented_curves}

Theorem~\ref{thm:speed_general} characterizes the \textit{ideal} sensitivity behavior of SR-PAC through the
population-optimal curves $V_{\mathrm{SR}}(\beta)$ and $\mathtt{MI}_{\mathrm{SR}}(\beta)$.
In practice, SR-PAC is implemented with finite samples and Monte-Carlo estimates, and the inner/outer
optimization may terminate before reaching the exact Stackelberg optimum.
This subsection defines implementation-level error quantities for interpreting how the conclusions of Theorem~\ref{thm:speed_general} behave under finite-sample estimation and optimization effects.

\paragraph{Implemented SR-PAC curves.}
Let $\widehat Q_{\mathrm{SR}}(\beta)$ denote the (possibly non-Gaussian) noise distribution returned by the
SR-PAC solver when targeting privacy budget $\beta$, and let $B\sim \widehat Q_{\mathrm{SR}}(\beta)$.
Define the \textit{implemented} noise-power and leakage curves
\[
\widehat V_{\mathrm{SR}}(\beta)\;\equiv\;\mathbb{E}\big[\|B\|_2^2\big],\qquad
\widehat{\mathtt{MI}}_{\mathrm{SR}}(\beta)\;\equiv\;\mathtt{MI}\big(X;\mathcal{M}(X)+B\big).
\]
Accordingly, define the implemented sensitivities
\[
\widehat{\mathtt{Priv}}^{\mathrm{SR}}_{\beta}\;\equiv\;\frac{d}{d\beta}\widehat{\mathtt{MI}}_{\mathrm{SR}}(\beta),
\qquad
\widehat{\mathtt{Util}}^{\mathrm{SR}}_{\beta}\;\equiv\;\frac{d}{d\beta}\bigl(-\widehat V_{\mathrm{SR}}(\beta)\bigr).
\]

\paragraph{Calibration and optimization errors.}
We separate two sources of deviation from the ideal curves:
\begin{itemize}
\item \textbf{Calibration error:}
\[
\varepsilon_{\mathrm{cal}}(\beta)\;\equiv\;\widehat{\mathtt{MI}}_{\mathrm{SR}}(\beta)\;-\;\beta.
\]
This captures (true) leakage-budget misalignment at target $\beta$.
In implementations, $\varepsilon_{\mathrm{cal}}(\beta)$ is controlled indirectly via finite-sample/Monte-Carlo
estimates of $\mathtt{MI}$ together with concentration bounds.
\item \textbf{Optimization suboptimality:}
\[
\eta_{\mathrm{opt}}(\beta)\;\equiv\;\widehat V_{\mathrm{SR}}(\beta)\;-\;V_{\mathrm{SR}}(\beta)\;\ge 0,
\]
i.e., the excess noise power relative to the population-optimal SR-PAC curve. 
Here, $\eta_{\mathrm{opt}}(\beta)$ is global suboptimality of the returned noise curve, and it is different from $\varepsilon_{\mathrm{opt}}$ used earlier for approximate Follower training at a fixed $Q$.
\end{itemize}

To compare implemented SR-PAC sensitivities against Auto-PAC, recall the Auto-PAC curve $V_{\mathrm{PAC}}(\beta)\equiv \mathrm{tr}(\Sigma_{B_{\mathrm{PAC}}}(\beta))$, where
$Q(\beta)=\mathcal{N}(0,\Sigma_{B_{\mathrm{PAC}}}(\beta))$ solves
$\mathtt{LogDet}(\mathcal{M}(X),B_{\mathrm{PAC}})=\beta$, and the corresponding leakage curve
$\mathtt{MI}_{\mathrm{PAC}}(\beta)\equiv \beta-\mathtt{Gap}_{\mathrm d}(Q(\beta))$ as defined in the paragraph above
Theorem~\ref{thm:speed_general}.

\begin{corollary}[Robustness of sensitivity conclusions]
\label{cor:srpac_robustness}
In the setting of Theorem~\ref{thm:speed_general}, assume
\[
\varepsilon_{\mathrm{cal}}(\beta)\in\big[0,\mathtt{Gap}_{\mathrm d}(Q(\beta))\big),\quad
\eta_{\mathrm{opt}}(\beta)\in\big[0,\,V_{\mathrm{PAC}}(\beta)-V_{\mathrm{SR}}(\beta)\big).
\]
In addition, assume $\varepsilon_{\mathrm{cal}}$ is differentiable in $\beta$.
Then:
\begin{itemize}
\item[(i)] The implemented privacy sensitivity remains close to perfect budget alignment:
\[
\bigl|\widehat{\mathtt{Priv}}^{\mathrm{SR}}_{\beta}-1\bigr|\;\le\;\bigl|\varepsilon'_{\mathrm{cal}}(\beta)\bigr|.
\]
\item[(ii)] The utility-sensitivity advantage of SR-PAC persists:
\[
\widehat{\mathtt{Util}}^{\mathrm{SR}}_{\beta}\;\ge\;\mathtt{Util}^{\mathrm{PAC}}_{\beta},
\]
with equality only in the jointly Gaussian case (equivalently, when the Auto-PAC bound is tight).
\end{itemize}
\end{corollary}

Part~(i) shows that SR-PAC retains (near) one-for-one budget-to-leakage control whenever the calibration error
varies smoothly with $\beta$.
Part~(ii) states that as long as the optimization suboptimality $\eta_{\mathrm{opt}}(\beta)$ is smaller than the
theoretical gap $V_{\mathrm{PAC}}(\beta)-V_{\mathrm{SR}}(\beta)$, SR-PAC continues to convert privacy budget into
utility at least as efficiently as Auto-PAC.

\section{Sliced R-PAC and Sliced SR-PAC}\label{app:sliced_rpac}

In this section, we introduce sliced variants, \textit{Sliced R-PAC Privacy} and \textit{Sliced SR-PAC}, to improve scalability in high-dimensional settings, including both high-dimensional data (or secret) spaces and high-dimensional mechanism outputs.
The sliced variants are based on \textit{sliced mutual information (SMI)} \cite{goldfeld2021sliced}.

\subsection{Preliminaries: Sliced Mutual Information}

Before introducing SMI, we set up some important notations and notions.
Let $\mathcal{P}(\mathbb{R}^d)$ denote the set of all Borel probability measures on $\mathbb{R}^d$.
Throughout, we primarily focus on absolutely continuous random variables.
In addition, we use $X$ and $Y$ for two arbitrary random variables.
Given a map $f:\mathbb{R}^d\to\mathbb{R}^{d'}$ and a distribution $P_X\in\mathcal{P}(\mathbb{R}^d)$, we write $f_\sharp P_X$ for the pushforward of $P_X$ under $f$, defined by
$f_\sharp P_X(A)=P_X\big(f^{-1}(A)\big)$ for measurable sets $A$.
Let $\mathbb{S}^{d-1}$ be the unit sphere in $\mathbb{R}^d$, with surface area $S_{d-1}=2\pi^{d/2}/\Gamma(d/2)$, where $\Gamma$ is the gamma function.
Finally, we define the slice in direction $\theta$ by $\pi^\theta(x)\equiv\theta^\top x$.

\begin{definition}[Sliced MI \cite{goldfeld2021sliced}]\label{def:SMI}
Let $(X,Y)\sim P_{X,Y}\in \mathcal{P}(\mathbb{R}^{d_x}\times \mathbb{R}^{d_y})$.
Draw $\Theta\sim \mathrm{Unif}(\mathbb{S}^{d_x-1})$ and $\Phi\sim \mathrm{Unif}(\mathbb{S}^{d_y-1})$ independently, and independently of $(X,Y)$.
The sliced mutual information (SMI) between $X$ and $Y$ is defined by
\begin{equation}\label{eq:SMI}
\begin{aligned}
    &\mathtt{SMI}(X;Y)
\equiv \mathtt{MI}(\Theta^\top X;\Phi^\top Y \mid \Theta,\Phi)\\
&= \frac{1}{S_{d_x-1}S_{d_y-1}}
\oint_{\mathbb{S}^{d_x-1}}\oint_{\mathbb{S}^{d_y-1}}
\mathtt{MI}(\theta^\top X;\phi^\top Y) \, d\theta \, d\phi.
\end{aligned}
\end{equation}
\end{definition}

Intuitively, SMI measures dependence between high-dimensional variables by averaging the mutual information of their one-dimensional random projections.
By the data processing inequality, $\mathtt{SMI}(X;Y)\leq \mathtt{MI}(X;Y)$ \cite{goldfeld2021sliced}, so this slicing process necessarily discards some information.
Even so, it has been shown that SMI retains several central features of mutual information, including the ability to distinguish independence from dependence, as well as analogues of the chain rule and entropy decompositions \cite{goldfeld2021sliced}.
Furthermore, SMI also enjoys the variational representation, in terms of an optimization, similar to the Donsker-Varadhan respesentation of MI (See Proposition 3 of \cite{goldfeld2021sliced}).

\begin{definition}[Sliced Entropy and Sliced Conditional Entropy \cite{goldfeld2021sliced}]\label{DEF:SH}
Let $(X,Y)\sim P_{X,Y}\in \mathcal{P}(\mathbb{R}^{d_x}\times \mathbb{R}^{d_y})$.
Draw $\Theta\sim \mathrm{Unif}(\mathbb{S}^{d_x-1})$ and $\Phi\sim \mathrm{Unif}(\mathbb{S}^{d_y-1})$ independently and independently of $(X,Y)$.
The \textit{sliced entropy} of $X$ is defined as
\[
\begin{aligned}
    \mathtt{SH}(X) &\equiv \mathcal{H}(\Theta^\top X \mid \Theta)\\
&= \mathbb{E}_{\Theta}\!\left[\mathcal{H}(\theta^\top X)\right],
\end{aligned}
\]
and the \textit{sliced conditional entropy} of $X$ given $Y$ is
\[
\begin{aligned}
    \mathtt{SH}(X\mid Y)
&\equiv \mathcal{H}\!\big(\Theta^\top X \,\big|\, \Theta,\Phi,\Phi^\top Y\big)\\
&= \mathbb{E}_{\Theta,\Phi}\!\left[\mathcal{H}\!\big(\theta^\top X \,\big|\, \phi^\top Y\big)\right].
\end{aligned}
\]
\end{definition}

Conceptually, sliced entropy captures the average uncertainty in random one-dimensional projections of $X$.
The sliced conditional entropy $\mathtt{SH}(X\mid Y)$ represents the residual uncertainty in $\Theta^\top X$ once both the projection direction and the corresponding one-dimensional projection of $Y$ (namely, $\Phi$ and $\Phi^\top Y$) are revealed.

\subsection{Sliced R-PAC Privacy}

We now formally define Sliced R-PAC Privacy, a sliced variant of R-PAC Privacy that quantifies residual privacy from a one-dimensional lens.
Sliced R-PAC replaces the original high-dimensional posterior disadvantage with an average over one-dimensional random projections.
Intuitively, slicing reduces the dimensionality of both the secret space and the mechanism output space, thereby improving scalability when either (i) the secret space is high-dimensional, or (ii) the released mechanism output is high-dimensional, or both.
Importantly, the slicing in Sliced R-PAC privacy is purely a tool for privacy quantification: we do not slice the underlying data or the mechanism outputs themselves.

Let $X\sim \mathcal{D}\in\mathcal{P}(\mathbb{R}^{d_x})$ with $d_x\ge 1$, and let $Y=\mathcal{M}(X)\in\mathbb{R}^{d_y}$ with $d_y\ge 1$.
Then $(X,Y)\sim P_{X,\mathcal{M}(X)}\in\mathcal{P}(\mathbb{R}^{d_x}\times\mathbb{R}^{d_y})$.
Let $\Theta\sim \mathrm{Unif}(\mathbb{S}^{d_x-1})$ and $\Phi\sim \mathrm{Unif}(\mathbb{S}^{d_y-1})$ be independent and independent of $(X,Y)$.
For each $(\theta,\phi)$, define the one-dimensional projections
\[
X_{\theta} \equiv \theta^\top X, \qquad Y_{\phi} \equiv \phi^\top Y.
\]
For random directions, we write $X_\Theta\equiv\Theta^\top X$ and $Y_\Phi\equiv\Phi^\top Y$.
We write $P_{X_{\theta}}$ and $P_{X_{\theta}\mid Y_{\phi}}$ for the corresponding (projected) marginal and posterior distributions.

\paragraph{Sliced intrinsic privacy.}
Let $\mathcal{R}$ be a fixed reference distribution on $\mathcal{X}\subseteq \mathbb{R}^{d_x}$ satisfying
$\mathrm{supp}(\mathcal{D})\subseteq\mathrm{supp}(\mathcal{R})$ and
$\mathtt{D}_{f}(\mathcal{D}\|\mathcal{R})<\infty$.
For each direction $\theta$, define the pushed-forward laws
$\mathcal{D}_\theta \equiv (\pi^\theta)_\sharp \mathcal{D}$ and $\mathcal{R}_\theta \equiv (\pi^\theta)_\sharp \mathcal{R}$ on $\mathbb{R}$.
We define the \textit{sliced intrinsic privacy} by
\begin{equation}\label{eq:sl_IntP}
\mathtt{IntP}^{\mathrm{sl}}_{f}(\mathcal{D})
\equiv
-\,\mathbb{E}_{\Theta}\Big[\mathtt{D}_{f}\big(\mathcal{D}_{\Theta}\,\|\,\mathcal{R}_{\Theta}\big)\Big].
\end{equation}

\paragraph{Sliced posterior advantage and posterior disadvantage.}
For each pair $(\theta,\phi)$, viewing $Y_\phi$ as the one-dimensional disclosed output, let $\delta^{\rho}_{o}(\theta,\phi)$ denote the adversary's optimal inference success after observing $Y_\phi$ under the same PAC success criterion as in Definition~\ref{def:PAC_advantage}.
Define the corresponding \textit{sliced posterior advantage} by
\begin{equation}\label{eq:sl_Delta}
\Delta_{f}^{\delta,\mathrm{sl}}
\equiv
\mathbb{E}_{\Theta,\Phi}\Big[\mathtt{D}_f\big(\mathbf{1}_\delta \,\|\, \mathbf{1}_{\delta^{\rho}_{o}(\Theta,\Phi)}\big)\Big].
\end{equation}
We then define the \textit{sliced posterior disadvantage} (residual guarantee) by
\begin{equation}\label{eq:sl_R}
\mathtt{R}_{f}^{\delta,\mathrm{sl}}
\equiv
\mathtt{IntP}^{\mathrm{sl}}_{f}(\mathcal{D})
-
\Delta_{f}^{\delta,\mathrm{sl}}.
\end{equation}
Consequently, the sliced intrinsic privacy admits the exact decomposition
\begin{equation}\label{eq:sl_decomp}
\mathtt{IntP}^{\mathrm{sl}}_{f}(\mathcal{D})
=
\mathtt{R}_{f}^{\delta,\mathrm{sl}}
+
\Delta_{f}^{\delta,\mathrm{sl}}.
\end{equation}

\begin{definition}[Sliced-R-PAC Privacy]\label{def:sl_rpac}
A mechanism $\mathcal{M}$ is \textit{$(\mathtt{R}_{f}^{\delta,\mathrm{sl}},\rho,\mathcal{D})$ Sliced-R-PAC private}
if it is $(\delta,\rho,\mathcal{D})$ PAC private and its sliced posterior disadvantage is given by
\eqref{eq:sl_R}, i.e.,
$\mathtt{R}_{f}^{\delta,\mathrm{sl}}
\equiv
\mathtt{IntP}^{\mathrm{sl}}_{f}(\mathcal{D})
-
\Delta_{f}^{\delta,\mathrm{sl}}$.
\end{definition}

Sliced-R-PAC Privacy shares the same semantics as PAC and R-PAC Privacy frameworks but uses a different privacy quantification measure.

\paragraph{KL Divergence and Sliced Conditional Entropy.}
When $\mathtt{D}_{f} = \mathtt{D}_{\mathtt{KL}}$ and $\mathcal{H}(X)$ is finite, we can obtain a clean information-theoretic characterization similar to Corollary~\ref{prop:Residual_PAC_basic}.
In particular, a $\mathcal{M} : \mathcal{X} \rightarrow \mathcal{Y}$ satisfies $(\mathtt{R}_{f}^{\delta,\mathrm{sl}},\rho,\mathcal{D})$ Sliced-R-PAC private if 
\[
\mathtt{R}_{f}^{\delta,\mathrm{sl}}\geq \mathtt{SH}(X\mid Y) - \mathcal{H}(\mathcal{R}_{\Theta}).
\]
Thus, up to a constant (i.e., $\mathcal{H}(\mathcal{R}_{\Theta})$) determined by the reference distribution $\mathcal{R}$ and $\Theta$, the KL-divergence instantiation of Sliced-R-PAC can be fully captured by the sliced conditional entropy $\mathtt{SH}(X\mid Y)$.

\subsection{Sliced SR-PAC Privatization}\label{sec:sl_srpac}

We now present a sliced variant of SR-PAC for automatic privatization in high dimensions, which uses slicing to define a scalable optimization to implement conditional entropy constraints.

Consider output perturbation with $Y=\mathcal{M}(X)+B$, where $B\sim Q\in\Gamma$.
Let $\Pi$ denote a decoder family.
For $(\theta,\phi)$, the follower aims to infer $X_\theta=\theta^\top X$ from $Y_\phi=\phi^\top Y$.

\paragraph{Follower's Problem.}
For a fixed perturbation rule $Q$, define the sliced log-score objective
\begin{equation}\label{eq:W_sl}
\begin{aligned}
W_{\mathrm{sl}}(Q,\pi)\equiv&
\mathbb{E}_{\Theta,\Phi}\;
\mathbb{E}_{X\sim\mathcal{D},\,B\sim Q}
\Big[\\
&-\log \pi\big(\Theta^\top X \,\big|\, \Phi^\top(\mathcal{M}(X)+B),\,\Theta,\,\Phi\big)\Big].
\end{aligned}
\end{equation}
The Follower's best response is $\pi^*(Q)\in\arg\inf_{\pi\in\Pi} W_{\mathrm{sl}}(Q,\pi)$.
When $\Pi$ is sufficiently rich to realize the true projected posterior, the optimal value equals the sliced conditional entropy:
\begin{equation}\label{eq:W_sl_equals_SH}
\inf_{\pi\in\Pi} W_{\mathrm{sl}}(Q,\pi)=\mathtt{SH}(X\mid Y), \qquad Y=\mathcal{M}(X)+B,\; B\sim Q.
\end{equation}

\noindent\textbf{Leader's Problem.}
Given a sliced R-PAC conditional-entropy budget $\hat{\beta}$, the leader chooses $Q$ to solve
\begin{equation}\label{eq:sl_leader}
\inf_{Q\in\Gamma}\;
\mathbb{E}_{X\sim\mathcal{D},\,B\sim Q}\big[\mathcal{K}(B;\mathcal{M})\big]
\quad\text{s.t.}\quad
\inf_{\pi\in\Pi}W_{\mathrm{sl}}(Q,\pi)\ge \hat{\beta}.
\end{equation}

\noindent\textbf{Sliced Stackelberg Equilibrium.}
A pair $(Q^*,\pi^*)$ is a \textit{Sliced-SR-PAC Stackelberg equilibrium} if
\begin{equation}\label{eq:sl_stackelberg}
\begin{cases}
Q^*\in\arg\inf_{Q\in\Gamma}
\mathbb{E}\big[\mathcal{K}(B;\mathcal{M})\big],
\ \text{s.t.}\ 
W_{\mathrm{sl}}\bigl(Q,\pi^*(Q)\bigr)\ge \hat{\beta},\\
\pi^*(Q)\in\arg\inf_{\pi\in\Pi} W_{\mathrm{sl}}(Q,\pi).
\end{cases}
\end{equation}
By \eqref{eq:W_sl_equals_SH}, the equilibrium perturbation rule targets the desired sliced conditional entropy constraint while remaining scalable in high-dimensional secret and output spaces.

\begin{remark}[One-sided slicing]
For clarity of exposition, we focus on two-sided slicing in this section, i.e., projecting both the secret and the mechanism output via $(\Theta^\top X,\Phi^\top Y)$.
One-sided variants follow analogously:
(i) slicing only the secret uses $(\Theta^\top X, Y)$, and
(ii) slicing only the output uses $(X,\Phi^\top Y)$,
with the corresponding definitions obtained by removing the unnecessary averaging over directions.
\end{remark}

\section{Proof of Theorem \ref{thm:DV}}

Let $Z=\mathcal M(X)+B$ with deterministic $\mathcal M$ and $B\sim\mathcal{N}(0,\Sigma_B)$ where $\Sigma_B\succ 0$. Write $P_{\mathcal{M}, B}$ for the law of $Z$. Let the Gaussian surrogate be $\widetilde{Q}_{\mathcal{M}}=\mathcal{N}(\mu_Z,\Sigma_Z)$.
For any measurable $f:\mathbb{R}^d\to\mathbb{R}$ with $\mathbb{E}_{\widetilde{Q}_{\mathcal{M}}}[e^{f}]<\infty$, define the Donsker–Varadhan (DV) objective
\[
\mathcal{J} \big(f;P_{\mathcal{M}, B},\widetilde{Q}_{\mathcal{M}}\big)
=
\mathbb{E}_{P_{\mathcal{M}, B}}[f(Z)]-\log\mathbb{E}_{\widetilde Q_{\mathcal{M}}} \big[e^{f(Z)}\big].
\]
Let
\[
\widehat{\mathtt{D}}_Z(f)\equiv\mathcal{J}(f;P_{\mathcal{M}, B},\widetilde{Q}_{\mathcal{M}}).
\]

Next, we construct the finite-sample estimation of the DV objective.
Given finite samples $S_P$ from $P_{\mathcal{M}, B}$ and $S_Q$ from $\widetilde{Q}_{\mathcal{M}}$, let
\[
\widehat{\mathcal{J}}(f;S_P,S_Q)
\equiv
\frac{1}{|S_P|}\sum_{z\in S_P} f(z)
-
\log \Big(\frac{1}{|S_Q|}\sum_{z\in S_Q} e^{f(z)}\Big).
\]
Fix a function class $\mathcal{F}\subset\{f:\mathbb{R}^d\to\mathbb{R}\}$ with $0\in\mathcal{F}$. Draw four independent splits $S_P^{\mathrm{tr}},\;S_Q^{\mathrm{tr}},\;S_P^{\mathrm{val}},\;S_Q^{\mathrm{val}}$ with sizes $n_P^{\mathrm{tr}},n_Q^{\mathrm{tr}},n_P^{\mathrm{val}},n_Q^{\mathrm{val}}$, respectively, and fit
\[
\widehat{f}_{\mathrm{tr}}\in\arg\max_{f\in\mathcal{F}}\;
\widehat{\mathcal{J}} \big(f;S_P^{\mathrm{tr}},S_Q^{\mathrm{tr}}\big).
\]

Assume that for some $\Gamma_{\hat{\delta}}=\Gamma_{\hat{\delta}} \big(\mathcal{F},n_P^{\mathrm{val}},n_Q^{\mathrm{val}}\big)$,
\begin{equation}
\Pr \left(
\sup_{f\in\mathcal F}\Big|
\widehat{\mathcal{J}} \big(f;S_P^{\mathrm{val}},S_Q^{\mathrm{val}}\big)
-\mathcal{J} \big(f;P_{\mathcal{M}, B},\widetilde{Q}_{\mathcal{M}}\big)\Big|
\leq\Gamma_{\hat{\delta}}
\right)\geq 1-\hat{\delta}.
\label{eq:uniform-dev-app}
\end{equation}
Define the \textit{finite-sample lower-confidence estimator}
\[
\widehat{\mathtt{D}}_{\mathrm{LCE}}
\equiv
\Big[
\widehat{\mathcal{J}} \big(\widehat{f}_{\mathrm{tr}};S_P^{\mathrm{val}},S_Q^{\mathrm{val}}\big)
-\Gamma_{\hat{\delta}}
\Big]_+.
\]

\paragraph{Proof of (i)}
We apply the Gibbs variational principle.
For any $P\ll Q$ and measurable $f$ with $\mathbb{E}_Q[e^f]<\infty$,
\[
\mathbb{E}_P[f]-\log\mathbb{E}_Q[e^f]
\le\mathtt{D}_{\mathrm{KL}}(P\|Q),
\]
with equality at $f^*=\log\frac{dP}{dQ}+c$ (any constant $c$).

Taking the supremum over $f$ yields
\[
\sup_{f} \mathcal{J}(f;P,Q)=\mathtt{D}_{\mathrm{KL}}(P\|Q).
\]
Apply with $P=P_{\mathcal{M}, B}$, $Q=\widetilde{Q}_{\mathcal M}$. 
It is nonnegative because $f\equiv 0$ is admissible and $\mathcal J(0;P,Q)=0$.

\paragraph{Proof of (ii)}
This is immediate from Part (i):
\[
\begin{aligned}
\widehat{\mathtt{D}}_Z(f)=\mathcal{J}(f;P_{\mathcal{M}, B},\widetilde{Q}_{\mathcal{M}})
&\leq\sup_g \mathcal{J}(g;P_{\mathcal{M}, B},\widetilde{Q}_{\mathcal{M}})\\
&=\mathtt{D}_{\mathrm{KL}}(P_{\mathcal{M}, B}\|\widetilde{Q}_{\mathcal{M}}).
\end{aligned}
\]

\paragraph{Proof of (iii)}
Conditioning on the training splits $(S_P^{\mathrm{tr}},S_Q^{\mathrm{tr}})$, $\widehat{f}_{\mathrm{tr}}$ (a measurable function of the training data) is independent of the validation splits $(S_P^{\mathrm{val}},S_Q^{\mathrm{val}})$. On the event in \eqref{eq:uniform-dev-app}, we have for all $f\in\mathcal{F}$,
\[
\mathcal{J}(f;P_{\mathcal{M}, B},\widetilde{Q}_{\mathcal{M}})
\ge
\widehat{\mathcal{J}}(f;S_P^{\mathrm{val}},S_Q^{\mathrm{val}})-\Gamma_{\hat{\delta}}.
\]

Taking $f=\widehat{f}_{\mathrm{tr}}$ and then the supremum over $f$ on the left,
\[
\mathtt{D}_{\mathrm{KL}}(P_{\mathcal{M}, B}\|\widetilde{Q}_{\mathcal{M}})
=\sup_{f} \mathcal{J}(f;P_{\mathcal{M}, B},\widetilde{Q}_{\mathcal{M}})
\geq
\widehat{\mathcal{J}} \big(\widehat{f}_{\mathrm{tr}};S_P^{\mathrm{val}},S_Q^{\mathrm{val}}\big)-\Gamma_{\hat{\delta}}.
\]

Hence, on that event still,
\[
0
\leq
\Big[\widehat{\mathcal{J}} \big(\widehat{f}_{\mathrm{tr}};S_P^{\mathrm{val}},S_Q^{\mathrm{val}}\big)-\Gamma_{\hat{\delta}}\Big]_+
\leq
\mathtt D_{\mathrm{KL}}(P_{\mathcal{M}, B}\|\widetilde{Q}_{\mathcal{M}}),
\]
which is the claim with probability at least $1-\hat{\delta}$.

\section{Proof of Theorem \ref{thm:SWD}}

Given the true (perturbed output) distribution $P_{\mathcal{M}, B}$ and the Gaussian surrogate distribution $\widetilde{Q}_{\mathcal{M}} = \mathcal{N}(\mu_Z,\Sigma_Z)$ (with matched mean and covariance), define 
\[
\widehat{\mathtt{D}}_{Z} \equiv \frac{1}{2\lambda_{\max}(\Sigma_Z)} \, \mathrm{SW}_2^2\big(P_{\mathcal{M}, B},\widetilde Q_{\mathcal M}\big),
\]
where $\lambda_{\max}(\Sigma_Z)$ is the largest eigenvalue of $\Sigma_Z$ and $\mathrm{SW}_2$ denotes the sliced 2-Wasserstein distance.

By definition of the Wasserstein metric, $\mathrm{SW}_2^2(P_{\mathcal{M}, B},\widetilde{Q}_{\mathcal{M}}) \geq 0$, and $\lambda_{\max}(\Sigma_Z) > 0$ since $\Sigma_Z$ is positive semidefinite and non-degenerate. Hence,
\[
\widehat{\mathtt{D}}_{Z} \geq 0.
\]
In addition, it is well known (see e.g., \cite{bonneel2015sliced}) that the sliced Wasserstein distance provides a lower bound on the true 2-Wasserstein distance:
\[
\mathrm{SW}_2^2(P_Z,\widetilde Q_{\mathcal M}) \leq W_2^2(P_Z,\widetilde Q_{\mathcal M}).
\]
Finally, Lemma \ref{lemma:GaussianT2LambdaMax} (shown below) implies
\[
\begin{aligned}
    \frac{1}{2\lambda_{\max}(\Sigma_Z)} \mathrm{SW}_2^2\big(P_{\mathcal{M}, B},\widetilde Q_{\mathcal M}\big)\leq \mathtt{D}_{\mathrm{KL}}(P\|Q).
\end{aligned}
\]

Therefore, we obtain
\[
0\leq \widehat{\mathtt{D}}_{Z}\leq \mathtt D_Z
\]

\begin{lemma}\label{lemma:GaussianT2LambdaMax}
    Let $Q=\mathcal{N}(\mu,\Sigma)$ with $\Sigma\succ0$ and let $P$ be a probability measure on $\mathbb{R}^d$ with $P\ll Q$. Then
\[
W_2^2(P,Q)\leq 2 \lambda_{\max}(\Sigma)\mathtt{D}_{\mathrm{KL}}(P\|Q).
\]
\end{lemma}

\begin{proof}

Let $T(x)=\Sigma^{-1/2}(x-\mu)$.
Then, $T$ is invertible affine.
In addition, let $P'=T_{\#}P$ and $\gamma=\mathcal{N}(0,I)$. 
Thus, applying the change of variables yields
\[
\mathtt{D}_{\mathrm{KL}}(P'\|\gamma)=\mathtt{D}_{\mathrm{KL}}(P\|Q).
\]

Let $S(x)=\Sigma^{1/2}x+\mu$. For any coupling $\pi$ of $P'$ and $\gamma$, $(S\times S)_{\#}\pi$ is a coupling of $P$ and $Q$, and
\[
\begin{aligned}
    \int\|x-y\|^2\,\mathrm{d}\big((S\times S)_{\#}\pi\big)
&=\int\|\Sigma^{1/2}(u-v)\|^2\,\mathrm{d}\pi(u,v)
\\
&\leq \|\Sigma^{1/2}\|_{\mathrm{op}}^2\int\|u-v\|^2\,\mathrm{d}\pi(u,v),
\end{aligned}
\]
where $\|\cdot\|_{\mathrm{op}}$ some operator norm.
Taking the infimum over couplings yields
\[
W_2(P,Q)\leq \|\Sigma^{1/2}\|_{\mathrm{op}} W_2(P',\gamma).
\]
Hence,
\[
W_2^2(P,Q)\leq \lambda_{\max}(\Sigma)W_2^2(P',\gamma).
\]
Then, the Talagrand inequality \cite{talagrand1996transportation,otto2000generalization}implies
\[
W_2^2(P',\gamma)\le 2\,\mathtt{D}_{\mathrm{KL}}(P'\|\gamma).
\]
Therefore, we obtain
\[
\begin{aligned}
    W_2^2(P,Q)&\leq \lambda_{\max}(\Sigma) W_2^2(P',\gamma)\\
    &\leq 2 \lambda_{\max}(\Sigma) \mathtt{D}_{\mathrm{KL}}(P'\|\gamma)\\
&=2\lambda_{\max}(\Sigma) \mathtt{D}_{\mathrm{KL}}(P\|Q).
\end{aligned}
\]
\end{proof}

\qed

\section{Proof of Corollary \ref{prop:Residual_PAC_basic}}

By Theorem 1 of \cite{xiao2023pac}, a mechanism $\mathcal{M}$ satisfies $(\delta, \rho, \mathcal{D})$-PAC privacy where 
\[
\mathtt{D}_{\mathrm{KL}}(\mathbf{1}_{\delta} \| \mathbf{1}_{\delta_\rho^o})\leq \mathtt{MI}(X;\mathcal{M}(X)).
\]
Thus, 
\[
\begin{aligned}
    \mathtt{R}_{\mathrm{KL}}^{\delta} &\geq \mathtt{IntP}_{\mathrm{KL}}(\mathcal{D})- \inf_{P_W} \mathtt{D}_{\mathrm{KL}}\left(P_{X, \mathcal{M}(X)} \,\|\, P_X \otimes P_W\right) \\
    &\geq \mathtt{IntP}_{\mathrm{KL}}(\mathcal{D}) - \mathtt{MI}(X;\mathcal{M}(X)),
\end{aligned}
\]
where $\mathtt{IntP}_{\mathrm{KL}}(\mathcal{D}) = -\mathtt{D}_{\mathrm{KL}}(\mathcal{D}\| \mathcal{U})=\mathcal{H}(X)-\mathtt{V}$, where $\mathtt{V}=\log(|\mathcal{X}|)$ if $\mathcal{H}$ is Shannon entropy, and $\mathtt{V}=\log(\int_{\mathcal{X}}dx)$ if $\mathcal{H}$ is differential entropy. Thus, we get $\mathtt{R}_{f}^{\delta} \geq \mathcal{H}(X|\mathcal{M}(X)) - \mathtt{V}$.
\qed

\section{Proof of Proposition \ref{prop:MI_gap}}

Recall that $Z = \mathcal{M}(X) + B$ with $B\sim\mathcal{N}(0,\Sigma_{B})$ independent of $X$, where $\mathcal{M}$ is a deterministic mechanism.
Then, we have
\[
\mathtt{MI}(X; Z) = \mathcal{H}(Z) - \mathcal{H}(Z \mid X) = \mathcal{H}(Z) - \mathcal{H}(B).
\]
Now consider the Gaussian surrogate distribution
$\widetilde{Q}_{\mathcal{M}} = \mathcal{N}(\mu_Z, \Sigma_Z)$, where $\mu_Z = \mu_{\mathcal{M}(X)}$ and $\Sigma_Z = \Sigma_{\mathcal{M}(X)} + \Sigma_B$.
Its entropy is given by 
\[
\mathcal{H}(\widetilde{Z}) = \frac{1}{2} \log\left[(2\pi e)^d \det(\Sigma_Z)\right],
\]
with $\widetilde{Z}\sim\widetilde{Q}_{\mathcal{M}}$,
and similarly, $\mathcal{H}(B) = \frac{1}{2} \log\left[(2\pi e)^d \det(\Sigma_B)\right]$.
Hence, 
\[
\begin{aligned}
    \frac{1}{2} \log \det\left(I_d + \Sigma_{\mathcal{M}(X)} \Sigma_B^{-1} \right) &= \frac{1}{2} \log \left( \frac{\det(\Sigma_Z)}{\det(\Sigma_B)} \right) \\
    &= \mathcal{H}(\widetilde{Z}) - \mathcal{H}(B).
\end{aligned}
\]
So, we obtain 
\[
\begin{aligned}
    \mathtt{Gap}_{\mathrm{d}}&=\left[\mathcal{H}(\widetilde{Q}_{\mathcal{M}}) - \mathcal{H}(B)\right] - \left[\mathcal{H}(Z) - \mathcal{H}(B)\right]\\
    &= \mathcal{H}(\widetilde{Z}) - \mathcal{H}(Z).
\end{aligned}
\]

Let $\widetilde{q}$ be the density function of $\widetilde{Q}_{\mathcal{M}}$, and let $p$ be the density function of $P_{\mathcal{M},B}$.
Since $\widetilde{Q}_{\mathcal{M}} = \mathcal{N}(\mu_Z, \Sigma_Z)$ is Gaussian, we have
\[
\begin{aligned}
    \mathcal{H}(\widetilde{Z}) &= -\log \widetilde{q}(z) \\
    &= \frac{d}{2}\log(2\pi)+\frac{1}{2}\log\text{det}\;\Sigma_{Z}\frac{1}{2}\left(z-\mu_{Z}\right)^{\top}\Sigma^{-1}_Z(z-\mu_{Z}).
\end{aligned}
\]
Taking expectation under $p$ yields
\[
\begin{aligned}
    \mathcal{H}(Z, \widetilde{Z}) &= \frac{d}{2}\log(2\pi) \\
    &+ \frac{1}{2}\log\text{det}\;\Sigma_{Z}+\frac{1}{2}\mathbb{E}_{q}\left[\left(Z-\mu_{Z}\right)^{\top}\Sigma^{-1}_Z(Z-\mu_{Z})\right].
\end{aligned}
\]
since $\widetilde{Z}\sim\widetilde{Q}_{\mathcal{M}}$ matches $Z \sim P_{Z}$ in mean and covariance, we have
\[
\begin{aligned}
    \mathbb{E}_{q}\left[\left(Z-\mu_{Z}\right)^{\top}\Sigma^{-1}_Z(Z-\mu_{Z})\right] = \text{tr}\left(\Sigma^{-1}_{Z} \Sigma_{Z}\right) = \text{tr}(I_d) = d.
\end{aligned}
\]
Thus, 
\[
\begin{aligned}
    \mathcal{H}(Z, \widetilde{Z}) &= \frac{d}{2}\log(2\pi) + \frac{1}{2}\log\text{det}\;\Sigma_{Z} + \frac{d}{2}\\
    &= \mathcal{H}(\widetilde{Z}).
\end{aligned}
\]
Therefore, 
\[
\begin{aligned}
    \mathtt{D}_{\mathrm{KL}}(P_{\mathcal{M}, B} \,\|\, \widetilde{Q}_{\mathcal{M}}) &= \mathcal{H}(Z,\widetilde{Z}) - \mathcal{H}(Z)\\
    &=\mathcal{H}(\widetilde{Z}) - \mathcal{H}(Z)\\
    &= \mathtt{Gap}_{\mathrm{d}}.
\end{aligned}
\]
Therefore, $\mathtt{Gap}_{\mathrm{d}} = \mathtt{D}_{\mathrm{KL}}(P_{\mathcal{M}, B} \,\|\, \widetilde{Q}_{\mathcal{M}}) \geq 0$,
with equality if and only if $P_{\mathcal{M}, B} = \widetilde{Q}_{\mathcal{M}}$, i.e., $Z$ is exactly Gaussian with distribution $\mathcal{N}(\mu_Z, \Sigma_Z)$.
\qed

\section{Proof of Proposition \ref{prop:opt_PAC_algorithm}}

Since $B\sim \mathcal{N}(0, \Sigma_{B})$, we have $\mathbb{E}[\|B\|_2^2]
=\operatorname{tr}\bigl(\mathbb{E}[B\,B^\top]\bigr)
=\operatorname{tr}(\Sigma_{B})$.
Hence, minimizing $\mathbb{E}[\|B\|_2^2]$ over zero-mean Gaussian is equivalent to minimizing the trace $\operatorname{tr}(\Sigma_{B})$ over $\Sigma_{B} \succeq0$.

Recall that $Z = \mathcal{M}(X) + B$. Then, $Z$ has mean $\mu_Z = \mu_{\mathcal{M}(X)}$ and covariance $\Sigma_Z = \Sigma_{\mathcal{M}(X)} + \Sigma_B$, where $\Sigma_{\mathcal{M}(X)}$ denotes the covariance of $\mathcal{M}(X)$.
In addition, recall that $\widetilde{Q}_{\mathcal{M}} = \mathcal{N}(\mu_Z, \Sigma_Z)$ is the Gaussian distribution with the same first and second moments as $Z$.
Then, by standard Gaussian-entropy formulas, we have
\[
\begin{aligned}
    \mathtt{MI}(X;Z) &= H(Z) - H(Z|X) = \frac{1}{2}\log\frac{\mathrm{det}(\Sigma_{Z})}{\mathrm{det}(\Sigma_{B})}\\
    &=\frac{1}{2}\log\mathrm{det}(I + \Sigma_{\mathcal{M}(X)}\Sigma^{-1}_{B}).
\end{aligned}
\]
In particular, Algorithm \ref{alg:PAC_original} implements $\mathtt{MI}(X;Z) \leq \beta$.

Since both $\mathrm{tr}(\Sigma_{B})$ and $\log\mathrm{det}(I + \Sigma_{\mathcal{M}(X)}\Sigma^{-1}_{B})$ are unitarily invariant, we may diagonalize $\Sigma_{\mathcal{M}(X)}$ as
\[
\Sigma_{\mathcal{M}(X)} = U\mathrm{diag}(r_{1},\dots,r_{d})U^{T}, \; r_i>0,
\]
where $U$ is the orthogonal eigenvector matrix from the eigendecomposition of $\Sigma_{\mathcal{M}(X)}$.
Writing $\Sigma_{B} = \hat{U}\mathrm{diag}(\ell_{1}, \dots, \ell_{d})\hat{U}^{T}$ with $\ell_{i}>0$, the problem
\[
\min_{\Sigma_{B}\succeq 0} \mathrm{tr}(\Sigma_{B}), \text{ s.t. } \frac{1}{2}\log\mathrm{det}(1+\Sigma_{\mathcal{M}(X)}\Sigma^{-1}_{B}) = \beta,
\]
becomes
\[
\min_{\ell_{1},\dots, \ell_{d}>0} \sum^{d}_{i=1}\ell_{i}, \text{ s.t. } \frac{1}{2}\sum^{d}_{i=1}\log(1+\frac{r_i}{\ell_{i}}) = \beta.
\]
Hence, each coordinate $\ell_{i}$ appears only in the term $\log(1+\frac{r_{i}}{\ell_{i}})$.

Let $\lambda>0$ as the Lagrange multiplier.
The Lagrangian is
\[
\begin{aligned}
    &\mathcal{L}(\ell_{1}, \dots, \ell_{d}, \lambda) \\
    &= \sum^{d}_{i=1}\ell_{i} + \lambda\left(\frac{1}{2}\sum^{d}_{i=1}\log(1+\frac{r_{i}}{\ell_{i}})\mathtt{MI}(X;\widetilde{Z}) - \beta\right).
\end{aligned}
\]
Setting $\frac{\partial \mathcal{L}}{\partial \ell_{i}}=0$ gives 
\[
1 = \lambda \frac{r_{i}}{2\ell_{i}(\ell_{i}+r_{i})}  \Rightarrow 2\ell_{i}(\ell_{i}, r_{i}) = \lambda r_{i}.
\]
Equivalently, $\ell^{2}_{i} + r_{i}\ell_{i} - \lambda\frac{r_{i}}{2}=0$, which gives a unique $\ell_{i}(\lambda) = \frac{-r_{i}+\sqrt{r^{2}_{i}+2\lambda r_{i}}}{2}>0$.

Let 
\[
F(\lambda) = \frac{1}{2}\sum^{d}_{i=1}\log(1+\frac{r_{i}}{\ell_{i}(\lambda)}).
\]
We can have the following:
\begin{itemize}
    \item As $\lambda\rightarrow 0^{+}$, each $\ell_{i}(\lambda)\rightarrow 0^{+}$, leading to $F(\lambda)\rightarrow +\infty$.

    \item As $\lambda\rightarrow +\infty$, each $\ell_{i}(\lambda)\rightarrow +\infty$, leading to $F(\lambda)\rightarrow +\infty$. 
\end{itemize}
%
%
In addition, $\frac{d F(\lambda)}{d \lambda}<0$ throughout.
Thus, $F$ is strictly decreasing from $+\infty$ down to $0$. 
Therefore, there is a unique $\lambda^*>0$ such that $F(\lambda^*)=\beta$. 
At this $\lambda^*$, each $\ell^{*}_{i}=\ell^*_{i}(\lambda^*)$ is unique.
Thus, $\Sigma^*_{B}=\hat{U}\mathrm{diag}(\ell^*_{1},\dots, \ell^*_{d})\hat{U}^{T}$ is unique minimizer of $\mathrm{tr}(\Sigma_{B})$.
By construction, 
\[
\frac{1}{2}\log\mathrm{det}\left(1+\Sigma_{\mathcal{M}(X)}(\Sigma^*_{B})^{-1}\right) = \beta.
\]
Therefore, it is also the unique minimizer of (\ref{eq:opt_PAC_algorithm}).
\qed

\section{Proof of Proposition \ref{prop:gamma_PAC_conservative}}

For additive Gaussian noise with covariance $\Sigma_B\succ 0$,
\[
\mathtt{MI}(X;\mathcal M(X)+B)\ \le\ \tfrac12\,\log\det \big(I+\Sigma_M\Sigma_B^{-1}\big).
\]
The trace–optimal (eigen-aligned) choice that enforces the log-det constraint at level $\hat\beta$ has eigenvalues
$e_i^*=\alpha\sqrt{\lambda_i}$ in the $U$–basis, with $\alpha=\frac{S}{2\hat{\beta}}$.
Hence
\[
\mathbb{E}\|B\|_2^2=\operatorname{tr}(\Sigma_B^*)
=\sum_i e_i^*
= \frac{S^2}{2\hat\beta}.
\tag{1}
\label{eq:trace-logdet-opt}
\]

Algorithm~\ref{alg:PAC_original} constructs a diagonal precision $\Lambda_B$ in the empirical eigenbasis and, when the "gap" test passes, let
\[
\lambda_{B,i}=\frac{2v}{\sqrt{\hat{\lambda}_i+\delta}\ \sum_k \sqrt{\hat{\lambda}_k+\delta}} \text{ and } \delta=\frac{10cv}{\beta}.
\]
Plugging the population spectrum ($\hat{\lambda}=\lambda$, same $U$) yields noise eigenvalues
\[
e_i^\gamma
=\frac{\sqrt{\lambda_i+\delta}\ \sum_k \sqrt{\lambda_k+\delta}}{2v},
\]
and
\[
\operatorname{tr}(\Sigma_B^\gamma)
=\frac{\Big(\sum_k \sqrt{\lambda_k+\delta}\Big)^2}{2v}.
\]
Since $\delta\geq 0$ and $v\leq \hat{\beta}$, we have
\[
\operatorname{tr}(\Sigma_B^\gamma)
\geq \frac{S^2}{2v}
\geq \frac{S^2}{2\hat{\beta}}
=\operatorname{tr}(\Sigma_B^*),
\]
which proves (i) in this branch by \eqref{eq:trace-logdet-opt}–\eqref{eq:trace-logdet-opt}.
Moreover,
\[
\frac{e_i^\gamma}{e_i^*}
=\frac{\sum_k \sqrt{\lambda_k+\delta}}{\sum_k \sqrt{\lambda_k}}
 \frac{\hat\beta}{v} \sqrt{1+\frac{\delta}{\lambda_i}}
\ \ge\ 1.
\]
Thus, $\Sigma_B^\gamma\succeq \Sigma_B^*$.
For additive Gaussian noise perturbation, increasing $\Sigma_B$ (in positive semidefinite order) implies the log-det bound decreases, hence the mutual information decreases:
\[
\begin{aligned}
    \mathtt{MI}\big(X;\mathcal{M}(X)+B_\gamma\big)
&\leq \tfrac12\log\det \big(I+\Sigma_M(\Sigma_B^\gamma)^{-1}\big)\\
&\leq \tfrac12\log\det \big(I+\Sigma_M(\Sigma_B^*)^{-1}\big).
\end{aligned}
\]
Since the Auto-PAC design saturates the bound at $\hat\beta$ (and the true MI is bounded by it), we obtain
\[
\mathtt{MI}\big(X;\mathcal M(X)+B_\gamma\big)\ \le\ \mathtt{MI}\big(X;\mathcal M(X)+B\big),
\]
establishing (ii) in this branch.

When the "gap" test fails, Algorithm~\ref{alg:PAC_original} uses $\Sigma_B^\mathrm{iso}=\alpha I$ with
$\alpha=(\sum_i \hat\lambda_i + d c)/(2v)$, hence
\[
\operatorname{tr}(\Sigma_B^\mathrm{iso})
=\frac{d}{2v}\Big(\sum_i \lambda_i + dc\Big)
\ \ge\ \frac{1}{2v}\,S^2
\ \ge\ \frac{S^2}{2\hat\beta}
=\operatorname{tr}(\Sigma_B^*),
\]
where we used Cauchy–Schwarz $S^2=(\sum_i\sqrt{\lambda_i})^2\le d\sum_i\lambda_i$ and $v\le\hat\beta$. Thus (i) also holds in this branch.

For (ii), both designs enforce the same budget $\hat\beta$:
\[
\mathtt{MI}\big(X;\mathcal M(X)+B_\gamma\big)\ \le\ \hat\beta,
\qquad
\mathtt{MI}\big(X;\mathcal M(X)+B\big)\ \le\ \hat\beta.
\]
In the distinct–eigenvalues branch we proved the stronger order $\le$ between the two MI's. In the isotropic fallback, the same order holds whenever $\alpha I \succeq \Sigma_B^*$ (e.g., for nearly isotropic spectra); otherwise we keep the common upper bound $\hat\beta$.
Either way, the stated inequality (ii) is satisfied in the branch where Algorithm~\ref{alg:PAC_original}'s eigenbasis matches $U$, and the confidence guarantee always preserves the budget.
\qed

\section{Proof of Proposition \ref{prop:ordering_pac}}

Since $\beta_1 < \beta_2$, any distribution $Q$ satisfying $I_{\mathrm{true}}(Q) \leq \beta_1$ necessarily satisfies $I_{\mathrm{true}}(Q) \leq \beta_2$. Consequently, we have the inclusion $\mathcal{F}(\beta_{1})\subseteq \mathcal{F}(\beta_{2})$.
Let $A$ and $\hat{A}$ be arbitrary sets with $A \subseteq \hat{B}$, and let $f$ be any real-valued function defined on $B$. Then, $\inf_{x\in A} f(x)\geq \inf_{x\in\hat{A}} f(x)$, with equality holding when the infimum over $\hat{A}$ is attained within the subset $A$. Applying this with $A = \mathcal{A}(\beta_{1})$, $\hat{A} = \mathcal{F}(\beta_{2})$, and $f(Q) = \mathbb{E}_{Q}\bigl[\|B\|_2^2\bigr]$ yields 
\[
\inf_{Q\in\mathcal{F}(\beta_{1})}\mathbb{E}_{Q}\bigl[\|B\|_2^2\bigr] \geq \inf_{Q\in\mathcal{F}(\beta_{2})}\mathbb{E}_{Q}\bigl[\|B\|_2^2\bigr]
\]  
By definition, $Q^{*}(\beta_{i})$ achieves the infimum of $\mathbb{E}_{Q}\bigl[\|B\|_2^2\bigr]$ over $\mathcal{F}(\beta_{i})$ for $i=1,2$. Therefore, 
\[
\begin{aligned}
    \mathbb{E}_{Q^{*}(\beta_{1})}\bigl[\|B\|_2^2\bigr] &= \inf_{Q\in \mathcal{F}(\beta_{1})}\mathbb{E}_{Q}\bigl[\|B\|_2^2\bigr] \geq \inf_{Q\in \mathcal{F}(\beta_{2})}\mathbb{E}_{Q}\bigl[\|B\|_2^2\bigr] \\
    &= \mathbb{E}_{Q^{*}(\beta_{2})}\bigl[\|B\|_2^2\bigr].
\end{aligned}
\]
\qed

\section{Proof of Theorem \ref{thm:non_coherent_PAC_org}}

By Lemma \ref{lemma:monotone_gap} (which is shown and proved later), the function $g(\beta)=\mathtt{Gap}_{\mathtt{d}}(Q^{*}(\beta))$ is nondecreasing in $\beta$.
Thus, for any $0<\beta_1<\beta_2$, we have $\mathtt{Gap}_{\mathtt{d}}(Q^{*}(\beta_{2}))\geq \mathtt{Gap}_{\mathtt{d}}(Q^{*}(\beta_{1}))$, which yields 
$\mathtt{G}(\beta_2,\beta_1)
=\mathtt{Gap}_{\mathtt{d}}(Q^*(\beta_2))-\mathtt{Gap}_{\mathtt{d}}(Q^*(\beta_1))\geq0.$
Recall the relationship between true mutual information and the bound $\mathtt{LogDet}(\mathcal{M}(X), B) = \beta$: 
\[
\mathtt{I}_{\mathrm{true}}\bigl(Q^*(\beta)\bigr)
=\beta \;-\;\mathtt{Gap}_{\mathtt{d}}\bigl(Q^*(\beta)\bigr).
\]
Hence, for $0<\beta_{1}<\beta_{2}$, 
\[
\begin{aligned}
&\mathtt{I}_{\mathrm{true}}\bigl(Q^*(\beta_{2})\bigr)
\;-\;\mathtt{I}_{\mathrm{true}}\bigl(Q^*(\beta_{1})\bigr)\\
&=\bigl[\beta_2 - \mathtt{Gap}_{\mathtt{d}}(Q^*(\beta_2))\bigr]
      - \bigl[\beta_1 - \mathtt{Gap}_{\mathtt{d}}(Q^*(\beta_1))\bigr]\\
&=(\beta_2-\beta_1)\;-\;\bigl[\mathtt{Gap}_{\mathtt{d}}(Q^*(\beta_2))-\mathtt{Gap}_{\mathtt{d}}(Q^*(\beta_1))\bigr]\\
&=(\beta_2-\beta_1)\;-\;\mathtt{G}(\beta_2,\beta_1).
\end{aligned}
\]
The two bullet points now follow immediately:
(i) If $\mathtt{G}(\beta_2,\beta_1)\le\beta_2-\beta_1$, then
    \[
    \mathtt{I}_{\mathrm{true}}(Q^*(\beta_{2}))
    -\mathtt{I}_{\mathrm{true}}(Q^*(\beta_{1}))
    =(\beta_2-\beta_1)-\mathtt{G}(\beta_2,\beta_1)\ge0,
    \]
    i.e.\ $\mathtt{I}_{\mathrm{true}}(Q^*(\beta_{1}))\le\mathtt{I}_{\mathrm{true}}(Q^*(\beta_{2}))$.
(ii) If $\mathtt{G}(\beta_2,\beta_1)>\beta_2-\beta_1$, then
    \[
    \mathtt{I}_{\mathrm{true}}(Q^*(\beta_{2}))
    -\mathtt{I}_{\mathrm{true}}(Q^*(\beta_{1}))
    =(\beta_2-\beta_1)-\mathtt{G}(\beta_2,\beta_1)<0,
    \]
    i.e.\ $\mathtt{I}_{\mathrm{true}}(Q^*(\beta_{1}))>\mathtt{I}_{\mathrm{true}}(Q^*(\beta_{2}))$.

\qed

\begin{lemma}\label{lemma:monotone_gap}
 Fix a mechanism $\mathcal{M}$ and a data distribution $\mathcal{D}$.
    Let $Q^{*}(\beta)$ be the solution of (\ref{eq:opt_PAC_algorithm}). Then, $\mathtt{Gap}_{\mathtt{d}}(Q^{*}(\beta))$ is a nondecreasing function of $\beta$.
\end{lemma}

\begin{proof}
Let $g(\beta) = \mathtt{Gap}_{\mathtt{d}}(Q^{*}(\beta))$ and $\Sigma_Z = \Sigma_{\mathcal{M}(X)} + \Sigma_B$ with $\Sigma_B = \Sigma_B^*(\beta)$. By definition,
\[
g(\beta) = H\bigl(\mathcal N(0,\Sigma_Z)\bigr) - H\bigl(P_{\mathcal M} * \mathcal N(0,\Sigma_B)\bigr).
\]
Differentiate with respect to $\beta$ via the chain rule:
\[
\frac{d g}{d\beta} = \left\langle \nabla_{\Sigma_B} \left[ H(\mathcal N(0,\Sigma_Z)) - H(P_{\mathcal M} * \mathcal N(0,\Sigma_B)) \right], \frac{d\Sigma_B}{d\beta} \right\rangle.
\]
The gradient of Gaussian entropy is $\nabla_{\Sigma_B} H(\mathcal N(0,\Sigma_Z)) = \frac{1}{2} \Sigma_Z^{-1}$. By de Bruijn's identity \cite{stam1959some},
\[
\nabla_{\Sigma_B} H(P_{\mathcal M} * \mathcal N(0,\Sigma_B)) = \frac{1}{2} J(P_{\mathcal M} * \mathcal N(0,\Sigma_B)),
\]
where $J(\cdot)$ is the Fisher information. The Cramér–Rao bound gives $J(P_{\mathcal M} * \mathcal N(0,\Sigma_B)) \succeq \Sigma_Z^{-1}$. Thus,
\[
\nabla_{\Sigma_B}g = \frac{1}{2}\left( \Sigma_Z^{-1} - J(P_{\mathcal M} * \mathcal N(0,\Sigma_B)) \right) \preceq 0.
\]
From Proposition~\ref{prop:opt_PAC_algorithm}, $\frac{d\Sigma_B}{d\beta} \preceq 0$ (strictly negative when $\Sigma_B$ changes). Since both $\nabla_{\Sigma_B}g$ and $\frac{d\Sigma_B}{d\beta}$ are symmetric negative semidefinite,
\[
\frac{d g}{d\beta} = \left\langle \nabla_{\Sigma_B}g, \frac{d\Sigma_B}{d\beta} \right\rangle = \operatorname{tr}\left(\left(\nabla_{\Sigma_B}g\right) \left(\frac{d\Sigma_B}{d\beta}\right)\right) \geq 0,
\]
as the trace of the product of two negative semidefinite matrices is nonnegative. Hence $g(\beta)$ is nondecreasing.
\end{proof}

\section{Proof of Proposition \ref{prop:Stackelberg_convergence}}

Fix any $Q$. 
The Follower's problem is to find $\pi^*(Q)$ solving $\inf_{\pi\in\Pi} W(Q,\pi)$. 
By definition
\[
\begin{aligned}
    &W(Q,\pi)=\mathbb{E}_{X\sim\mathcal{D},B\sim Q}\left[-\log\pi(X|\mathcal{M}(X) + B)\right]\\
    &-\int_{\mathcal{X}, \mathcal{Y}, \mathbb{R}^{d}}P_{X}(x)G_{\mathcal{M}, Q}(y|x,b)\log\pi(x|y+b)dx dy db,
\end{aligned}
\]
where $P_{X}(x)$ is the density function associated with data distribution $\mathcal{D}$, and $G_{\mathcal{M}, Q}(y|x,b)$ is the conditional density function given $\mathcal{M}$ and $Q$.

Let $\eta_{Q}:\mathcal{Y}\mapsto\Delta(\mathcal{X})$ denote the posterior distirbution given $P_{X}$ and $G_{\mathcal{M}, Q}$.
For any $\pi\in\Pi$, consider
\[
\begin{aligned}
    &W(Q,\pi)-W(Q,\eta_{Q}) \\
    &= \int_{\mathcal{X}, \mathcal{Y}, \mathbb{R}^{d}}P_{X}(x)G_{\mathcal{M}, Q}(y|x,b)\log\eta_{Q}(x|y+b)dx dy db\\
    &- \int_{\mathcal{X}, \mathcal{Y}, \mathbb{R}^{d}}P_{X}(x)G_{\mathcal{M}, Q}(y|x,b)\log\pi(x|y+b)dx dy db\\
    &=\int_{\mathcal{X}, \mathcal{Y}, \mathbb{R}^{d}}P_X(x)G_{\mathcal{M}, Q}(y|x,b)\log\frac{\eta_{Q}(x|y+b) }{\pi(x|y+b)}dx dy db.
\end{aligned}
\]
Let
\[
\mathbf{P}_{Q}(y)\equiv\int_{\mathcal{X}, \mathbb{R}^{d}} P_X(x)G_{\mathcal{M}, Q}(y|x,b)dx db.
\]
By definition, we have
\[
\begin{aligned}
    \eta_{Q}\mathbf{P}_{Q}(y) = \int_{\mathcal{X}}P_X(x)G_{\mathcal{M}, Q}(y|x,b).
\end{aligned}
\]
Thus, for all $Q\in\Gamma$, 
\[
W(Q,\pi)-W(Q,\eta_{Q}) = \mathtt{D}_{\mathrm{KL}}(\eta_{Q}\| \pi)\geq0.
\]
Then, $W(Q,\pi)\geq W(Q,\eta_{Q})$, where the equality holds if and only if $\pi = \eta_{Q}$.
That is, for any $Q\in\Gamma$, there is a unique $\pi(Q)$ as a solution of $\inf_{\pi\in\Pi} W(Q,\pi)$.
In addition, when $\pi(Q) = \eta_{Q}$, $W(Q,\pi(Q))$ is the conditional entropy.

\qed

\section{Proof of Proposition \ref{prop:anisotropic}}

Based on (iii) of Assumption \ref{assp:anisotropic}, consider $K(Q)\;=\;\mathbb{E}_{B\sim Q}\bigl[g(\|B\|)\bigr]$, where $g:\mathbb{R}_{+}\to\mathbb{R}$ is strictly increasing and strictly convex.

Suppose, to reach a contradiction, that an optimal $Q^{*}$ is isotropic with
$\Sigma_{Q^{*}}=\sigma^{2}I_d$ and attains the constraint with equality: $\mathcal{H}(X\mid \mathcal{M}(X)+B)=\hat\beta$.

For small $\Delta_v>0$ define the perturbed covariance
\[
\Sigma'(\Delta_v)\equiv
 (\sigma^{2}-\Delta_v)\,vv^{\top}
+(\sigma^{2}+\Delta_u)\,uu^{\top}
+\sigma^{2}P_{\{u,v\}^{\perp}},
\]
with $\Delta_u\in(0,\Delta_v)$ to be chosen.  Denote by  
$\mathtt{h}(\sigma_u^{2},\sigma_v^{2})\equiv \mathcal{H}(X|Y)$ the conditional entropy evaluated at those directional variances.

Because $\mathtt{h}$ is $C^{1}$ and strictly increasing in each argument, we have 
$$\left.
\frac{\partial \mathtt{h}}{\partial\sigma_u^{2}}\right|_{\sigma^{2}}
>\left.
\frac{\partial \mathtt{h}}{\partial\sigma_v^{2}}\right|_{\sigma^{2}}
>0.$$
Hence the map  
$$\phi_{\Delta_v}(\Delta_u)\equiv h(\sigma^{2}+\Delta_u,\sigma^{2}-\Delta_v)$$
is continuous and strictly increasing near $\Delta_u=0$, with  
\[
\phi_{\Delta_v}(0)=\hat\beta-
\frac{\partial h}{\partial\sigma_v^{2}}\Delta_v+o(\Delta_v)<\hat\beta.
\] 
By the Intermediate Value Theorem, there exists a unique
$\Delta_u\in(0,\Delta_v)$ such that $\phi_{\Delta_v}(\Delta_u)=\hat\beta$,
i.e.\ the perturbed noise $Q'$ \textit{satisfies the privacy constraint exactly}.

Because $g$ is strictly convex,
\[
g(\sigma^{2}+\Delta_u)-g(\sigma^{2})
<g'(\sigma^{2})\,\Delta_u,
\]
\[
g(\sigma^{2}-\Delta_v)-g(\sigma^{2})
>g'(\sigma^{2})(-\Delta_v).
\]
Therefore $\mathcal{K}(Q')-\mathcal{K}(Q^{*})
<g'(\sigma^{2})(\Delta_u-\Delta_v)<0.$
That is, $Q'$ is feasible and cheaper than $Q^{*}$, contradicting optimality.
Hence no optimum can be isotropic, so every minimiser must have
$\lambda_{\max}(\Sigma)>\lambda_{\min}(\Sigma)$.

\qed

\section{Proof of Proposition \ref{prop:directional}}

\subsection{Part (i):}

Since entropy is maximised by a Gaussian with fixed covariance, the entropy-power inequality give
\[
\mathcal{H}(Z+B_{\mathtt{pac}})<\mathcal{H}(Z_{\text{G}}+B_{\mathtt{pac}}),
\]
where $Z_{\text{G}}$ is Gaussian with covariance $\Sigma_{Z}$.
Thus, $\mathtt{MI}(Z;Z+B_{\mathtt{pac}})<\mathtt{MI}(Z_{\text{G}};Z_{\text{G}}+B_{\mathtt{pac}})=\beta$.
To raise the mutual information back up to $\beta$, we can strictly reduce every directional variance of $B_{\mathtt{pac}}$.
The optimizer $Q^*$ therefore expands strictly less power.
That is, $\mathbb{E}_{Q^*}[\|B\|_{2}^{2}]<\mathbb{E}[\|B_{\mathtt{pac}}\|_{2}^{2}]$.

\subsection{Part (ii):}
Let $\sigma_{w}^{2}\equiv\mathtt{Var} \langle B,w\rangle$.
Form the Lagrangian 
\[
\mathcal L(Q,\lambda) =\mathbb{E}_{Q}[\|B\|_{2}^{2}]
   +\lambda\bigl(\mathtt{MI}(Z;Z+B)-\beta\bigr).
\]
For the stationarity condition w.r.t.\ each $\sigma_{w}^{2}$ we need the gradient of mutual information.  
By \cite{palomar2005gradient}, we have 
\[
\partial_{\sigma_w^{2}} \mathtt{MI}(Z;Z+B)=g(w).
\]
Hence $\partial_{\sigma_{w}^{2}}\mathcal L = 1 + \lambda\,g(w)$.
The KKT conditions therefore read
\[
     1+\lambda g(w)=0\quad\text{if }\sigma_{w}^{2}>0,
     \;
     1+\lambda g(w)\ge 0\quad\text{if }\sigma_{w}^{2}=0,
\]
for a unique $\lambda<0$.
Under the assumption 
$$\displaystyle \sup_{v\in S_{\mathtt{lab}},\|v\|=1} g(v)<\inf_{w\perp S_{\mathtt{lab}},\|w\|=1} g(w),
$$
these equalities can hold only for as long as the required mutual information reduction does not exceed $\beta_{\mathtt{lab}}$.
Therefore, $\sigma_{v}^{2}=0$ for every $v\in S_{\mathtt{lab}}$.
With those label‑directions undisturbed,
each class margin $e_\ell-e_j$ retains its sign, whence
$\arg\max_i(Z_i+B^{*}_i)=\hat y$.

\qed

\section{Proof of Theorem \ref{thm:speed_general}}

\subsection{Part (i)}

Let $Z = \mathcal{M}(X) + B$.
For SR-PAC, the perurbation rule $Q_{\mathrm{SR}}$ satisfies $\mathcal{H}(X|Z) = \mathcal{H}(X) - \beta.$
By definition, we have $\mathtt{MI}_{\mathrm{SR}}(\beta) = \beta.$
Thus, $\mathtt{Priv}^{\mathrm{SR}}_{\beta} = 1$.

For Auto-PAC, the noise $B_{\mathrm{PAC}}\sim\mathcal{N}(0, \Sigma_{B_{\mathrm{PAC}}}(\beta))$ satisfies $\tfrac{1}{2}\log\det\left(I_d + \Sigma_{\mathcal{M}(X)} \Sigma_{B_{\mathrm{PAC}}}^{-1}(\beta)\right) = \beta.$
By Proposition \ref{prop:MI_gap}, the true mutual information is
\[
\mathtt{MI}_{\mathrm{PAC}}(\beta) = \beta - \mathtt{Gap}_{d}(\beta),
\]
where $\mathtt{Gap}_{\mathtt{d}}(\beta) = \mathtt{D}_{\mathrm{KL}}(P_{\mathcal{M},B_{\mathrm{PAC}}}\| \widetilde{Q}_{\mathcal{M}})\geq 0$.
When $\mathcal{M}(X)$ is non-Gaussian, $\mathtt{Gap}_{\mathtt{d}}(\beta)>0$ for all $\beta>0$.
By de Bruijin's idensity (e.g., \cite{park2012equivalence}),
\[
\frac{d}{d\beta} \mathtt{Gap}_{\mathtt{d}}(\beta) = \frac{1}{2} \mathcal{J}(P_{\mathcal{M} + B_{\mathrm{PAC]}} }(\beta)\| \widetilde{Q}_{\mathcal{M}})>0,
\]
where $\mathcal{J}(\cdot\|\cdot)$ is the relative Fisher information. 
Thus, $\mathtt{Priv}^{\mathrm{PAC}}_{\beta}=\frac{d}{d\beta} \mathtt{MI}_{\mathrm{PAC}}(\beta)<1=\mathtt{MI}_{\mathrm{SR}}(\beta)$.

\subsection{Part (ii)} 
It is well known that for a fixed prior, mutual information is convex in the channel law. 
When $Z= \mathcal{M}(X) + B$, the "channel law" in our setting of the deterministic mechanism is determined by the perturbation rule $Q$. Thus, the mapping $Q\mapsto \mathtt{MI}(Q)\equiv\mathtt{MI}(X;\mathcal{M}(X)+B$ is convex.
The objective $\mathcal{K}(Q) = \mathbb{E}_{Q}[\|B\|^{2}_{2}]$ is linear (hence convex) in $Q$.
In addition, the constraint set $\{Q: \mathtt{MI}(Q)\leq \beta\}$ is convex.
Then, Slater's condition holds because: 
\begin{itemize}
    \item[(i)]  when $\Sigma_{B} \rightarrow \infty$, $\mathtt{MI}(Q)\rightarrow 0<\beta$; 

    \item[(ii)] $V(\beta)$ is finite for all $\beta>0$ since $\mathbb{E}[\|\mathcal{M}(X)\|^2_{2}]<\infty$.
\end{itemize}
Hence, the strong duality applies here. 
Thus, $V(\beta)$ is convex and differentiable. 
The primal-dual problem is formulated as
\[
\widehat{V}(\beta) = \inf_{Q}\;\max_{\lambda}\mathcal{K}(Q) + \lambda(\mathtt{MI}(Q) -\beta).
\]
The envelop theorem implies $\widehat{V}'(\beta) = \lambda^*(\beta)>0$, where $\lambda^*(\beta)$ is the unique optimal dual variable (because $\mathcal{K}(Q) + \lambda(\mathtt{MI}(Q) -\beta)$ is strict convex in $Q$ for $\lambda>0$).
Therefore, $\lambda^*(\beta)$ is non-decreasing.

Let $\widetilde{\beta}(\beta) = \beta - \mathtt{Gap}_{d}(Q(\beta))<\beta$.
Since the Gaussian noise $B_{\mathrm{PAC}}(\beta)$ satisfies $\mathtt{MI}_{\mathtt{PAC}}(B_{\mathrm{pAC}}(\beta)) = \widetilde{\beta}(\beta)$, we have
\[
V_{\mathrm{PAC}}(\beta) = \mathcal{K}(B_{\mathrm{PAC}}(\beta))\geq V(\widetilde{\beta}(\beta)).
\]
Since $\widetilde{\beta}(\beta)<\beta$ and $V$ is strictly increasing, $V(\widetilde{\beta}(\beta))> V(\beta)$.
Therefore, for all $\beta>0$,
\[
\Delta(\beta) \equiv V_{\mathrm{PAC}}(\beta)  - V_{\mathrm{SR}}(\beta)>0,
\]
and $\lim_{\beta\rightarrow 0^+} \Delta(\beta) = 0$.

By Lemma \ref{lem:slopegap1} (stated and proved below) to $g(\beta) = V_{\mathrm{PAC}}(\beta)$ and $f(\beta) = V(\beta)$, we have $g'(\beta)>f'(\beta)$ for all $\beta>0$. That is, $V'_{\mathrm{PAC}}(\beta)>V'_{\mathrm{SR}}(\beta)$.
Thus, $\mathtt{Util}^{\mathrm{SR}}_{\beta}\geq \mathtt{Util}^{\mathrm{PAC}}_{\beta}$, with equality only for Gaussian $\mathcal{M}(X)$.

\qed

\begin{lemma}[Height gap $ \Rightarrow $ slope gap]\label{lem:slopegap1}
Let $g,f:(0,\infty) \to \mathbb R$ be differentiable,
and assume $f$ is convex.
If $g(\beta)>f(\beta)$ for every $\beta>0$ and $g(0)=f(0)$,
then $g'(\beta)>f'(\beta)$ for every $\beta>0$.
\end{lemma}

\begin{proof}
Fix $\beta>0$.  For $h>0$ small,
$f(\beta+h)\ge f(\beta)+hf'(\beta)$ by convexity.
Hence
\[
  \frac{g(\beta+h)-g(\beta)}{h}
  \;\ge\;f'(\beta)+\frac{g(\beta)-f(\beta)}{h}.
\]
Sending $h\downarrow0$ gives $g'(\beta)\ge f'(\beta)$.
If equality held we would need $g(\beta)=f(\beta)$, contradicting the
strict height gap.  Hence $g'(\beta)>f'(\beta)$.
\end{proof}

\section{More on Experiments}\label{app:more_on_experiments}

\subsection{More On The results of Fig. \ref{fig:DP_comparisons}}

Fig. \ref{fig:DPvs_SR_PAC} zooms in on the SR-PAC curves from Fig. \ref{fig:DP_comparisons} and shows a clear monotone increase in expected noise power $\mathbb{E}\|B\|^2=\operatorname{tr}(\Sigma)$ as $\beta$ decreases. This is consistent with the tighter privacy requirement $H_c \ge H_M-\beta$ pushing the mechanism to add more noise in the high-privacy (small-$\beta$) regime.

Fig. \ref{fig:DP_target} further shows that SR-PAC implements the privacy constraint \textit{conservatively}: the achieved conditional entropy $H_c$ typically lies at or slightly above the target line $\mathcal{H}(\mathcal{X})-\beta$. Equivalently, the effective mutual information $\mathtt{MI}(M;Y)=\mathcal{H}(\mathcal{X})-H_c$ is at or below the nominal budget $\beta$; i.e., the realized mechanism is (slightly) \textit{more private than necessary}. 
This benign overshoot is expected from finite-sample estimation and our fixed-CRN calibration with a positive tolerance, which is designed to avoid budget violations. If desired, the conservatism can be reduced by tightening the calibration tolerance, enlarging the CRN bank, or applying a final back-off on the noise scale until $H_c$ falls within a small band above the target. Importantly, even with this conservative bias, SR-PAC attains lower noise power and higher accuracy than Auto-PAC and Efficient-PAC at the same nominal $\beta$.

\begin{figure}[t] %
  \centering
  \begin{subfigure}{0.8\columnwidth}
    \includegraphics[width=\linewidth]{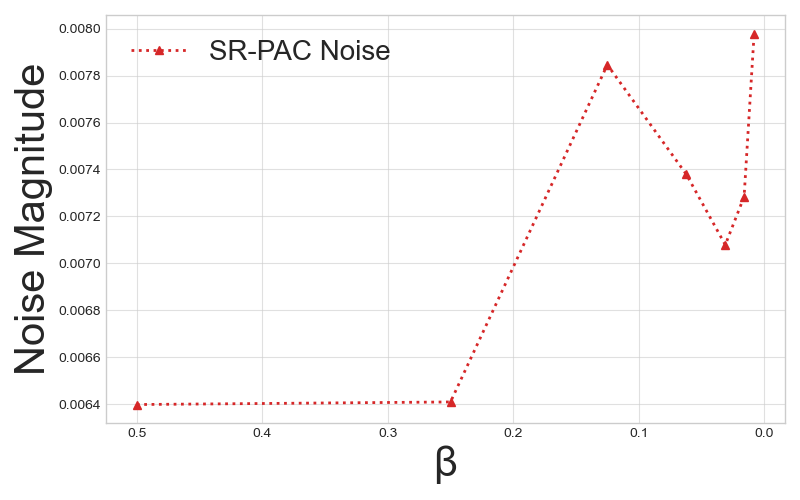}
    \caption{Iris}
  \end{subfigure}
  \hfill
  \begin{subfigure}{0.8\columnwidth}
    \includegraphics[width=\linewidth]{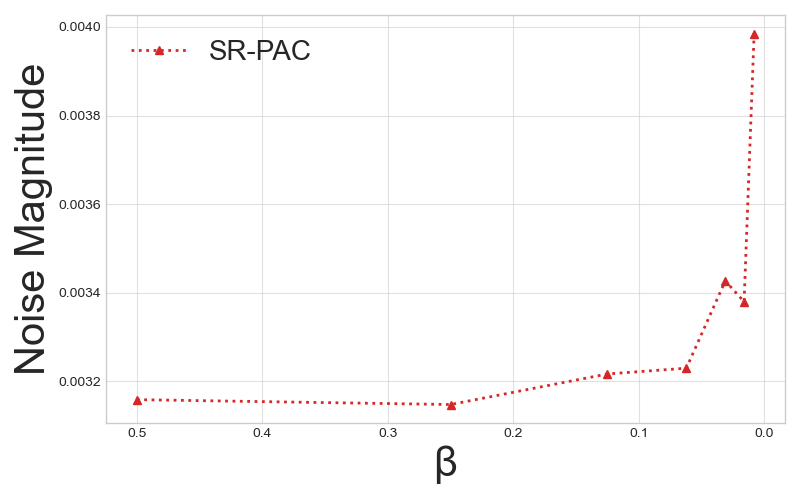}
    \caption{Rice}
  \end{subfigure}
  \caption{Noise magnitudes of SR-PAC of Fig. \ref{fig:DP_comparisons} \textbf{a} and \textbf{b}. All the numerical values are shown in Tables \ref{tab:Iris_value} and \ref{tab:Rice_value}.}
  \label{fig:DPvs_SR_PAC}
\end{figure}

\begin{figure}[t] %
  \centering
  \begin{subfigure}{0.8\columnwidth}
    \includegraphics[width=\linewidth]{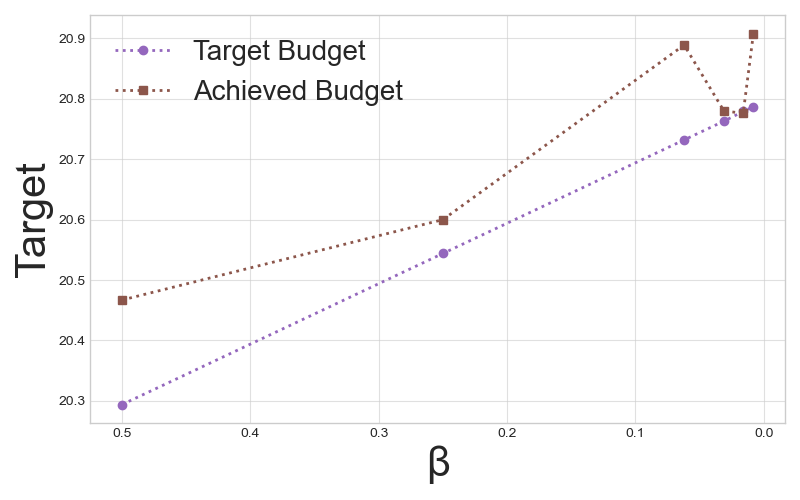}
    \caption{Iris}
  \end{subfigure}
  \hfill
  \begin{subfigure}{0.8\columnwidth}
    \includegraphics[width=\linewidth]{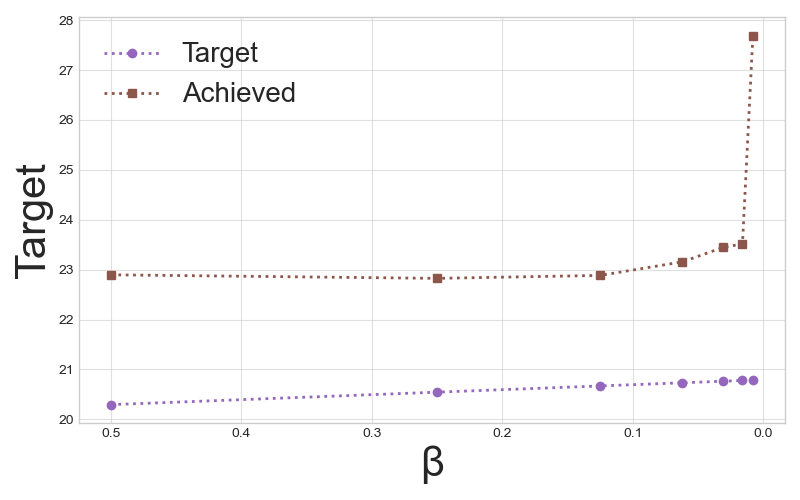}
    \caption{Rice}
  \end{subfigure}
  \caption{SR-PAC's performance of implementing the target privacy budget of Fig. \ref{fig:DP_comparisons} \textbf{a} and \textbf{b}. All the numerical values are shown in Tables \ref{tab:Iris_value} and \ref{tab:Rice_value}.}
  \label{fig:DP_target}
\end{figure}

\begin{figure}[t] %
  \centering
  \begin{subfigure}{0.8\columnwidth}
    \includegraphics[width=\linewidth]{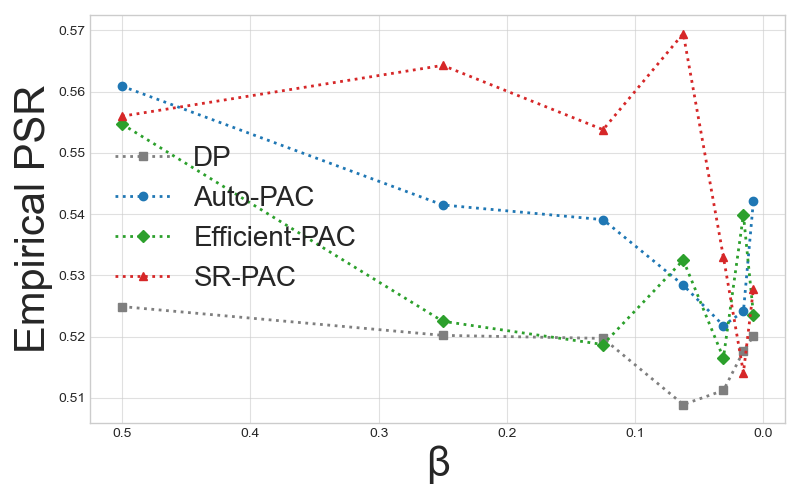}
    \caption{Iris}
  \end{subfigure}
  \hfill
  \begin{subfigure}{0.8\columnwidth}
    \includegraphics[width=\linewidth]{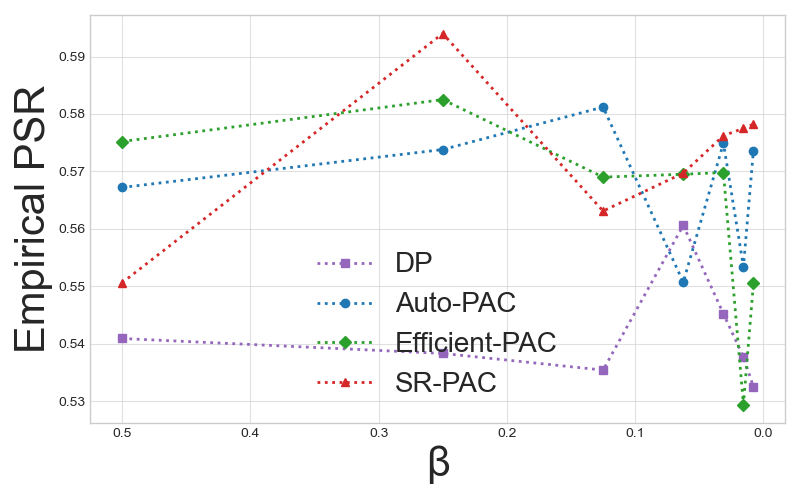}
    \caption{Rice}
  \end{subfigure}
  \caption{The performance of empirical membership inference attack using empritical LIRA, measured by the empirical posterior success rate (PSR). All the numerical values are shown in Tables \ref{tab:Iris_value} and \ref{tab:Rice_value}.}
  \label{fig:DP_psr}
\end{figure}

\subsection{Empirical Membership Inference Attack}

We use the Likelihood-Ratio Attack (LIRA) described in \cite{carlini2022membership} to perform the empirical membership inference attacks (MIAs) on the mechanisms privatized by SR-PAC, Auto-PAC, Efficient-PAC, and DP in Section \ref{S:expdp}, using Iris and Rice datasets.
The empirical posterior success rate (PSR) is measured as the average accuracy of the MIA.

A theoretical ordering of privacy budgets in mutual information or conditional entropy does not in general imply an ordering of a membership–inference attack's PSR. Finite–sample effects, non–optimal attacks, calibration error, and run–to–run variance can all break that implication. We therefore read the PSR curves in Fig.~\ref{fig:DP_psr} empirically, as a \textit{diagnostic} rather than a ground–truth ranking of privacy strength.

As $\beta$ decreases (higher privacy), all mechanisms tend to push PSR toward $0.5$ (chance), but the trends are not strictly monotone and the rankings cross. 
This is expected.
In particular, DP is generally conservative and can display lower PSR at very small $\beta$, while PAC–based methods may sit closer to chance but with visible fluctuations. Auto-PAC and Efficient-PAC allocate anisotropic noise from second–order structure (covariance scaling or eigen–allocation); in small–sample regimes (e.g, Iris and Rice), those moment estimates are noisy or ill–conditioned, so the resulting implementations are unstable and can become conservative (over–noisy), which may depress PSR at a fixed $\beta$. 

Our goal with SR-PAC is privacy budget fidelity (i.e., to address the conservativeness of Auto-PAC and Efficient-PAC privatization), not to minimize PSR per se.
SR-PAC enforces the conditional–entropy target directly and typically attains the desired leakage with less noise than Auto-PAC and Efficient-PAC. 
Hence it is plausible—and observed in Fig.~\ref{fig:DP_psr}—that SR-PAC's PSR can be comparable to, or occasionally above, over–noised baselines at the same $\beta$. 
The take–away is that PSR complements our main metric: SR-PAC achieves tighter budget implementation with lower noise power, while differences in PSR reflect both budget alignment and attack/model mismatch rather than a strict ordering induced by $\beta$.

\subsection{Discussion On The results of Fig. \ref{fig:all_metrics}}

Fig.~\ref{fig:all_metrics} empirically demonstrates that SR-PAC enforces the mutual-information budget more tightly than Auto-PAC and Efficient-PAC, yielding higher accuracy with lower noise. The performance gaps (both accuracy and noise magnitude) increase as $\beta$ decreases (i.e., in higher-privacy regimes). This behavior can be understood as follows.

\textit{Budget alignment vs. Gaussian surrogate.}
SR-PAC enforces the privacy constraint directly in terms of conditional entropy, aligning the mutual information budget with the actual leakage bound. In contrast, Auto-PAC and Efficient-PAC rely on Gaussian surrogates that ignore the higher-order, non-Gaussian structure of the outputs, leading to conservative privacy budget implementation. This conservativeness can lead them to add more noise than necessary to meet a given $\beta$, and the inefficiency becomes a larger fraction of the total budget when $\beta$ is small, amplifying their disadvantage in high-privacy regimes.

\textit{Directional selectivity under tight budgets.}
SR-PAC learns an anisotropic, task-directed noise shape that concentrates perturbations away from task-critical directions (Theorem \ref{prop:directional}), thereby preserving classification margins while still raising conditional entropy to the desired level. By contrast, Auto-PAC and Efficient-PAC are also anisotropic but task-agnostic: Auto-PAC scales the raw logit covariance ($\Sigma_B \propto \Sigma_{\text{raw}}$), and Efficient-PAC allocates along the eigenbasis with $e_i \propto \sqrt{\lambda_i}$. Both rely only on second-order statistics and ignore label-conditioned and higher-order structure, so they can spend budget along decision-sensitive axes whenever those coincide with high-variance directions. This mismatch is particularly costly under tight privacy budgets (small $\beta$), where misallocated power yields larger accuracy loss for the same leakage target.

\textit{Non-Gaussian exploitation.}
As $\beta$ decreases, the non-Gaussian structure of the outputs matters more. SR-PAC can use less total noise by optimally exploiting the geometries of the outputs (e.g., via leveraging flexible posteriors and calibration with fixed common random numbers).
Gaussian surrogates cannot capture this effect, so their "privacy per unit noise" degrades as $\beta$ shrinks.

\textit{Utility sensitivity.}
Viewing the required noise power as a utility curve $\mathtt{Util}^{\mathrm{SR}}_{\beta}$ (Theorem \ref{thm:speed_general}), SR-PAC exhibits better sensitivity (i.e., larger accuracy retention per unit budget). As $\beta$ decreases, the noise power of Auto-PAC and Efficient-PAC baselines grows faster than that of SR-PAC, widening the gap in both accuracy and magnitude.
When $\beta$ is large (loose privacy), all methods add little noise and their performance converges. As $\beta$ decreases (stricter privacy), SR-PAC optimally places noise where it is least harmful, yielding progressively larger gains over Auto-PAC and Efficient-PAC in both accuracy and noise efficiency.

\begin{table*}[t]
\centering
\caption{CIFAR-10 dataset results. Accuracy: higher is better. Noise Magnitude $\mathbb{E}[\lVert B\rVert^2_2]$: lower is better. Target Match: smaller $\Delta$/Rel.\% is better. All results are averaged over 35 trials.}
\label{tab:values_C10}
\setlength{\tabcolsep}{3.5pt}

\begin{subtable}[t]{0.32\textwidth}
\centering
\caption{Accuracy (Fig. \ref{fig:c10_acc})}
\label{tab:acc_case2}
\scriptsize
\begin{tabular}{@{}r S[table-format=1.2] S[table-format=1.2] S[table-format=1.2]@{}}
\toprule
{$\beta$} & {\textbf{Auto-PAC}} & {\textbf{Efficient-PAC}} & {\textbf{SR-PAC}} \\
\midrule
10.00 & 64.3\% & 62.8\% & \textbf{67.7\%} \\
9.50 & 63.5\% & 61.8\% & \textbf{67.2\%} \\
9.00 & 62.8\% & 60.8\% & \textbf{66.7\%} \\
8.50 & 61.8\% & 59.9\% & \textbf{66.0\%} \\
8.00 & 61.0\% & 58.8\% & \textbf{65.2\%} \\
7.50 & 59.7\% & 57.8\% & \textbf{64.3\%} \\
7.00 & 58.5\% & 56.7\% & \textbf{63.5\%} \\
6.50 & 57.0\% & 55.5\% & \textbf{62.6\%} \\
6.00 & 55.5\% & 54.1\% & \textbf{61.3\%} \\
5.50 & 53.9\% & 52.5\% & \textbf{60.0\%} \\
5.00 & 52.1\% & 50.9\% & \textbf{58.7\%} \\
4.50 & 50.0\% & 49.1\% & \textbf{57.0\%} \\
4.00 & 48.0\% & 46.8\% & \textbf{55.0\%} \\
3.50 & 45.6\% & 44.2\% & \textbf{53.1\%} \\
3.00 & 42.6\% & 41.8\% & \textbf{50.9\%} \\
2.50 & 39.3\% & 38.6\% & \textbf{48.2\%} \\
2.00 & 35.8\% & 35.2\% & \textbf{45.1\%} \\
1.50 & 31.5\% & 31.3\% & \textbf{41.4\%} \\
1.00 & 26.8\% & 26.8\% & \textbf{36.4\%} \\
0.50 & 20.8\% & 20.9\% & \textbf{30.9\%} \\
0.25 & 17.0\% & 17.1\% & \textbf{27.5\%} \\
\bottomrule
\end{tabular}
\end{subtable}
\hfill
\begin{subtable}[t]{0.32\textwidth}
\centering
\caption{Noise Magnitude (Fig. \ref{fig:c10_noise})}
\label{tab:noise_case2}
\scriptsize
\begin{tabular}{@{}r S[table-format=4.0] S[table-format=4.0] S[table-format=4.0]@{}}
\toprule
{$\beta$} & {\textbf{Auto-PAC}} & {\textbf{Efficient-PAC}} & {\textbf{SR-PAC}} \\
\midrule
10.00 & 203.614 & 194.129 & \textbf{86.717} \\
9.50 & 228.794 & 218.269 & \textbf{101.123} \\
9.00 & 257.622 & 245.929 & \textbf{113.220} \\
8.50 & 290.772 & 277.767 & \textbf{127.892} \\
8.00 & 329.089 & 314.602 & \textbf{144.288} \\
7.50 & 373.640 & 357.471 & \textbf{162.920} \\
7.00 & 425.798 & 407.708 & \textbf{184.447} \\
6.50 & 487.357 & 467.057 & \textbf{209.897} \\
6.00 & 560.704 & 537.842 & \textbf{238.969} \\
5.50 & 649.098 & 623.229 & \textbf{272.759} \\
5.00 & 757.093 & 727.649 & \textbf{311.272} \\
4.50 & 891.269 & 857.500 & \textbf{357.395} \\
4.00 & 1061.490 & 1022.375 & \textbf{409.579} \\
3.50 & 1283.251 & 1237.343 & \textbf{474.419} \\
3.00 & 1582.374 & 1527.514 & \textbf{553.758} \\
2.50 & 2005.329 & 1938.079 & \textbf{652.985} \\
2.00 & 2645.047 & 2559.404 & \textbf{781.609} \\
1.50 & 3718.355 & 3602.333 & \textbf{979.469} \\
1.00 & 5875.714 & 5699.364 & \textbf{1290.274} \\
0.50 & 12369.384 & 12012.950 & \textbf{1844.353} \\
0.25 & 25372.957 & 24657.045 & \textbf{2327.319} \\
\bottomrule
\end{tabular}
\end{subtable}
\hfill
\begin{subtable}[t]{0.32\textwidth}
\centering
\caption{Target Match (Fig. \ref{fig:c10_match})}
\label{tab:target_case2}
\scriptsize
\begin{tabular}{@{}r S[table-format=2.4] S[table-format=2.4] S[table-format=1.6, round-mode=places, round-precision=4] S[table-format=2.2]@{}}
\toprule
{$\beta$} & {Target} & {\textbf{Achieved}} & {$\Delta$} & {Rel.\%} \\
\midrule
10.00 & 3.0000 & \textbf{2.9762} & 0.0238 & 0.79 \\
9.50 & 3.5000 & \textbf{3.4901} & 0.0099 & 0.28 \\
9.00 & 4.0000 & \textbf{3.9872} & 0.0128 & 0.32 \\
8.50 & 4.5000 & \textbf{4.4889} & 0.0111 & 0.25 \\
8.00 & 5.0000 & \textbf{4.9873} & 0.0127 & 0.25 \\
7.50 & 5.5000 & \textbf{5.4868} & 0.0132 & 0.24 \\
7.00 & 6.0000 & \textbf{5.9853} & 0.0147 & 0.24 \\
6.50 & 6.5000 & \textbf{6.4895} & 0.0105 & 0.16 \\
6.00 & 7.0000 & \textbf{6.9897} & 0.0103 & 0.15 \\
5.50 & 7.5000 & \textbf{7.4881} & 0.0119 & 0.16 \\
5.00 & 8.0000 & \textbf{7.9863} & 0.0137 & 0.17 \\
4.50 & 8.5000 & \textbf{8.4852} & 0.0148 & 0.17 \\
4.00 & 9.0000 & \textbf{8.9867} & 0.0133 & 0.15 \\
3.50 & 9.5000 & \textbf{9.4868} & 0.0132 & 0.14 \\
3.00 & 10.0000 & \textbf{9.9873} & 0.0127 & 0.13 \\
2.50 & 10.5000 & \textbf{10.4908} & 0.0092 & 0.09 \\
2.00 & 11.0000 & \textbf{10.9896} & 0.0104 & 0.09 \\
1.50 & 11.5000 & \textbf{11.4915} & 0.0085 & 0.07 \\
1.00 & 12.0000 & \textbf{11.9920} & 0.0080 & 0.07 \\
0.50 & 12.5000 & \textbf{12.4939} & 0.0061 & 0.05 \\
0.25 & 12.7500 & \textbf{12.7420} & 0.0080 & 0.06 \\
\bottomrule
\end{tabular}
\end{subtable}

\end{table*}

\begin{table*}[t]
\centering
\caption{CIFAR-100 dataset results. Accuracy: higher is better. Noise Magnitude $\mathbb{E}[\lVert B\rVert^2_2]$: lower is better. Target Match: smaller $\Delta$/Rel.\% is better. All results are averaged over 35 trials.}
\label{tab:values_C100}
\setlength{\tabcolsep}{3.5pt}

\begin{subtable}[t]{0.32\textwidth}
\centering
\caption{Accuracy (Fig. \ref{fig:c100_acc})}
\label{tab:acc_corr}
\scriptsize
\begin{tabular}{@{}r S[table-format=2.1] S[table-format=2.1] S[table-format=2.1]@{}}
\toprule
{$\beta$} & {\textbf{Auto-PAC}} & {\textbf{Efficient-PAC}} & {\textbf{SR-PAC}} \\
\midrule
80.00 & 55.7\% & 56.4\% & \textbf{59.1\%} \\
75.00 & 55.2\% & 56.2\% & \textbf{59.3\%} \\
70.00 & 54.5\% & 56.0\% & \textbf{58.4\%} \\
65.00 & 53.8\% & 55.7\% & \textbf{57.8\%} \\
60.00 & 52.8\% & 55.4\% & \textbf{57.2\%} \\
55.00 & 51.8\% & 55.1\% & \textbf{56.2\%} \\
50.00 & 50.4\% & 54.6\% & \textbf{55.7\%} \\
45.00 & 48.9\% & 54.1\% & \textbf{55.8\%} \\
40.00 & 46.9\% & 53.7\% & \textbf{54.8\%} \\
35.00 & 44.4\% & 52.9\% & \textbf{53.4\%} \\
30.00 & 41.3\% & 51.6\% & \textbf{52.1\%} \\
25.00 & 37.6\% & 50.1\% & \textbf{50.3\%} \\
20.00 & 33.3\% & 48.3\% & \textbf{49.2\%} \\
15.00 & 27.5\% & 45.4\% & \textbf{48.0\%} \\
10.00 & 20.4\% & 39.8\% & \textbf{46.6\%} \\
7.00  & 15.2\% & 34.2\% & \textbf{46.3\%} \\
5.00  & 11.7\% & 28.7\% & \textbf{45.5\%} \\
4.00  & 9.6\%  & 24.7\% & \textbf{45.0\%} \\
3.00  & 7.5\%  & 20.1\% & \textbf{45.1\%} \\
2.00  & 5.5\%  & 15.0\% & \textbf{44.8\%} \\
1.00  & 3.3\%  & 8.7\%  & \textbf{45.0\%} \\
0.50  & 2.3\%  & 5.1\%  & \textbf{45.1\%} \\
\bottomrule
\end{tabular}
\end{subtable}
\hfill
\begin{subtable}[t]{0.32\textwidth}
\centering
\caption{Noise Magnitude (Fig. \ref{fig:c10_noise})}
\label{tab:noise_corr}
\scriptsize
\begin{tabular}{@{}r S[table-format=6.3] S[table-format=6.3] S[table-format=6.3]@{}}
\toprule
{$\beta$} & {\textbf{Auto-PAC}} & {\textbf{Efficient-PAC}} & {\textbf{SR-PAC}} \\
\midrule
80.00 & 295.712 & 156.079 & \textbf{0.015} \\
75.00 & 335.745 & 166.485 & \textbf{9.210} \\
70.00 & 382.613 & 178.377 & \textbf{73.285} \\
65.00 & 437.928 & 192.098 & \textbf{121.950} \\
60.00 & 503.837 & 208.106 & \textbf{173.816} \\
55.00 & 583.265 & 227.025 & \textbf{217.047} \\
50.00 & 680.307 & 249.727 & \textbf{256.837} \\
45.00 & 800.875 & 277.474 & \textbf{267.036} \\
40.00 & 953.831 & 312.159 & \textbf{318.676} \\
35.00 & 1153.101 & 356.753 & \textbf{409.873} \\
30.00 & 1421.886 & 416.212 & \textbf{480.533} \\
25.00 & 1801.943 & 499.454 & \textbf{659.901} \\
20.00 & 2376.780 & 624.318 & \textbf{709.042} \\
15.00 & 3341.231 & 832.423 & \textbf{799.096} \\
10.00 & 5279.782 & 1248.635 & \textbf{927.919} \\
7.00  & 7778.857 & 1783.764 & \textbf{955.096} \\
5.00  & 11114.847 & 2497.270 & \textbf{1060.981} \\
4.00  & 14035.296 & 3121.587 & \textbf{1027.777} \\
3.00  & 18904.016 & 4162.116 & \textbf{1057.745} \\
2.00  & 28643.398 & 6243.175 & \textbf{1114.597} \\
1.00  & 57865.395 & 12486.350 & \textbf{1059.941} \\
0.50 & 116312.383 & 24972.699 & \textbf{1062.913} \\
\bottomrule
\end{tabular}
\end{subtable}
\hfill
\begin{subtable}[t]{0.32\textwidth}
\centering
\caption{Target Match (Fig. \ref{fig:c10_match})}
\label{tab:target_corr}
\scriptsize
\begin{tabular}{@{}r S[table-format=3.4] S[table-format=3.4] S[table-format=1.6, round-mode=places, round-precision=4] S[table-format=2.2]@{}}
\toprule
{$\beta$} & {Target} & {\textbf{Achieved}} & {$\Delta$} & {Rel.\%} \\
\midrule
80.00 & 30.0000  & \textbf{32.9105} & 2.9105  & 9.70 \\
75.00 & 35.0000  & \textbf{35.0027} & 0.0027  & 0.01 \\
70.00 & 40.0000  & \textbf{40.0269} & 0.0269  & 0.07 \\
65.00 & 45.0000  & \textbf{45.0010} & 0.0010  & 0.00 \\
60.00 & 50.0000  & \textbf{50.0244} & 0.0244  & 0.05 \\
55.00 & 55.0000  & \textbf{55.0116} & 0.0116  & 0.02 \\
50.00 & 60.0000  & \textbf{60.0306} & 0.0306  & 0.05 \\
45.00 & 65.0000  & \textbf{65.0453} & 0.0453  & 0.07 \\
40.00 & 70.0000  & \textbf{70.0253} & 0.0253  & 0.04 \\
35.00 & 75.0000  & \textbf{75.0114} & 0.0114  & 0.02 \\
30.00 & 80.0000  & \textbf{80.0208} & 0.0208  & 0.03 \\
25.00 & 85.0000  & \textbf{85.0612} & 0.0612  & 0.07 \\
20.00 & 90.0000  & \textbf{90.0738} & 0.0738  & 0.08 \\
15.00 & 95.0000  & \textbf{95.0488} & 0.0488  & 0.05 \\
10.00 & 100.0000 & \textbf{100.0729} & 0.0729  & 0.07 \\
7.00  & 103.0000 & \textbf{103.0800} & 0.0800  & 0.08 \\
5.00  & 105.0000 & \textbf{105.0304} & 0.0304  & 0.03 \\
4.00  & 106.0000 & \textbf{106.1096} & 0.1096  & 0.10 \\
3.00  & 107.0000 & \textbf{107.0857} & 0.0857  & 0.08 \\
2.00  & 108.0000 & \textbf{108.0618} & 0.0618  & 0.06 \\
1.00  & 109.0000 & \textbf{109.0935} & 0.0935  & 0.09 \\
0.50 & 109.5000 & \textbf{109.5334} & 0.0334  & 0.03 \\
\bottomrule
\end{tabular}
\end{subtable}

\end{table*}

\begin{table*}[t]
\centering
\caption{MNIST dataset. Accuracy: higher is better. Noise Magnitude $\mathbb{E}[\lVert B\rVert^2_2]$: lower is better. Target Match: smaller $\Delta$/Rel.\% is better. All results are averaged over 35 trials.}
\label{tab:values_MNIST}
\setlength{\tabcolsep}{3.5pt} 

\begin{subtable}[t]{0.32\textwidth}
\centering
\caption{Accuracy (Fig. \ref{fig:mnist_acc})}
\label{tab:acc_beta_le7}
\scriptsize
\begin{tabular}{@{}r S[table-format=2.1] S[table-format=2.1] S[table-format=2.1]@{}}
\toprule
{$\beta$} & {\textbf{Auto-PAC}} & {\textbf{Efficient-PAC}} & {\textbf{SR-PAC}} \\
\midrule
7.00 & 92.3\% & 85.6\% & \textbf{98.4\%} \\
6.50 & 91.3\% & 84.3\% & \textbf{98.4\%} \\
6.00 & 89.4\% & 82.9\% & \textbf{98.4\%} \\
5.50 & 87.1\% & 81.0\% & \textbf{97.0\%} \\
5.00 & 84.4\% & 78.9\% & \textbf{95.4\%} \\
4.50 & 81.1\% & 76.5\% & \textbf{93.4\%} \\
4.00 & 77.5\% & 73.8\% & \textbf{91.3\%} \\
3.50 & 73.2\% & 70.4\% & \textbf{89.1\%} \\
3.00 & 68.0\% & 66.5\% & \textbf{87.1\%} \\
2.50 & 62.1\% & 62.0\% & \textbf{85.2\%} \\
2.00 & 55.3\% & 56.7\% & \textbf{83.6\%} \\
1.50 & 47.3\% & 50.2\% & \textbf{82.0\%} \\
1.00 & 38.5\% & 41.9\% & \textbf{80.5\%} \\
0.50 & 28.1\% & 30.6\% & \textbf{79.1\%} \\
0.25 & 21.0\% & 23.5\% & \textbf{78.3\%} \\
\bottomrule
\end{tabular}
\end{subtable}
\hfill
\begin{subtable}[t]{0.32\textwidth}
\centering
\caption{Noise Magnitude (Fig. \ref{fig:mnist_noise})}
\label{tab:noise_beta_le7}
\scriptsize
\begin{tabular}{@{}r S[table-format=4.0] S[table-format=4.0] S[table-format=4.0]@{}}
\toprule
{$\beta$} & {\textbf{Auto-PAC}} & {\textbf{Efficient-PAC}} & {\textbf{SR-PAC}} \\
\midrule
7.00 & 86.743 & 133.458 & \textbf{0.000} \\
6.50 & 99.284 & 143.724 & \textbf{0.000} \\
6.00 & 114.226 & 155.701 & \textbf{0.000} \\
5.50 & 132.234 & 169.855 & \textbf{28.292} \\
5.00 & 154.234 & 186.841 & \textbf{55.036} \\
4.50 & 181.569 & 207.601 & \textbf{77.870} \\
4.00 & 216.246 & 233.551 & \textbf{97.935} \\
3.50 & 261.423 & 266.915 & \textbf{116.500} \\
3.00 & 322.360 & 311.401 & \textbf{133.862} \\
2.50 & 408.524 & 373.682 & \textbf{150.256} \\
2.00 & 538.847 & 467.102 & \textbf{164.749} \\
1.50 & 757.501 & 622.803 & \textbf{179.350} \\
1.00 & 1196.996 & 934.204 & \textbf{194.086} \\
0.50 & 2519.882 & 1868.408 & \textbf{208.289} \\
0.25 & 5168.960 & 3736.815 & \textbf{215.495} \\
\bottomrule
\end{tabular}
\end{subtable}
\hfill
\begin{subtable}[t]{0.32\textwidth}
\centering
\caption{Target Match (Fig. \ref{fig:mnist_match})}
\label{tab:target_beta_le7}
\scriptsize
\begin{tabular}{@{}r S[table-format=2.4] S[table-format=2.4] S[table-format=1.6, round-mode=places, round-precision=4] S[table-format=2.2]@{}}
\toprule
{$\beta$} & {Target} & {\textbf{Achieved}} & {$\Delta$} & {Rel.\%} \\
\midrule
7.00 & 4.7478 & \textbf{5.8283} & 1.0805 & 22.76 \\
6.50 & 5.2478 & \textbf{5.8283} & 0.5805 & 11.06 \\
6.00 & 5.7478 & \textbf{5.8283} & 0.0805 & 1.40 \\
5.50 & 6.2478 & \textbf{6.2481} & 0.0003 & 0.00 \\
5.00 & 6.7478 & \textbf{6.7489} & 0.0011 & 0.02 \\
4.50 & 7.2478 & \textbf{7.2529} & 0.0051 & 0.07 \\
4.00 & 7.7478 & \textbf{7.7541} & 0.0063 & 0.08 \\
3.50 & 8.2478 & \textbf{8.2578} & 0.0100 & 0.12 \\
3.00 & 8.7478 & \textbf{8.7573} & 0.0095 & 0.11 \\
2.50 & 9.2478 & \textbf{9.2510} & 0.0032 & 0.03 \\
2.00 & 9.7478 & \textbf{9.7523} & 0.0045 & 0.05 \\
1.50 & 10.2478 & \textbf{10.2496} & 0.0018 & 0.02 \\
1.00 & 10.7478 & \textbf{10.7519} & 0.0041 & 0.04 \\
0.50 & 11.2478 & \textbf{11.2606} & 0.0128 & 0.11 \\
0.25 & 11.4978 & \textbf{11.5088} & 0.0110 & 0.10 \\
\bottomrule
\end{tabular}
\end{subtable}

\end{table*}

\begin{table*}[t]
\centering
\caption{Ag-news dataset results. Accuracy: higher is better. Noise Magnitude $\mathbb{E}[\lVert B\rVert^2_2]$: lower is better. Target Match: smaller $\Delta$/Rel.\% is better. All results are averaged over 35 trials.}
\label{tab:values_ag_news}
\setlength{\tabcolsep}{3.5pt}

\begin{subtable}[t]{0.32\textwidth}
\centering
\caption{Accuracy (Fig. \ref{fig:ag_acc})}
\label{tab:acc_filtered}
\scriptsize
\begin{tabular}{@{}r S[table-format=2.1] S[table-format=2.1] S[table-format=2.1]@{}}
\toprule
{$\beta$} & {Auto-PAC} & {Efficient-PAC} & {\textbf{SR-PAC}} \\
\midrule
9.50 & 96.7\% & 89.6\% & \textbf{96.6\%} \\
9.00 & 96.6\% & 89.3\% & \textbf{95.8\%} \\
8.50 & 96.5\% & 88.9\% & \textbf{95.1\%} \\
8.00 & 96.3\% & 88.4\% & \textbf{94.6\%} \\
7.50 & 96.0\% & 87.9\% & \textbf{93.5\%} \\
7.00 & 95.7\% & 87.3\% & \textbf{92.6\%} \\
6.50 & 95.3\% & 86.7\% & \textbf{92.3\%} \\
6.00 & 94.8\% & 85.9\% & \textbf{91.8\%} \\
5.50 & 94.1\% & 85.1\% & \textbf{91.6\%} \\
5.00 & 93.1\% & 84.2\% & \textbf{91.6\%} \\
4.50 & 91.9\% & 83.0\% & \textbf{89.3\%} \\
4.00 & 90.2\% & 81.7\% & \textbf{89.5\%} \\
3.50 & 88.1\% & 80.1\% & \textbf{89.5\%} \\
3.00 & 85.2\% & 78.1\% & \textbf{89.4\%} \\
2.50 & 81.5\% & 75.6\% & \textbf{87.1\%} \\
2.00 & 76.6\% & 72.4\% & \textbf{88.6\%} \\
1.50 & 70.2\% & 68.0\% & \textbf{86.2\%} \\
1.00 & 61.8\% & 61.8\% & \textbf{87.8\%} \\
0.50 & 50.2\% & 52.0\% & \textbf{87.0\%} \\
0.25 & 42.2\% & 44.2\% & \textbf{87.6\%} \\
0.08 & 35.4\% & 36.9\% & \textbf{85.1\%} \\
0.06 & 34.2\% & 35.6\% & \textbf{86.0\%} \\
0.02 & 33.6\% & 34.9\% & \textbf{86.8\%} \\
\bottomrule
\end{tabular}
\end{subtable}
\hfill
\begin{subtable}[t]{0.32\textwidth}
\centering
\caption{Noise Magnitude (Fig. \ref{fig:ag_noise})}
\label{tab:noise_filtered}
\scriptsize
\begin{tabular}{@{}r S[table-format=5.1] S[table-format=5.1] S[table-format=5.2]@{}}
\toprule
{$\beta$} & {Auto-PAC} & {Efficient-PAC} & {\textbf{SR-PAC}} \\
\midrule
9.50 & 1.2   & 20.7   & \textbf{0.54} \\
9.00 & 1.5   & 21.9   & \textbf{1.79} \\
8.50 & 2.0   & 23.2   & \textbf{3.07} \\
8.00 & 2.5   & 24.6   & \textbf{3.99} \\
7.50 & 3.3   & 26.2   & \textbf{6.05} \\
7.00 & 4.3   & 28.1   & \textbf{7.70} \\
6.50 & 5.5   & 30.3   & \textbf{8.70} \\
6.00 & 7.2   & 32.8   & \textbf{9.90} \\
5.50 & 9.3   & 35.8   & \textbf{10.57} \\
5.00 & 12.2  & 39.4   & \textbf{10.24} \\
4.50 & 16.1  & 43.7   & \textbf{15.72} \\
4.00 & 21.4  & 49.2   & \textbf{15.88} \\
3.50 & 28.7  & 56.2   & \textbf{16.05} \\
3.00 & 39.2  & 65.6   & \textbf{16.49} \\
2.50 & 54.8  & 78.7   & \textbf{23.51} \\
2.00 & 79.5  & 98.4   & \textbf{19.39} \\
1.50 & 122.3 & 131.2  & \textbf{25.96} \\
1.00 & 210.5 & 196.8  & \textbf{21.14} \\
0.50 & 480.8 & 393.6  & \textbf{23.16} \\
0.25 & 1025.7 & 787.3 & \textbf{21.77} \\
0.08 & 3346.4 & 2460.3 & \textbf{26.33} \\
0.06 & 4484.3 & 3280.3 & \textbf{27.45} \\
0.02 & 13588.6 & 9841.0 & \textbf{29.47} \\
\bottomrule
\end{tabular}
\end{subtable}
\hfill
\begin{subtable}[t]{0.32\textwidth}
\centering
\caption{Target Match (Fig. \ref{fig:ag_match})}
\label{tab:target_filtered}
\scriptsize
\begin{tabular}{@{}r S[table-format=2.4] S[table-format=2.4] S[table-format=1.6, round-mode=places, round-precision=4] S[table-format=2.2]@{}}
\toprule
{$\beta$} & {Target} & {\textbf{Achieved}} & {$\Delta$} & {Rel.\%} \\
\midrule
9.50 & 0.2109 & \textbf{0.2112} & 0.0003 & 0.14 \\
9.00 & 0.7109 & \textbf{0.7119} & 0.0010 & 0.14 \\
8.50 & 1.2109 & \textbf{1.2166} & 0.0057 & 0.47 \\
8.00 & 1.7109 & \textbf{1.7179} & 0.0070 & 0.41 \\
7.50 & 2.2109 & \textbf{2.2207} & 0.0098 & 0.44 \\
7.00 & 2.7109 & \textbf{2.7178} & 0.0069 & 0.25 \\
6.50 & 3.2109 & \textbf{3.2152} & 0.0043 & 0.13 \\
6.00 & 3.7109 & \textbf{3.7220} & 0.0111 & 0.30 \\
5.50 & 4.2109 & \textbf{4.2152} & 0.0043 & 0.10 \\
5.00 & 4.7109 & \textbf{4.7198} & 0.0089 & 0.19 \\
4.50 & 5.2109 & \textbf{5.2149} & 0.0040 & 0.08 \\
4.00 & 5.7109 & \textbf{5.7188} & 0.0079 & 0.14 \\
3.50 & 6.2109 & \textbf{6.2142} & 0.0033 & 0.05 \\
3.00 & 6.7109 & \textbf{6.7129} & 0.0020 & 0.03 \\
2.50 & 7.2109 & \textbf{7.2273} & 0.0164 & 0.23 \\
2.00 & 7.7109 & \textbf{7.7260} & 0.0151 & 0.20 \\
1.50 & 8.2109 & \textbf{8.2306} & 0.0197 & 0.24 \\
1.00 & 8.7109 & \textbf{8.7154} & 0.0045 & 0.05 \\
0.50 & 9.2109 & \textbf{9.2115} & 0.0006 & 0.01 \\
0.25 & 9.4609 & \textbf{9.4754} & 0.0145 & 0.15 \\
0.08 & 9.6309 & \textbf{9.6524} & 0.0215 & 0.22 \\
0.06 & 9.6509 & \textbf{9.6534} & 0.0025 & 0.03 \\
0.02 & 9.6909 & \textbf{9.6975} & 0.0066 & 0.07 \\
\bottomrule
\end{tabular}
\end{subtable}

\end{table*}

\begin{table*}[t]
\centering
\caption{Iris dataset results. Empirical Posterior Success Rate (PSR), Noise Magnitude, and Target Match. Empirical PSR: higher is better. Noise Magnitude: lower is better. Target Match: smaller $\Delta$/Rel.\% is better. All results are averaged over 35 trials.}
\label{tab:Iris_value}
\setlength{\tabcolsep}{3.5pt}

\begin{subtable}[t]{0.32\textwidth}
\centering
\caption{Empirical PSR}
\label{tab:psr_le_point5_case2}
\scriptsize
\begin{tabular}{@{}r S[table-format=1.4] S[table-format=1.4,round-mode=places, round-precision=4] S[table-format=1.4,round-mode=places, round-precision=4] S[table-format=1.4,round-mode=places, round-precision=4]@{}}
\toprule
{$\beta$} & {\textbf{SR-PAC}} & {\textbf{DP}} & {\textbf{Auto-PAC}} & {\textbf{Efficient-PAC}} \\
\midrule
0.5000 & \textbf{0.5560} & 0.5249 & 0.5609 & 0.5547 \\
0.2500 & \textbf{0.5643} & 0.5202 & 0.5415 & 0.5225 \\
0.1250 & \textbf{0.5538} & 0.5197 & 0.5391 & 0.5187 \\
0.0625 & \textbf{0.5695} & 0.5089 & 0.5284 & 0.5325 \\
0.0312 & \textbf{0.5330} & 0.5112 & 0.5218 & 0.5165 \\
0.0156 & \textbf{0.5141} & 0.5176 & 0.5241 & 0.5399 \\
0.0078 & \textbf{0.5278} & 0.5201 & 0.5422 & 0.5236 \\
\bottomrule
\end{tabular}
\end{subtable}
\hfill
\begin{subtable}[t]{0.32\textwidth}
\centering
\caption{Noise Magnitude}
\label{tab:noise_le_point5_case2}
\scriptsize
\begin{tabular}{@{}r S[table-format=1.6] S[table-format=1.6,round-mode=places, round-precision=4] S[table-format=1.6,round-mode=places, round-precision=4] S[table-format=1.6,round-mode=places, round-precision=4]@{}}
\toprule
{$\beta$} & {\textbf{SR-PAC}} & {\textbf{DP}} & {\textbf{Auto-PAC}} & {\textbf{Efficient-PAC}} \\
\midrule
0.5000 & \textbf{0.006399} & 0.004223 & 0.086523 & 0.025572 \\
0.2500 & \textbf{0.006410} & 0.008499 & 0.090899 & 0.026145 \\
0.1250 & \textbf{0.007846} & 0.013628 & 0.107344 & 0.032059 \\
0.0620 & \textbf{0.007379} & 0.020764 & 0.099071 & 0.028286 \\
0.0310 & \textbf{0.007079} & 0.031401 & 0.108686 & 0.039986 \\
0.0160 & \textbf{0.007284} & 0.048060 & 0.114727 & 0.037930 \\
0.0080 & \textbf{0.007979} & 0.075252 & 0.116107 & 0.038072 \\
\bottomrule
\end{tabular}
\end{subtable}
\hfill
\begin{subtable}[t]{0.32\textwidth}
\centering
\caption{Target Match (SR-PAC)}
\label{tab:target_le_point5_case2}
\scriptsize
\begin{tabular}{@{}r S[table-format=2.3] S[table-format=2.3] S[table-format=1.4,round-mode=places, round-precision=4] S[table-format=2.2]@{}}
\toprule
{$\beta$} & {Target $H_c$} & {\textbf{Achieved $H_c$}} & {$\Delta$} & {Rel.\%} \\
\midrule
0.5000 & 20.294 & \textbf{20.467} & 0.1730 & 0.85 \\
0.2500 & 20.544 & \textbf{20.600} & 0.0560 & 0.27 \\
0.1250 & 20.669 & \textbf{20.884} & 0.2150 & 1.04 \\
0.0620 & 20.732 & \textbf{20.889} & 0.1570 & 0.76 \\
0.0310 & 20.763 & \textbf{20.780} & 0.0170 & 0.08 \\
0.0160 & 20.779 & \textbf{20.776} & -0.0030 & -0.01 \\
0.0080 & 20.787 & \textbf{20.908} & 0.1210 & 0.58 \\
\bottomrule
\end{tabular}
\end{subtable}

\end{table*}

\begin{table*}[t]
\centering
\caption{Rice dataset results. Empirical Posterior Success Rate (PSR), Noise Magnitude, and Target Match. Empirical PSR: higher is better. Noise Magnitude: lower is better. Target Match: smaller $\Delta$/Rel.\% is better. All results are averaged over 35 trials.}
\label{tab:Rice_value}
\setlength{\tabcolsep}{3.5pt}

\begin{subtable}[t]{0.32\textwidth}
\centering
\caption{Empirical PSR}
\label{tab:psr_le_point5}
\scriptsize
\begin{tabular}{@{}r S[table-format=1.4] S[table-format=1.4,round-mode=places, round-precision=4] S[table-format=1.4,round-mode=places, round-precision=4] S[table-format=1.4,round-mode=places, round-precision=4]@{}}
\toprule
{$\beta$} & {\textbf{SR-PAC}} & {\textbf{DP}} & {\textbf{Auto-PAC}} & {\textbf{Efficient-PAC}} \\
\midrule
0.5000 & \textbf{0.5506} & 0.5409 & 0.5672 & 0.5752 \\
0.2500 & \textbf{0.5940} & 0.5383 & 0.5738 & 0.5825 \\
0.1250 & \textbf{0.5631} & 0.5354 & 0.5812 & 0.5690 \\
0.0625 & \textbf{0.5697} & 0.5606 & 0.5507 & 0.5695 \\
0.0312 & \textbf{0.5762} & 0.5452 & 0.5749 & 0.5698 \\
0.0156 & \textbf{0.5775} & 0.5377 & 0.5534 & 0.5294 \\
0.0078 & \textbf{0.5782} & 0.5325 & 0.5736 & 0.5505 \\
\bottomrule
\end{tabular}
\end{subtable}
\hfill
\begin{subtable}[t]{0.32\textwidth}
\centering
\caption{Noise Magnitude}
\label{tab:noise_le_point5}
\scriptsize
\begin{tabular}{@{}r S[table-format=1.6] S[table-format=1.6,round-mode=places, round-precision=4] S[table-format=1.6,round-mode=places, round-precision=4] S[table-format=1.6,round-mode=places, round-precision=4,round-mode=places, round-precision=4]@{}}
\toprule
{$\beta$} & {\textbf{SR-PAC}} & {\textbf{DP}} & {\textbf{Auto-PAC}} & {\textbf{Efficient-PAC}} \\
\midrule
0.5000 & \textbf{0.003159} & 0.001996 & 0.116784 & 0.011234 \\
0.2500 & \textbf{0.003148} & 0.002985 & 0.129103 & 0.011742 \\
0.1250 & \textbf{0.003217} & 0.003812 & 0.141698 & 0.016540 \\
0.0620 & \textbf{0.003230} & 0.004668 & 0.134833 & 0.011502 \\
0.0310 & \textbf{0.003427} & 0.005631 & 0.158874 & 0.022465 \\
0.0160 & \textbf{0.003380} & 0.006770 & 0.228656 & 0.047462 \\
0.0080 & \textbf{0.003985} & 0.008165 & 0.213591 & 0.090694 \\
\bottomrule
\end{tabular}
\end{subtable}
\hfill
\begin{subtable}[t]{0.32\textwidth}
\centering
\caption{Target Match (SR-PAC)}
\label{tab:target_le_point5}
\scriptsize
\begin{tabular}{@{}r S[table-format=2.3] S[table-format=2.3] S[table-format=1.4,round-mode=places, round-precision=4] S[table-format=2.2]@{}}
\toprule
{$\beta$} & {Target} & {\textbf{Achieved}} & {$\Delta$} & {Rel.\%} \\
\midrule
0.5000 & 20.294 & \textbf{22.896} & 2.6020 & 12.82 \\
0.2500 & 20.544 & \textbf{22.825} & 2.2810 & 11.10 \\
0.1250 & 20.669 & \textbf{22.884} & 2.2150 & 10.72 \\
0.0620 & 20.732 & \textbf{23.156} & 2.4240 & 11.69 \\
0.0310 & 20.763 & \textbf{23.444} & 2.6810 & 12.91 \\
0.0160 & 20.779 & \textbf{23.516} & 2.7370 & 13.17 \\
0.0080 & 20.787 & \textbf{27.687} & 6.9000 & 33.19 \\
\bottomrule
\end{tabular}
\end{subtable}

\end{table*}

\end{document}